%% file: main.tex
\documentclass[11pt]{article}

\usepackage{amsthm}
\usepackage{graphicx} 
\usepackage{array} 

\usepackage{amsmath, amssymb, amsfonts, verbatim}
\usepackage{hyphenat,epsfig,subcaption,multirow}
\usepackage[font=small,labelfont=bf, hypcap=true]{caption}
\usepackage{nicefrac}
\usepackage{paralist}

\usepackage{wrapfig}

\usepackage{tocloft}

\usepackage[usenames,dvipsnames]{xcolor}
\usepackage[ruled]{algorithm2e}

\DeclareFontFamily{U}{mathx}{\hyphenchar\font45}
\DeclareFontShape{U}{mathx}{m}{n}{
      <5> <6> <7> <8> <9> <10>
      <10.95> <12> <14.4> <17.28> <20.74> <24.88>
      mathx10
      }{}
\DeclareSymbolFont{mathx}{U}{mathx}{m}{n}
\DeclareMathSymbol{\bigtimes}{1}{mathx}{"91}

\usepackage{tcolorbox}
\tcbuselibrary{skins,breakable}
\tcbset{enhanced jigsaw}

\usepackage[normalem]{ulem}
\usepackage[compact]{titlesec}

\definecolor{DarkRed}{rgb}{0.5,0.1,0.1}
\definecolor{DarkBlue}{rgb}{0.1,0.1,0.5}

\usepackage{nameref}
\definecolor{ForestGreen}{rgb}{0.1333,0.5451,0.1333}
\definecolor{Red}{rgb}{0.9,0,0}
\usepackage[linktocpage=true,
	pagebackref=true,colorlinks,
	linkcolor=DarkRed,citecolor=ForestGreen,
	bookmarks,bookmarksopen,bookmarksnumbered]
	{hyperref}
\usepackage[noabbrev,nameinlink]{cleveref}
\crefname{property}{property}{Property}
\creflabelformat{property}{(#1)#2#3}
\crefname{equation}{eq}{Eq}
\creflabelformat{equation}{(#1)#2#3}

\usepackage{bm}
\usepackage{url}
\usepackage{xspace}
\usepackage[mathscr]{euscript}

\usepackage{tikz}
\usetikzlibrary{arrows}
\usetikzlibrary{arrows.meta}
\usetikzlibrary{shapes}
\usetikzlibrary{backgrounds}
\usetikzlibrary{positioning}
\usetikzlibrary{decorations.markings}
\usetikzlibrary{decorations.pathreplacing} 
\usetikzlibrary{patterns}
\usetikzlibrary{calc}
\usetikzlibrary{fit}
\usetikzlibrary{decorations}

\usepackage[framemethod=TikZ]{mdframed}

\usepackage[noend]{algpseudocode}
\makeatletter
\def\BState{\State\hskip-\ALG@thistlm}
\makeatother

\usepackage{cite}
\usepackage{enumitem}
\setlist[itemize]{leftmargin=20pt}
\setlist[enumerate]{leftmargin=20pt}

\usepackage[margin=1in]{geometry}

\usepackage{thmtools}
\usepackage{thm-restate}

\newtheorem{theorem}{Theorem}
\newtheorem{lemma}{Lemma}[section]
\newtheorem{proposition}[lemma]{Proposition}
\newtheorem{corollary}[lemma]{Corollary}
\newtheorem{claim}[lemma]{Claim}

\newtheorem{problem}{Problem}

\newtheorem*{corollary*}{Corollary}

\newtheorem*{claim*}{Claim}
\newtheorem*{assumption*}{Assumption}
\newtheorem*{proposition*}{Proposition}
\newtheorem*{lemma*}{Lemma}

\newtheorem{observation}[lemma]{Observation}

\newtheorem*{theorem*}{Theorem}

\crefname{lemma}{Lemma}{Lemmas}
\crefname{claim}{claim}{claims}
\crefname{property}{Property}{Properties}

\newtheorem{mdresult}{Result}
\newenvironment{result}{\begin{mdframed}[backgroundcolor=lightgray!40,topline=false,rightline=false,leftline=false,bottomline=false,innertopmargin=5pt]\begin{mdresult}}{\end{mdresult}\end{mdframed}}

\newtheorem{remark}[lemma]{Remark}

\theoremstyle{definition}

\newtheorem{invariant}{Invariant}
\crefname{invariant}{Invariant}{Invariants}
\newtheorem{flaw}{Flaw}

\newtheorem*{mdproblem*}{Problem}
\newenvironment{Problem*}{\begin{mdframed}[hidealllines=false,innerleftmargin=10pt,backgroundcolor=gray!10,innertopmargin=5pt,innerbottommargin=5pt,roundcorner=10pt]\begin{mdproblem*}}{\end{mdproblem*}\end{mdframed}}
\newtheorem{mddefinition}[lemma]{Definition}
\newenvironment{Definition}{\begin{mdframed}[hidealllines=false,innerleftmargin=10pt,backgroundcolor=white!10,innertopmargin=5pt,innerbottommargin=5pt,roundcorner=10pt]\begin{mddefinition}}{\end{mddefinition}\end{mdframed}}
\newtheorem*{mddefinition*}{Definition}
\newenvironment{Definition*}{\begin{mdframed}[hidealllines=false,innerleftmargin=10pt,backgroundcolor=white!10,innertopmargin=5pt,innerbottommargin=5pt,roundcorner=10pt]\begin{mddefinition*}}{\end{mddefinition*}\end{mdframed}}
\newtheorem{mdremark}{Remark}

\newenvironment{ourbox}{\begin{mdframed}[hidealllines=false,innerleftmargin=10pt,backgroundcolor=white!10,innertopmargin=10pt,innerbottommargin=10pt,roundcorner=10pt]}{\end{mdframed}}

\newenvironment{coloredbox}[1]{\begin{mdframed}[hidealllines=false,innerleftmargin=10pt,backgroundcolor=#1!10,innertopmargin=0pt,innerbottommargin=5pt]}{\end{mdframed}}

\newtcolorbox{redbox}[1]{
  colback=red!5,
  colframe=red!35,
  boxrule=0.5pt,
  arc=2pt,
  title=\textbf{#1},
  fonttitle=\normalsize,
  coltitle=black
}

\newtheorem{mdalgorithm}{Algorithm}
\newenvironment{Algorithm}{\begin{ourbox}\begin{mdalgorithm}}{\end{mdalgorithm}\end{ourbox}}

\allowdisplaybreaks

\newtheorem{mdassumption*}{Assumption}

\renewcommand{\qed}{\nobreak \ifvmode \relax \else
      \ifdim\lastskip<1.5em \hskip-\lastskip
      \hskip1.5em plus0em minus0.5em \fi \nobreak
      \vrule height0.75em width0.5em depth0.25em\fi}

\renewcommand{\leq}{\leqslant}
\renewcommand{\geq}{\geqslant}
\renewcommand{\le}{\leq}

\input{macros}

\title{Dynamic $(\Delta+1)$ Vertex Coloring with Adaptive Adversaries: \\ Entropy Compression to the Rescue\footnote{A preliminary version of this paper appears in the ACM-SIAM Symposium on Discrete Algorithms, SODA 2027.}}
\author{Sepehr Assadi\footnote{(sepehr@assadi.info) School of Computer Science, University of Waterloo. 
Supported in part by an NSERC
Discovery Grant and a Faculty of Math Research Chair grant. \smallskip} \\ {\small University of Waterloo} \and 
\and Helia Yazdanyar\footnote{(hyazdanyar@uwaterloo.ca) School of Computer Science, University of Waterloo. Supported in part by a Cheriton Scholarship from School of Computer Science, Cotton Family Women in Mathematics Graduate Scholarship, Barbara Hayes-Roth Award for Women in Math and Computer Science, and SA's NSERC Discovery Grant.} \\ {\small University of Waterloo}
}

\date{}

\begin{document}
\maketitle

\pagenumbering{roman}
	
\input{abstract}
	
\clearpage

\setcounter{tocdepth}{3}
\tableofcontents
	
\clearpage
	
\setcounter{page}{1}
\pagenumbering{arabic} 

\input{intro.tex}

\input{prelim.tex}

\input{technical-overview.tex}

\clearpage

\input{algorithm.tex}

\clearpage

\input{sparse.tex}

\clearpage

\input{sparse-subroutines.tex}

\clearpage

\input{dense.tex}

\paragraph{AI Disclosure:} No AI tools were used in any capacity in the research or preparation of this submission, except possibly for minor typo editing. 

\paragraph{Acknowledgement.} We are thankful to David Wajc for helpful discussions, and to Yuval Gil and Magnus Halldorsson for sharing their independent result in \cite{HGil2026} with us.

\bibliographystyle{alpha}
\bibliography{general}

\clearpage

\appendix

\input{missingProofs.tex}

\end{document}

%% file: macros.tex
\newcommand{\Ot}{\ensuremath{\widetilde{O}}}
\newcommand{\eps}{\ensuremath{\varepsilon}}
\newcommand{\Paren}[1]{\Big(#1\Big)}

\newcommand{\bracket}[1]{\left[#1\right]}
\newcommand{\paren}[1]{\ensuremath{\left(#1\right)}\xspace}
\newcommand{\card}[1]{\left\vert{#1}\right\vert}

\newcommand{\IN}{\ensuremath{\mathbb{N}}}

\newcommand{\abs}[1]{\ensuremath{\left|#1\right|}}

\newcommand{\floor}[1]{{\left\lfloor{#1}\right\rfloor}}
\newcommand{\prob}[1]{\Pr\paren{#1}}
\newcommand{\expect}[1]{\Exp\bracket{#1}}

\newcommand{\set}[1]{\ensuremath{\left\{ #1 \right\}}}
\newcommand{\poly}{\mbox{\rm poly}}

\DeclareMathOperator*{\Exp}{\ensuremath{{\mathbb{E}}}}
\DeclareMathOperator*{\Prob}{\ensuremath{\textnormal{Pr}}}
\renewcommand{\Pr}{\Prob}

\newenvironment{tbox}{\begin{tcolorbox}[
		enlarge top by=5pt,
		enlarge bottom by=5pt,
		 breakable,
		 boxsep=0pt,
                  left=4pt,
                  right=4pt,
                  top=10pt,
                  arc=0pt,
                  boxrule=1pt,toprule=1pt,
                  colback=white
                  ]
	}
{\end{tcolorbox}}

\newcommand{\II}{\ensuremath{\mathbb{I}}}

\newcommand{\mireal}[1][]{
  \ifx\relax#1\relax%
    \II(\mione \,; \mitwo)%
  \else%
    \II(\mione \,; \mitwo\mid #1)%
  \fi
}

\newcommand{\Decomposition}{\textnormal{\ensuremath{\algs{Decomposition}}}\xspace}
\newcommand{\ExpandMatching}{\textnormal{\ensuremath{\algs{ExpandAntiMatching}}}\xspace}
\newcommand{\ColorDense}{\textnormal{\ensuremath{\algs{ColorDenser}}}\xspace}
\newcommand{\ColorOH}{\textnormal{\ensuremath{\algs{PartialColorOH}}}\xspace}
\newcommand{\ColorSparse}{\textnormal{\ensuremath{\algs{AttemptColorSparser}}}\xspace}

\newcommand{\Cleanup}{\textnormal{\ensuremath{\algs{CleanUp}}}\xspace}

\newcommand{\UpdateSparseColor}{\textnormal{\ensuremath{\algs{UpdateSparserColor}}}\xspace}

\newcommand{\colS}{\varphi^{\star}}

\newcommand{\Fail}{\ensuremath{\textnormal{\emph{Fail}}}\xspace}
\newcommand{\Pass}{\ensuremath{\textnormal{\emph{Pass}}}\xspace}

\newcommand{\algs}[1]{\textnormal{\ensuremath{\texttt{#1}}}}
\newcommand{\Resample}{\ensuremath{\algs{Resample}}\xspace}

\newcommand{\Attempt}{\ensuremath{\algs{AttemptColor}}\xspace}

\newcommand{\randTable}{\ensuremath{A}\xspace}

%% file: abstract.tex

\begin{abstract}
	We present a randomized algorithm for maintaining a $(\Delta+1)$ vertex coloring of a fully dynamic graph against an adaptive adversary with $\Ot(\sqrt{n})$ amortized update time. 
	
	\medskip
	
	Our algorithm improves upon the recent randomized algorithms of Behnezhad, Rajaraman, and Wasim (SODA 2025) with $\Ot(n^{8/9})$ and Flin and Halld\'orsson (ICALP 2025) with $\Ot(n^{2/3})$ amortized update time.  
	Our result (nearly) matches a ``fundamental limit'' of $O(\sqrt{n})$ update time identified by Flin and Halld\'orsson and settles their main open question. 
	
	\medskip
	
	The previous algorithms' main improvement came from efficient coloring of the dense parts of the graph. Our approach complements these prior work by handling the sparse parts more efficiently, 
	using ideas inspired by the recently developed dynamic entropy compression method. A key new aspect of our work is to tailor this method specifically to dynamic algorithms---as opposed to original
	local search algorithms---and achieve update times that can be sublinear in the number of involved constraints per update.

\end{abstract}


%% file: intro.tex

\section{Introduction}\label{sec:intro} 

A $c$-coloring of a graph $G=(V,E)$ is an assignment of colors from $\set{1,\ldots,c}$ to vertices so no edge is monochromatic. 
One of the most basic graph coloring problems is that of finding a $(\Delta+1)$ coloring in graphs with maximum degree $\Delta$. Every graph admits such a coloring 
that can be found easily using a naive greedy algorithm. There is also a large body of work on $(\Delta+1)$ coloring across various computational models including distributed~\cite{Linial92,barenboim2015,HarrisSS16, ChangLP2020,ghaffariKuhn2021}, sublinear time and streaming~\cite{AssadiCK19,AssadiY25}, parallel~\cite{ChangFGUZ19, CzumajDP21,ghaffariKuhn2021, ghaffariK2025}, 
communication complexity~\cite{FlinM24, changMNS2025}, and dynamic graphs~\cite{BhattacharyaCHN18, HenzingerP22, BhattacharyaGKL22 ,BehnezhadRW25, FlinH25}. This last line of work, namely, dynamic algorithms for $(\Delta+1)$ coloring, is the focus of our paper. 

A (fully) dynamic algorithm for the $(\Delta+1)$ coloring problem is given a graph that undergoes edge insertions and deletions and its goal is to maintain a $(\Delta+1)$ coloring after each update quickly. 
It is easy to obtain a dynamic algorithm with $O(\Delta)$ update time: if an update makes an edge monochromatic, scan the neighbors of one of its endpoints and find a color missing from them to assign to the endpoint. 
This basic bound was significantly improved in the pioneering work of~\cite{BhattacharyaCHN18} to a \emph{randomized} algorithm with $O(\log{\Delta})$ update time, and was further improved
by~\cite{BhattacharyaGKL22,HenzingerP22} to an $O(1)$ update time algorithm. 

While remarkably fast, there is a caveat to the results in~\cite{BhattacharyaCHN18,BhattacharyaGKL22,HenzingerP22}: they only work against an \emph{oblivious} adversary meaning that the edge updates need to be fixed in advance and cannot
adapt to the algorithm's output. This is problematic for many applications of dynamic algorithms wherein the output of the algorithm can indeed
influence its future updates. As a result, there has been a significant interest in the literature in obtaining dynamic algorithms against \emph{adaptive} adversaries that provide the next update based on prior outputs of the algorithm; 
see, e.g.~\cite{BeimelKMNSS22,BernsteinBFKS26} and references therein for more on this general topic. 
For some time after~\cite{BhattacharyaCHN18}, the best known dynamic algorithm for $(\Delta+1)$ coloring against adaptive adversaries remained the naive $O(\Delta)$ update time deterministic algorithm, 
which can be as slow as $\Theta(n)$ update time. This is until the breakthrough of~\cite{BehnezhadRW25} who presented an $\Ot(n^{8/9})$ update time randomized
algorithm. This was then improved by~\cite{FlinH25} to an $\Ot(n^{2/3})$ update time randomized algorithm which constitutes the state-of-the-art. The authors in~\cite{FlinH25} identified two natural bounds
for the update time to target: (1) $O(\sqrt{n})$ as a ``fundamental limit'' and (2) $O(n^{2/3})$ as a ``natural barrier'' of existing techniques; we will discuss the rationale behind these bounds in~\Cref{sec:techniques}. The algorithm of~\cite{FlinH25} matches 
this second bound, leading the authors to ask: 
\begin{quote}
\!\!\cite[Problem 1]{FlinH25}:~\emph{``Does there exist a fully dynamic $(\Delta+1)$-coloring [algorithm] with $o(n^{2/3})$ update time that remains robust against adaptive adversaries?''}
\end{quote}
We settle this question in the affirmative in this paper.  

\begin{result}\label{res:main}
	There is a randomized algorithm that maintains a $(\Delta+1)$ coloring of any $n$-vertex fully dynamic graph with maximum degree $\Delta$ against an adaptive adversary in $\Ot(\sqrt{n})$ amortized update time with high probability. 
\end{result}

\Cref{res:main} answers the open question of~\cite{FlinH25} and improves the update time all the way to the fundamental limit identified in their work. 

It is also worth drawing a parallel between~\Cref{res:main} and the closely related maximal matching problem (see~\cite{BehnezhadRW25} for references to these connections). Maximal matchings also admit $O(1)$ update time dynamic
algorithms against oblivious adversaries~\cite{Solomon16} (see also~\cite{BaswanaGS11}). But, the first $o(n)$ update time algorithm against adaptive adversaries appears much later and is due to the recent breakthrough of~\cite{BernsteinBKS25} 
who obtained $\Ot(n^{3/4})$ amortized update time randomized algorithms and $\Ot(n^{8/9})$ time deterministic ones. This result was very recently improved to $n^{1/2+o(1)}$ amortized update time~\cite{ChuzhoyKS26}. 
While achieving similar numerical bounds, the techniques behind these work and ours as well as~\cite{BehnezhadRW25,FlinH25} are entirely disjoint. 

Finally, very recently~\cite{AssadiY2026} (see also~\cite{HaeuplerMRST2026}) presented an algorithm for $O(\Delta/\ln{(\Delta)})$ coloring of \emph{triangle-free} graphs in $n^{o(1)}$ update time against adaptive adversaries. That problem 
is fundamentally different from $(\Delta+1)$ coloring on \emph{general} graphs that have a ``rigid'' constraint of only having a single slack color per vertex (the result of~\cite{AssadiY2026} is best compared to the prior algorithm of~\cite{BhattacharyaCHN18}
for $\Delta+o(\Delta)$ coloring general graphs with $\poly\!\log{\!(n)}$ update time).  A related corollary of our work is for coloring \emph{$\gamma$-sparse} graphs wherein neighborhood of any vertex has at most $(1-\gamma) \cdot \binom{\Delta}{2}$ edges (triangle-free graphs are $1$-sparse with this definition). 

\begin{corollary}\label{cor:locally-sparse}
		There is a randomized algorithm that given any $\gamma \in (0,1)$ on any $n$-vertex dynamic $\gamma$-sparse graph with maximum degree $\Delta$ (sufficiently large w.r.t. $\gamma$) 
		maintains a $(1-\Theta(\gamma))\Delta$-coloring  
		against an adaptive adversary in $\Ot(\sqrt{n}/\gamma)$ amortized update time with high probability. 
\end{corollary}
\Cref{cor:locally-sparse} gives a dynamic algorithm for the classical result of~\cite{MolloyR97} on coloring $\gamma$-sparse graphs with less than $(\Delta+1)$ colors. Coloring sparse vertices
is also often used as subroutines in other coloring problems (including~\Cref{res:main} and~\cite{BehnezhadRW25,FlinH25}), and thus~\Cref{cor:locally-sparse} can be its own independent interest. 

\subsection{Our Techniques}\label{sec:techniques}
We give a detailed overview of our
approach in~\Cref{sec:high-level}, and stick to high level bits here. 

Our algorithm, similar to~\cite{BehnezhadRW25,FlinH25}, 
uses \emph{sparse-dense decompositions}~\cite{Reed98} that have found numerous applications in distributed and sublinear algorithms~\cite{HarrisSS16, AssadiCK19,ChangLP2020}. Both~\cite{BehnezhadRW25,FlinH25} spent bulk of their efforts in coloring the dense parts of the graph and primarily relied on existing ideas---the \emph{one-shot coloring technique}~\cite{MolloyR02, HarrisSS16, AssadiCK19, ChangLP2020} and sublinear time algorithms for $(\Delta+1)$ coloring~\cite{AssadiCK19,AssadiY25}---for coloring the sparse parts. 

As pointed out by~\cite{FlinH25}, 
``the randomized argument underlying this approach [of one-shot coloring for sparse vertices] is fragile, requiring a complete rerun every $\Theta(\Delta)$ update''. Since each run of coloring sparse vertices requires nearly $\Theta(n^2/\Delta)$ time (see~\cite{AssadiCK19} for both upper and lower bounds), even just this step leads to a $\Theta(n^2/\Delta^2)$ amortized update time. Therefore, at best, these algorithms can have $\Omega(\min\set{\Delta,(n^2/\Delta^2)})$ update time\footnote{The minimum is because one can always run the naive deterministic $O(\Delta)$ update time algorithm instead.} which can be $\Omega(n^{2/3})$ in the worst case. This is the ``natural barrier'' identified by~\cite{FlinH25}, who optimized the arguments in~\cite{BehnezhadRW25} using structural insights on dense parts of the graph, to reduce the update time to this threshold. 

Our work complements the approaches in~\cite{BehnezhadRW25,FlinH25} by replacing the ``fragile'' arguments for sparse vertices. Instead of restarting with a sublinear time
algorithm at $\Theta(\Delta)$ intervals, we design a dynamic algorithm that \emph{continuously} maintains a coloring of sparse vertices with no restarts. 
The key to our algorithm
is the \textbf{dynamic entropy compression} method---a dynamic version of the celebrated technique of~\cite{moser2009, moserTardos2010} for algorithmic Lovasz Local Lemma (LLL)---that was discovered very recently and independently in~\cite{AssadiY2026, HaeuplerMRST2026}. 
 A main technical contribution of our work is to use this technique beyond local search to analyze an \emph{inherently} dynamic process,
 which in turn allows us 
 to obtain update times that can be sublinear in the number of involved constraints (as opposed to~\cite{AssadiY2026, HaeuplerMRST2026}; see~\Cref{sec:high-level} for more details).  


The remainder of our algorithm mostly follows prior work in~\cite{BehnezhadRW25,FlinH25} with three crucial differences: (1) we need a new classification of dense parts and a new algorithm for handling a certain family of dense vertices 
missing from~\cite{BehnezhadRW25,FlinH25}, and (2) we no longer have any notion of restarts and phases that was used in~\cite{FlinH25,BehnezhadRW25}  even when coloring dense vertices; and (3) we obtain our ``high probability guarantees'' in an
amortized way that relies on some less standard techniques. We discuss our algorithm and its similarities and differences with~\cite{FlinH25,BehnezhadRW25} as well as its connections to dynamic entropy compression~\cite{AssadiY2026, HaeuplerMRST2026} in~\Cref{sec:high-level}. 

Finally, we note that, as pointed out by~\cite{FlinH25}, $O(\sqrt{n})$ forms a ``fundamental limit'' of current techniques: 
finding an available color for recoloring a vertex $v$ is primarily done by either (1) iterating over all its neighbors or (2) picking a random color and going over all vertices with that 
color to check if any are neighbor to $v$. 
The former strategy requires $O(\Delta)$ time and the latter $O(n/\Delta)$ (expected) time (given a random color is used by $n/(\Delta+1)$ vertices in expectation), leading to $\Theta(\sqrt{n})$ time
for recoloring a vertex in the worst-case (of course, assuming a random color even works for coloring the vertex with ``good enough'' probability, which involves the bulk of the effort in algorithms using this subroutine such as~\cite{AssadiCK19,AssadiY25,BehnezhadRW25,FlinH25})%
\footnote{We emphasize that this $\Theta(\sqrt{n})$ is not a formal lower bound, but rather the ``limit'' of current $(\Delta+1)$ coloring algorithms using random trial arguments
which is currently the only known method for dynamic and sublinear-time algorithms. But, it is also worth mentioning that in case of sublinear-time algorithms, $\Omega(n\sqrt{n})$ is a formal lower bound on the time needed to $(\Delta+1)$ color a graph~\cite{AssadiCK19}, 
implying that we cannot hope for a faster subroutine for coloring a \emph{single} vertex \emph{statically} better than $\Theta(\sqrt{n})$ time in the worst case.}.

\paragraph{Independent and concurrent work.} Parallel to us, Gil and Halldorsson \cite{HGil2026} obtained a fully dynamic $(\Delta+1)$-coloring algorithm with $\tilde{O}(\sqrt{n})$ worst case update time against an adaptive adversary.
The two works use entirely disjoint techniques.
Our approach combines dynamic entropy compression with sparse–dense decomposition.
In contrast, they entirely forego the use of sparse–dense decomposition and instead employ a phased-based algorithm that uses global sampling to maintain a set of core vertices whose coloring can be updated in $\tilde{O}(\sqrt{n})$ time, and extends this coloring to the remaining vertices through random color trials.






%% file: prelim.tex

\newcommand{\varbs}[1]{\textnormal{\ensuremath{\textsf{#1}}}}

\newcommand{\out}{\textnormal{\ensuremath{\varbs{out}}}}
\newcommand{\antiDegree}{\textnormal{\ensuremath{\varbs{anti}}}}

\newcommand{\avgout}{\textnormal{\ensuremath{\overline{\varbs{out}}}}}
\newcommand{\avgantiDegree}{\textnormal{\ensuremath{\overline{\varbs{anti}}}}}

\newcommand{\OUT}{\textnormal{\ensuremath{\varbs{OUT}}}}
\newcommand{\ANTIDegree}{\textnormal{\ensuremath{\varbs{ANTI}}}}

\newcommand{\Avail}[2]{\textnormal{\ensuremath{\varbs{AVAIL}_{#2}(#1)}}}
\newcommand{\avail}[2]{\textnormal{\ensuremath{\varbs{avail}_{#2}(#1)}}}
\newcommand{\Free}[2]{\textnormal{\ensuremath{\varbs{Blank}_{#2}(#1)}}}
\newcommand{\free}[2]{\textnormal{\ensuremath{\varbs{blank}_{#2}(#1)}}}

\newcommand{\slack}[2]{\textnormal{\ensuremath{\varbs{slack}_{#2}(#1)}}}

\newcommand{\constMD}{\ensuremath{\textnormal{const}_{\ref{prop:maintaining-decomposition}}}}
\newcommand{\constSD}{\ensuremath{\textnormal{const}_{\ref{def:sparse-dense-decomp}}}}

\newcommand{\Vsparse}{\ensuremath{V_{\textnormal{\texttt{sparse}}}}}

\newcommand{\KKah}{\ensuremath{\mathcal{K}_{\textnormal{\texttt{AH}}}}}
\newcommand{\KKoh}{\ensuremath{\mathcal{K}_{\textnormal{\texttt{OH}}}}}
\newcommand{\KKol}{\ensuremath{\mathcal{K}_{\textnormal{\texttt{OL}}}}}

\newcommand{\Ncolor}[2]{\textnormal{\ensuremath{\varbs{colored}_{#2}(#1)}}}
\newcommand{\Setcolored}[2]{\textnormal{\ensuremath{\varbs{Colored}_{#2}(#1)}}}

\newcommand{\RandomColorTrial}{\textnormal{\ensuremath{\algs{RandomColorTrial}}}\xspace}

\section{Preliminaries}\label{sec:prelim}

\paragraph{Notation.} For any integer $b \geq 1$, define $[b] := \set{1,\cdots,b}$. 
For a graph $G=(V,E)$, we use $N(v)$ to denote the neighbors of $v \in V$, and $\deg(v):=\card{N(v)}$ to denote the degree of $v$. 
For a set $X \subseteq V$, we further use $N_X(v) := N(v) \cap X$ and $\deg_X(v) := \card{N_X(v)}$ to denote the neighbors and degree of $v$ to $X$. 
We say $(u,v)$ is an \emph{anti-edge} in $G$ if $u$ and $v$ are not neighbor in $G$. 


A \textbf{partial coloring} of a graph $G=(V,E)$ is a function
$\varphi : V \to \mathcal{C} \cup \{\perp\}$, where $\mathcal{C}$ is a set of colors and
$\perp$ denotes an uncolored vertex, also referred to as a \emph{blank} vertex.
A partial coloring $\varphi$ is said to be \emph{proper} if no edge is monochromatic, however vertices assigned color $\perp$ are exempt, i.e., are allowed to be adjacent to each other. 
For a color $c\in [\Delta+1]$, we call the set of all vertices with color $c$ under the coloring $\varphi$ the \textbf{color class} of $c$ and use $\mathcal{C}_c$ to denote it.

For a vertex $v \in V$, the set $\Avail{v}{\varphi}$ denotes the colors \emph{available} to $v$ under $\varphi$, namely, colors whose assignments to $v$ keeps $\varphi$ proper. 
Moreover, $\Free{v}{\varphi}$ denotes the neighbors of $v$ that are blank and $\Setcolored{v}{\varphi}$ to be neighbors that are colored, both under $\varphi$. 
We further define 
\begin{alignat*}{2}
 	&\avail{v}{\varphi} := \card{\Avail{v}{\varphi}}, \quad &&\free{v}{\varphi} := \card{\Free{v}{\varphi}}, \\
	&\Ncolor{v}{\varphi} :=\abs{\Setcolored{v}{\varphi}}, \quad \quad &&\slack{v}{\varphi} := \avail{v}{\varphi} - \free{v}{\varphi}. 
\end{alignat*}
Notice that $\slack{v}{\varphi}$ measures the number of ``extra colors'' available to $v$ under the coloring $\varphi$, and that $\slack{v}{\varphi} \geq 1$ always (for any $(\Delta+1)$ coloring $\varphi$). 

\paragraph{Dynamic graphs.} We denote the static vertex set of the underlying $n$-vertex graph by $V := [n]$.  
Let $\mathcal{G} = (G_0, G_1,\cdots, G_t, \cdots)$ be the sequence of dynamic graphs presented to the algorithm: the initial graph $G_0$ is empty and each graph $G_t$ is obtained from the previous graph $G_{t-1}$ by inserting 
or deleting a single edge $e_t = (u_t,v_t)$. We use $G_t := (V,E_t)$ to denote the graph at the $t$-th step; $N_t(v)$ and $\deg_t(v)$ for any $v \in V$ and $r \geq 0$ are defined analogously with respect to $G_t$. 

As is common in this context, we assume $\Delta$ is known a priori and every vertex has degree at most $\Delta$ throughout the updates. 

An \textbf{oblivious} adversary fixes the choices of graphs in $\mathcal{G}$ before presenting them to the algorithm whereas an \textbf{adaptive} adversary can choose each graph $G_t$ based on the output of the algorithm
on prior graphs $G_1,\ldots,G_{t-1}$ (in our case, we can even assume the prior random coin tosses of the algorithm are available to the adversary; this is often called a clairvoyant adversary in the literature; see., e.g., \cite{HaeuplerMRST2026}). 


Finally, the \textbf{recourse} of a dynamic coloring algorithm is the number of changes to the output of the algorithm after each single update.

\paragraph{Concentration Inequalities.} We use the following standard version of Chernoff bound. 

\begin{proposition}[Chernoff-Hoeffding bound; cf.~\cite{MolloyR02}]\label{prop:chernoff}
	Let $X_1,\ldots,X_n$ be $n$ independent random variables where each $X_i \in [0,b]$. Define $X:= \sum_{i=1}^{n} X_i$. Then, for any $t > 0$, 
	\begin{align*}
		\Pr\Paren{\card{X - \expect{X}} > t} \leq 2 \cdot \exp\paren{-\frac{2t^2}{n \cdot b^2}}.
	\end{align*}
	Moreover, for any $\delta \in (0,1)$, and $\mu_{min} \leq \expect{X} \leq \mu_{max}$: 
	\begin{align*}
		\Pr\paren{X > (1+\delta) \cdot \mu_{max}} \leq \exp\paren{-\frac{\delta^2 \cdot\mu_{max}}{3b}} \quad , \quad		\Pr\paren{X < (1-\delta) \cdot\mu_{min}} \leq \exp\paren{-\frac{\delta^2 \cdot\mu_{min}}{2b}}.
	\end{align*} 
\end{proposition}

We will also be using some more specialized concentration inequalities in certain parts of the proofs such as Talagrand's inequality~\cite{talagrand1995concentration,MolloyR02} or an adaptive form of Chernoff bound~\cite{freedman1975tail,kuszmaul2025}
but postpone stating them to when we need them.

\paragraph{Random Color Trials.} The following is a standard subroutines used by our algorithm and various prior work~\cite{AssadiY25, BehnezhadRW25,FlinH25}. 


\begin{ourbox}
	$\RandomColorTrial(v)$: Sample a color $c$ for vertex $v$ and if $N(v)\cap \mathcal{C}_c = \emptyset$ color $v$ by $c$ and terminate, otherwise repeat.
\end{ourbox}

\begin{proposition}\label{prop:random-color-trial}
	Given a vertex $v$ and partial coloring $\varphi$, the subroutine $\RandomColorTrial(v)$ finds an available color for $v$ in expected $O(n/\avail{v}{\varphi})$ time.	
\end{proposition}
\begin{proof}
	First note that as there are $n$ vertices and $\Delta+1$ colors, for a uniformly at randomly sampled color $c$, we have
	\[
	\Exp\abs{\mathcal{C}_{c}} \leq \frac{n}{\Delta+1}.
	\]
	
	Let $c_1$ be the first random color sampled by $\RandomColorTrial(v)$.
	We can expand the expected runtime by going over the cases where the first sampled color is available to $v$ or not as below
	\begin{align*}
		&\expect{\text{runtime of }\RandomColorTrial(v)} 
		\\
		&= O(1)\cdot \Exp\abs{\mathcal{C}_{c_1}} + (1-\Pr(c_1\in \Avail{v}{\varphi}))\cdot \expect{\text{runtime of }\RandomColorTrial(v)};
	\end{align*}
	this is because if $c_1$ is not available, we are going to run the same exact process again. 
	Finally, the equation solves to 
	\begin{align*}
		\expect{\text{runtime of }\RandomColorTrial(v)} 
		&= O(1)\cdot 1/\Pr(c_1\in \Avail{v}{\varphi})\cdot \Exp\abs{\mathcal{C}_{c_1}} \\
		&= O(1)\cdot \frac{\Delta+1}{\avail{v}{\varphi}}\cdot \frac{n}{\Delta+1} = O\paren{\frac{n}{\avail{v}{\varphi}}},
	\end{align*}
	concluding the proof.
\end{proof}

%% file: technical-overview.tex

\section{Technical Overview}\label{sec:high-level}

We now give a streamlined overview of our proof of~\Cref{res:main} and a comparison with prior work in~\cite{BehnezhadRW25,FlinH25} as well as~\cite{AssadiY2026,HaeuplerMRST2026}. Moreover, given the importance of the dynamic entropy compression technique in our work---and our hope that it can be a general technique of its own independent interest for designing \emph{randomized} dynamic algorithms against adaptive adversaries---we attempt to give a full account of this technique in our overview. 
This section oversimplifies many details and the discussions will be informal for the sake of intuition. 

\subsection{Prior Approaches and Challenges}
\subsubsection*{Background: Sparse-Dense Decompositions and One-Shot Coloring} 

We start with the general paradigm of graph coloring via sparse-dense decompositions.  
At a high level, a sparse-dense decomposition partitions the vertices of the graph into
two separate parts: 
\begin{itemize}
	\item \textbf{sparse vertices}: vertices with $\Omega(\Delta^2)$ edges \emph{missing} from their neighborhoods, that is their neighborhoods are ``far from'' being a $(\Delta+1)$-clique; 
	\item \textbf{dense vertices}: remaining vertices, clustered into \textbf{almost-cliques} $K_1,...,K_k$, each of which can be turned a $(\Delta+1)$-clique by adding or removing $o(1)$ fraction of their vertices and edges.
\end{itemize}
See~\Cref{sec:sparse-dense-decomp} for the formal definition and~\Cref{fig:sparse-dense} for an illustration. 

\begin{wrapfigure}{r}{0.35\textwidth}
    \centering
    \vspace{-10pt}
    \includegraphics[width=\linewidth]{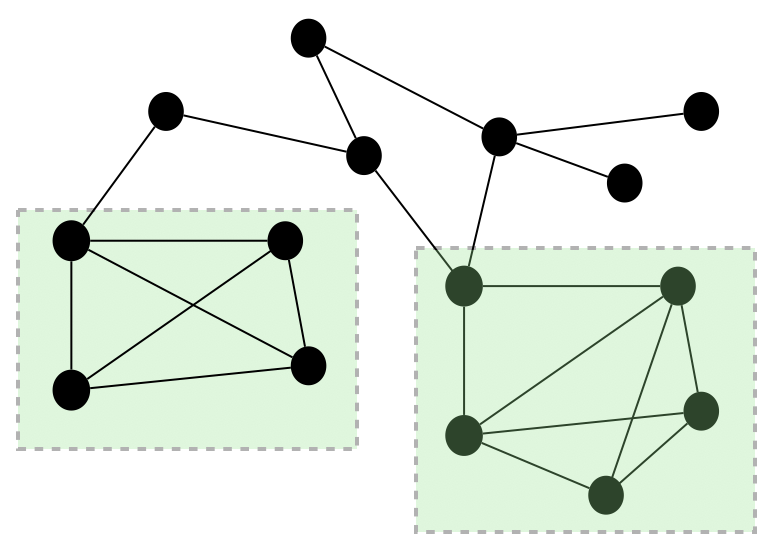}
    \caption{An illustration of a sparse-dense decomposition on a graph with $\Delta=4$ (each box is an almost-clique and remaining vertices are sparse).}
    \label{fig:sparse-dense}
    \vspace{-10pt}
\end{wrapfigure}

The appeal of this decomposition across various scenarios is as follows. Coloring sparse vertices can effectively be reduced to an $O(\Delta)$ coloring problem, that is, with a non-trivial slack compared to the ``rigid'' $(\Delta+1)$ coloring. 
On the other hand, almost-cliques have a lot more structure than arbitrary graphs that one can exploit to obtain a $(\Delta+1)$ coloring. We refer the interested reader to~\cite{MolloyR10,AssadiCK19,ChangLP2020} for quite different instantiations of this general paradigm. 

The reduction for sparse vertices is often obtained via the so-called \emph{one-shot coloring} technique: suppose we independently sample a color from $(\Delta+1)$ for each vertex and assign it if none of its neighbors
have sampled the same color. It is not hard to argue that each vertex has an $\Omega(1)$ probability of being colored in this process. What is more interesting---and quite harder to prove---is that in fact, in the neighborhood of each sparse vertex, 
there will be $\Omega(\Delta)$ \emph{anti-edges} whose both endpoints retain the \emph{same} color (see~\cite{MolloyR02, HarrisSS16, ChangLP2020} for various proofs of this). 
Hence, from hereon, each sparse vertex has a slack of $\Omega(\Delta)$ colors and one can exploit this slack effectively \emph{as if} the goal was an $O(\Delta)$ coloring instead. 

Finally, coloring the dense parts of the graph often involves the bulk of the efforts and are done using a wide range of techniques that are quite tailored to specific setting at hand. 

\subsubsection*{Sparse-Dense Decompositions and One-Shot Coloring in~\cite{BehnezhadRW25,FlinH25}}

The work of~\cite{BehnezhadRW25} is the first to use this paradigm for designing \emph{dynamic} algorithms, followed by~\cite{FlinH25} who further optimized this approach, especially in terms of handling dense vertices. 

Specifically,~\cite{BehnezhadRW25} developed an algorithm for maintaining a sparse-dense decomposition of a fully dynamic graph at all time steps. This algorithm is randomized and its 
update time depends on the sparsity parameter of the decomposition ($\eps$ in~\Cref{sec:sparse-dense-decomp}). The sparsity parameter is (substantially) sub-constant in~\cite{BehnezhadRW25} but it is a small constant
in~\cite{FlinH25} as well as ours and thus the decomposition can be maintained in $O(\log{n})$ amortized update time for~\cite{FlinH25} and us. Hence, for the rest of this discussion we assume having access to this decomposition for free. 

Both~\cite{BehnezhadRW25,FlinH25} color the sparse vertices in \emph{phases} as follows. Each phase consists of $\Theta(\Delta)$ updates of the adversary. 
At the beginning of each phase, they uncolor all vertices and run the one-shot coloring algorithm described earlier to create $\Omega(\Delta)$ slack for sparse vertices. This step can be implemented
in $\Ot(n^2/\Delta)$ time using the sublinear time $(\Delta+1)$ coloring algorithms of~\cite{AssadiCK19,AssadiY25}. 
Each update can reduce the slack of any fixed vertex by only one. As such,  
throughout the entire $\Theta(\Delta)$ updates of the phase, the sparse vertices still maintain an $\Omega(\Delta)$ slack.  

Given the maintained slack, dynamic coloring of sparse vertices can be done using \emph{random color trials} in $O(n/\Delta)$ expected time (by \Cref{prop:random-color-trial} and the $\RandomColorTrial$ subroutine). 
Of course, we also have to maintain this slack which
requires \emph{restarting} the entire process for each phase and spending $\Ot(n^2/\Delta)$ time to run one-shot coloring via~\cite{AssadiCK19,AssadiY25} (given the sublinear time lower bound of~\cite{AssadiCK19}, we also
cannot hope to run this step faster). 
Amortized over the $\Theta(\Delta)$ updates of the phase, this leads to an $\Ot(n^2/\Delta^2)$ amortized time which is the dominating 
term in the algorithm. As we can balance this update time against the trivial $\Delta$ update time algorithm, this leads to an $\Ot(n^{2/3})$ amortized update time algorithm for $(\Delta+1)$ coloring sparse vertices.

Of course,~\cite{BehnezhadRW25,FlinH25} also need to handle dense vertices which have their own unique challenges. But as is clear from the above discussion, the $\Ot(n^{2/3})$ bottleneck in the update time already applies to coloring
sparse vertices and a faster algorithm needs to handle these parts differently. We can thus already switch to discussing how to handle this challenge. 


\subsubsection*{Coloring Sparse Vertices without Phases and Restarts?} 

A natural goal is to let go of phases and restarts, while still maintaining the \emph{invariant} that $\slack{v}{\varphi} = \Omega(\Delta)$ for output coloring $\varphi$ and every sparse vertex $v$ at all times. 
Assuming we maintain this, the same random trial approaches of~\cite{BehnezhadRW25,FlinH25} ensures that we can recolor any vertex made blank by the adversary in $O(n/\Delta)$ expected time, which balanced out with the $O(\Delta)$ time algorithm leads to an $O(\sqrt{n})$ update time algorithm as desired. The challenge is clearly in maintaining this {invariant} as there are various ways for the adversary to force our coloring it to violate it across many vertices, e.g., by replacing a same-colored anti-edge
with an edge, that can affect slack of $\Omega(\Delta)$ vertices in one go. 

\clearpage

\begin{wrapfigure}{r}{0.5\textwidth}
    \centering
    \vspace{-10pt}
    \includegraphics[width=\linewidth]{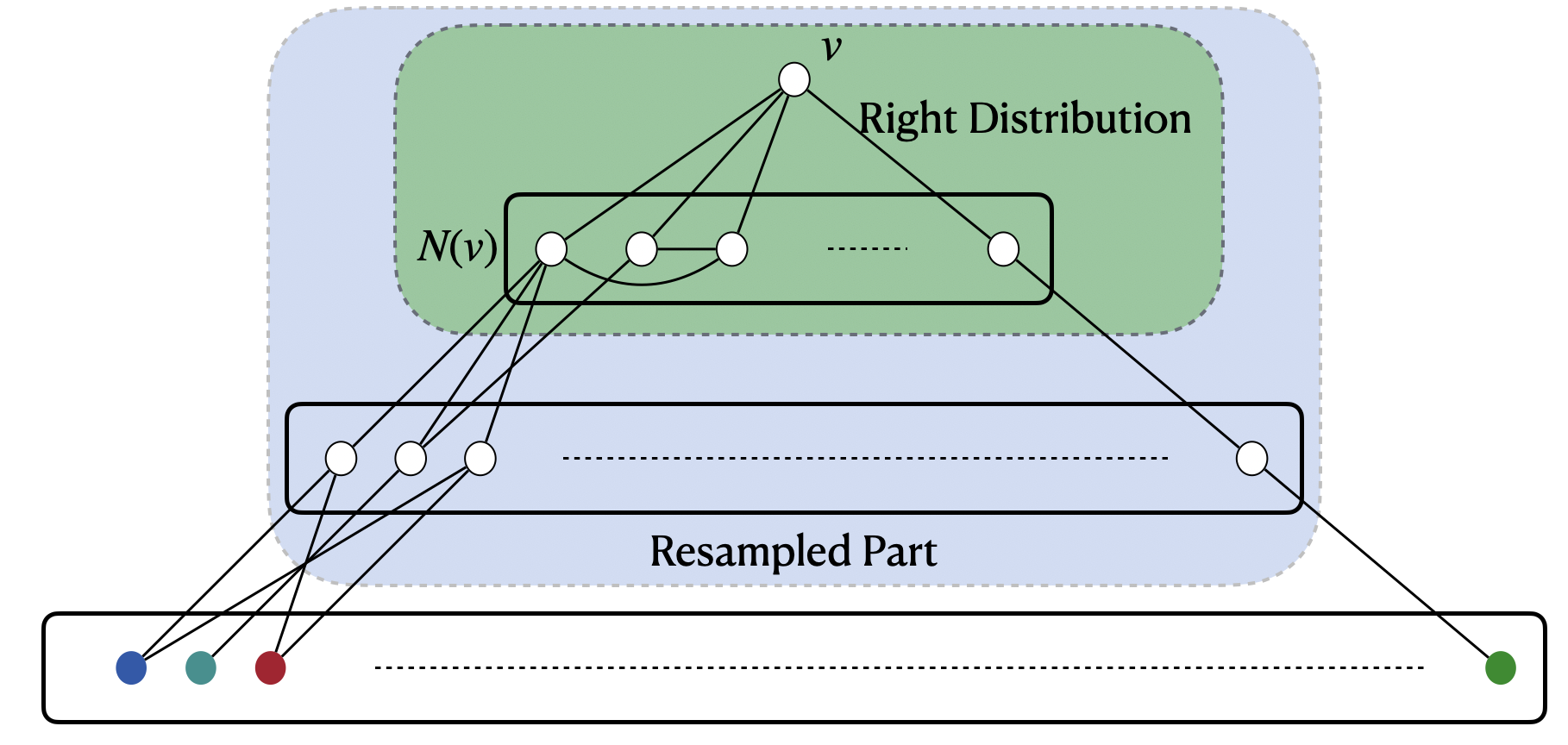}
    \caption{Resampling random colors in two-hop neighborhood of $v$ ensures neighbors of $v$ have the same distribution as one-shot coloring. This is not true
    for all resampled vertices given the colors blocked by their own neighbors that are not resampled.}
    \label{fig:two-hop}
    \vspace{-15pt}
\end{wrapfigure}
A natural approach---especially in light of the recent successes in~\cite{AssadiY2026,HaeuplerMRST2026}---is to \emph{``locally'' fix} any vertex $v$ when it violates the invariant. This means 
resampling the colors in the \emph{two-hop} neighborhood of $v$ to ensure the distribution of color assignments in $N(v)$ itself
follows that of one-shot coloring. Recoloring  fixes the invariant for the vertex $v$ with high probability (by same one-shot coloring arguments as before) but it has the  potential to violate it for vertices in the \emph{three-hop} neighborhood of $v$. We can then just continue this fixing for other violated vertices recursively. 

At first glance, it may not be clear why such an approach would terminate even on a \emph{static} graph as there is no clear notion of ``progress'' in these local fixes. This question however has been answered by an elegant line of work on algorithmic Lovasz Local Lemma (LLL) and the \emph{entropy compression} technique of~\cite{moser2009, moserTardos2010} provides the analytical tools for measuring such a progress. 
Very recently,~\cite{AssadiY2026,HaeuplerMRST2026} developed a \emph{dynamic} entropy compression as well that is tailored to analyzing dynamic algorithms of the type we described. 
In particular,~\cite{HaeuplerMRST2026} provides a general framework that allows one to dynamically maintain almost any assignment guaranteed to exist by LLL. Thus, the framework of~\cite{HaeuplerMRST2026} directly guarantees the convergence of
the dynamic algorithm we described after each update. 

The caveat is in the \emph{rate} of convergence: 
the update time this way will have some $\poly(\Delta)$ term dependence (it will be $O(\Delta^3)$ or $O(\Delta^2)$ at best depending on which implementation from~\cite{AssadiY2026,HaeuplerMRST2026} is used but certainly it cannot be $o(\Delta)$ time). 
This entirely defeats the purpose for $(\Delta+1)$ coloring as we could have just run the trivial $O(\Delta)$ update time algorithm instead%
\footnote{
As stated earlier~\cite{AssadiY2026} does in fact achieve an $n^{o(1)} = o(\Delta)$ update time for $O(\Delta/\ln(\Delta))$ coloring of triangle-free graphs. But, this is done in two steps: (1) designing an $O(\Delta^2\log{n})$
update time algorithm using dynamic entropy compression, and, then (2) using this as a \emph{black-box} in a reduction inspired by~\cite{BhattacharyaCHN18} to reduce the update time to $\Delta^{o(1)}\log{n}$ while losing a constant factor 
in the number of colors used. Such an ``approximate'' reduction is inherently not suitable for $(\Delta+1)$ coloring so this two-step approach is not applicable for us (in fact, the reduction does not even work in this context
for $\gamma$-sparse graphs in~\Cref{cor:locally-sparse} for $\gamma < 1$, even ignoring the approximation aspect).}.


Nevertheless, a simple observation is that
the problem we need to solve is \emph{not} a local problem---in the sense of requiring any LLL argument to begin with \emph{if} we were considering it in the \emph{static} setting---and thus we can deviate from 
the specific frameworks of~\cite{moser2009, moserTardos2010} and other LLL approaches and focus on the problem in an \emph{inherently dynamic} way. Our approach exploits this. 

\subsection{Our Approach for Sparse Vertices}
\subsubsection*{Our Algorithm for Coloring Sparse Vertices}
The goal in our algorithm is to have a partial coloring $\colS$ which can generate $\Omega(\Delta)$ slack for sparse vertices, and then further \emph{extend} it by coloring its blank vertices in the output coloring $\varphi$ (by definition, 
any slack generated by $\colS$ will also be inherited by $\varphi$). Using a similar language as in the LLL literature, we say a vertex $v$ has a \textbf{flaw} iff $\slack{v}{\colS} = o(\Delta)$. While the typical goal, e.g., in~\cite{HaeuplerMRST2026}, 
would be to maintain a flawless coloring, this is \emph{not} what we will do. Instead, we will 
only fix this flaw on a need-based basis. Specifically, we first have the following subroutine (for now, let us assume we only have sparse vertices; we will come back to dense vertices later): 
\begin{ourbox}
$\Attempt(v)$: attempts to color $v$ in $\varphi$ by sampling $\Theta(\log{n})$ random colors for it and returns \Fail if it could not find an available color for $v$ among any of them.  
\end{ourbox}
\vspace{-5pt}
Similar to random color trials (\Cref{prop:random-color-trial}),~\Attempt only requires $\Ot(n/\Delta)$ expected time for coloring $v$ \emph{assuming} $\slack{v}{\varphi} = \Omega(\Delta)$, i.e., $v$ is flawless. 

Of course, there will be time steps wherein $\Attempt(v)$ returns \Fail (essentially) when $v$ has a flaw. In that case, we will try to \emph{just} fix $v$ by recoloring its neighbors (as opposed to two-hop neighborhoods
required by ``LLL approaches'' as in~\cite{HaeuplerMRST2026}). For this resampling to be successful, we need the list of available colors for vertices to be large enough to ensure the probabilistic guarantees needed for creating enough slack for $v$ hold. 
We also need the {distribution} of resamplings here to be \emph{oblivious} to the current state of the graph and underlying colorings $\colS$ and $\varphi$.\footnote{We will discuss this in more details later but roughly speaking this is because while we are doing the resampling at some time $t$, we will be analyzing its effect at another time $t'$ and need the distributions to match.} Given these, we have the following resampling subroutine: 
\begin{ourbox}
	$\Resample(v)$: for $u \in N(v)$ in lexicographic order, run the following \textbf{sampling process}: 
	\begin{enumerate}
	\item with a small constant probability, sample a color $c_u \in [\Delta+1]$ and otherwise let $\colS(u) = \perp$;
	\item if $c_u \in \Avail{u}{\colS}$ (after updating for prior vertices in the loop), let $\colS(u) = c_u$. 
	\end{enumerate}
\end{ourbox}
We emphasize that $\Resample$ only looks at the colors in $\colS$ and not the entire $\varphi$. So, it is possible for it to assign a color to a vertex $u$ that also appears in $N(u)$ under $\varphi$ (but not $\colS$). 
Thus, after running $\Resample$, we also need a \emph{clean up} step. 

\begin{ourbox}
	$\Cleanup(v)$: after calling $\Resample(v)$: 
	\begin{enumerate}
	\item for any $u$ with a new color $c$ in $\colS$, $\varphi(u) \leftarrow c$ and \emph{uncolor} vertices in $N(u) \cap \mathcal{C}_c$ (w.r.t. $\varphi$); 
	\item for any $u$  uncolored in $\Resample$ or in the previous line, run $\Attempt(u)$. 
	\end{enumerate}
\end{ourbox}

We can now put these together with one additional component (whose role become clear in the analysis): whenever an edge $(u,v)$ is inserted or deleted, we recolor both $u$ and $v$ (in $\colS$ and by extension $\varphi$ also) using the same 
sampling process defined in $\Resample$ and do the same clean up for them (if needed). The whole dynamic algorithms is then as follows: update the edge as above, color any blank vertex by $\Attempt$ and if it returns \Fail, run $\Resample$ followed by $\Cleanup$.  

At this point, it is not even clear that $\Resample(v)$ is able to fix the flaw of the vertex $v$; for instance, if $\Avail{u}{\colS}$ is small for all $u \in N(v)$, this subroutine may end up instead just uncoloring all neighbors of $v$. So, there is no 
hope in that case. We can however prove the following result. The statement of this lemma may sound rather cryptic/over-generalized at first but this is precisely what is needed to make the magic of entropy compression work (we will discuss this later). 

\begin{lemma}[Formalized in~\Cref{lem:prob-slack-sparse}]\label{lem:informal-sampling}
	Fix any \emph{ordered} set of $\Delta$ vertices $u_1,\ldots,u_{\Delta}$ with $\Omega(\Delta^2)$ anti-edges between them. Each vertex is associated with a \textbf{blocked set} $B(u_i) \subseteq [\Delta+1]$ of size at most $\Delta/4$. 
	Suppose for every $u_i$, with a small constant probability, we sample a color $c(u_i)$ and assign it $u_i$ iff: $(1)$ the color $c(u_i) \notin B(u_i)$, and (2) no neighbor $u_j$ of $u_i$ with $j < i$ have also sampled this color. All other
	vertices are marked blank. 
	Then, 
	\[
		\Pr\paren{\text{there are less than $\Omega(\Delta)$ same-colored anti-edges between $u_1,\ldots,u_{\Delta}$}} \leq 2^{-\Omega(\Delta)}. 
	\]
\end{lemma}
One way of interpreting~\Cref{lem:informal-sampling} is this: fix a time $t$ and a vertex $v$ in the current graph $G_t$ with coloring $\colS$ and let $u_1,\ldots,u_{\Delta}$ be $G_t[N(v)]$; define 
$B(u_i) = [\Delta+1] \setminus \Avail{u_i}{\colS}$. \Cref{lem:informal-sampling} says that as long as all $\avail{u_i}{\colS} \geq 3\Delta/4$, $\Resample(v)$ fixes the flaw of $v$ with high probability. 
The proof of this lemma uses this interpretation and that the assumption of $\avail{u_i}{\colS} \geq 3\Delta/4$ for all $u_i \in N(v)$ means for each anti-edge $(u,w)$ inside $N(v)$, $\card{\Avail{u}{\colS} \cap \Avail{w}{\colS}} = \Omega(\Delta)$, 
which is enough to prove this lemma similar to one-shot coloring proofs in~\cite{MolloyR02, HarrisSS16, ChangLP2020}. 

Nevertheless, this interpretation is \emph{not} what we will be using this lemma for, as it will become clear in the analysis sketch later. 

We have still not addressed how $\Resample(v)$ is helpful when $\Avail{u}{\colS}$ is small for some $u \in N(v)$. This is precisely the reason 
that our sampling process only even considers coloring a vertex in $\colS$ with a small constant probability. Intuitively, this means most vertices of the graph are blank and so $\Avail{\cdot}{\colS}$-sets should be large. 
More specifically, suppose before calling $\Resample(v)$, we first call $\Resample(u)$ for every neighbor $u \in N(v)$ where $\avail{u}{\colS} < 3\Delta/4$ (without caring about whether $u$ had a flaw to begin with). It is easy to see that this uncolors most neighbors of $u$ and leads to $\avail{u}{\colS} \geq 3\Delta/4$ with high probability; we can then go back to $\Resample(v)$
and apply~\Cref{lem:informal-sampling} to ensure we fixed the flaw of $v$ (again, this is not what we really need and that interpretation is only for the intuition). 

Formalizing this step requires additional care by introducing this as yet another flaw and then running a ``lazy check'' in the background (see~\Cref{sec:another-flaw}). 
We shall ignore this step in this high-level overview and in the following just \textbf{assume} $\avail{u}{\colS} \geq 3\Delta/4$ for all vertices $u \in V$ at all times. 
We can now get to the analysis of the algorithm. 

\subsubsection*{The Analysis of Our Algorithm for Coloring Sparse Vertices}

We already argued that the calls to $\Attempt$ take $\Ot(n/\Delta)$ expected time. So, analyzing our algorithm boils down to understanding the following
two parts for $\Resample$: (1) how much time a $\Resample$ call takes and how many vertices it uncolors (to account for time of the call to $\Attempt$ for them inside $\Cleanup$), and (2) how often
we have to call $\Resample$.  Among these the second one is the main part of the analysis, but let us start with the first part first. 

\paragraph{Runtime and recourse of $\Resample$.} The runtime of $\Resample$ involves going over all neighbors of $v$ in $O(\Delta)$ time to sample $O(\Delta)$ colors, and then for each sampled color checking its color class (to ensure
it is available to the vertex that sampled it). Given the expected size of a color class for a random color is $O(n/\Delta)$, the entire runtime of $\Resample$ will be $O(n)$ time in expectation. 

As for the number of vertices $\Resample$ uncolors, it may uncolor all $\Delta$ neighbors of $v$ and in addition, for any vertex that it colors, it needs to uncolor its neighbors with the same color. One can see that
the expected number of neighbors of a vertex that have the same color as the vertex itself, for a randomly chosen color, is $\Delta/(\Delta+1) = O(1)$, and thus, the total number of vertices uncolored by a call to $\Resample$ is $O(\Delta)$ in expectation. 

All in all, a  call to $\Resample$ requires $\Ot(n)$ expected time to handle through subsequent calls to $\Cleanup$ and $\Attempt$ \emph{ignoring} the potential other calls made to $\Resample$. 

\paragraph{Number of times $\Resample$ is called.} The heart of our analysis is the following lemma (this lemma needs to be proven in a more general form in the main body of the proof and is 
captured in the proofs of~\Cref{lem:recourse-updatesparse-whp,lem:runtime-UpdateSparseColor}; the following lemma however is correct and its proof sketch below can be fully formalized). 
\begin{lemma}\label{lem:informal-heart}
	Over a sequence of $t$ updates, the total number of times $\Resample$ is called in our dynamic algorithm is $O(t/\Delta)$ with high probability. 
\end{lemma}

Combined with the bounds of $\Ot({n}/\Delta)$ on the runtime of $\Attempt$ (as well as the extra color samplings done per edge insertions and deletions)
 and the $\Ot(n)$ time spent per each $\Resample$ call,~\Cref{lem:informal-heart} implies an $\Ot(n/\Delta)$ bound on the amortized update time of our algorithm. 
Again, balancing this out against the $O(\Delta)$ time will give us the desired $\Ot(\sqrt{n})$ update time. Ensuring this bound holds with high probability rather than in expectation requires another argument which we will touch upon at the end of our 
overview discussion.

Our proof of~\Cref{lem:informal-heart} is where we use a dynamic entropy compression argument. Firstly, suppose we maintain a table 
 \[
\randTable:V \times \IN \rightarrow \{0, 1\}\times [\Delta+1],
\]
for the randomness used by the algorithm during calls to $\Resample$ (see~\Cref{fig:informal-indices}). 

\begin{figure}[h!]
    \centering
    \begin{subfigure}[b]{0.38\textwidth}
        \centering
        \includegraphics[width=\textwidth]{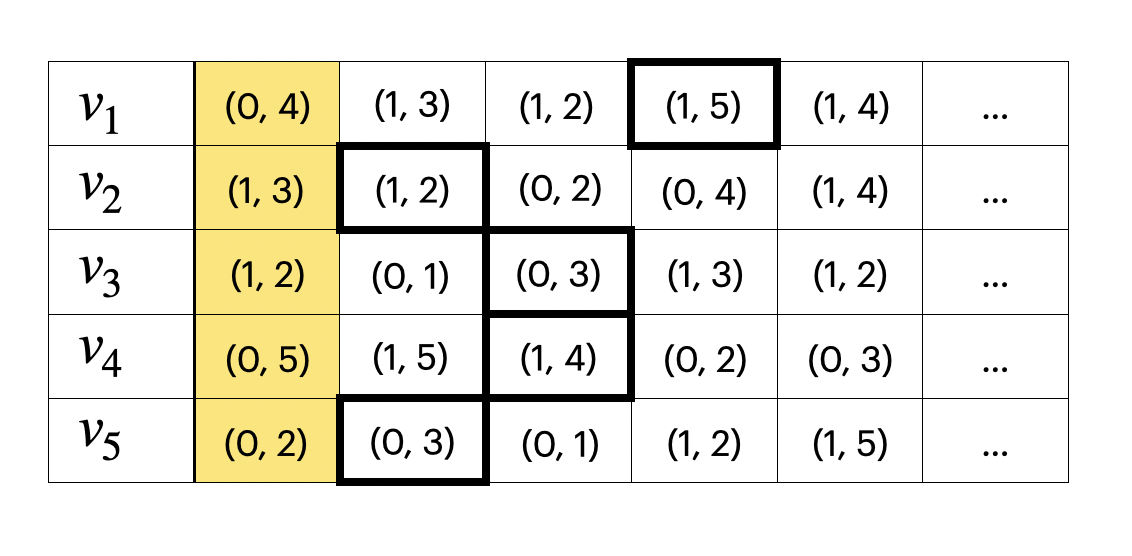}
        \caption{Before calling $\Resample(v_1)$.}
        \label{fig:index-sub1}
    \end{subfigure}
    \hfill
    \begin{subfigure}[b]{0.38\textwidth}
        \centering
        \includegraphics[width=\textwidth]{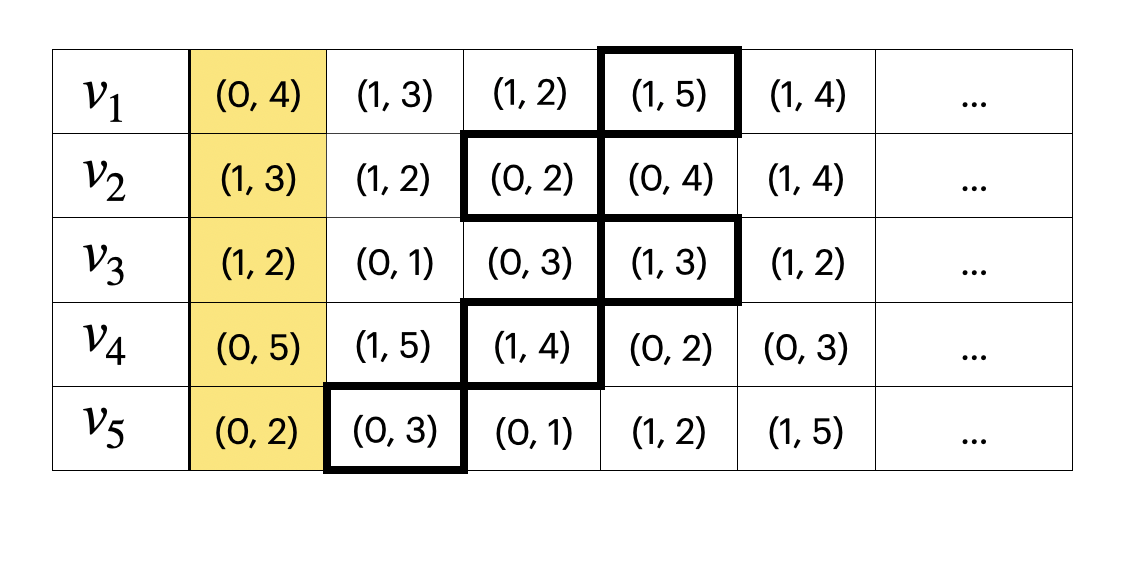}
        \caption{After calling $\Resample(v_1)$.}
        \label{fig:index-sub2}
    \end{subfigure}
    \begin{subfigure}[b]{0.20\textwidth}
        \centering
        \includegraphics[width=\textwidth]{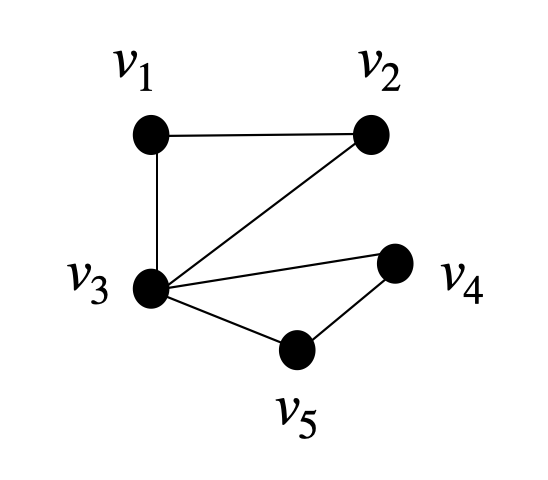}
        \caption{The Graph}
        \label{fig:informal-index-graph}
    \end{subfigure}
    \hfill
    \caption{The table show the random table of assignments for the graph in \Cref{fig:informal-index-graph} and the bold boxes show the ``Index'' of each vertex at the current state.}
    \label{fig:informal-indices}
\end{figure}

Specifically, whenever we call $\Resample$ and need to pick a random color $c_u$ for a vertex $u$, we look at entry $(u,i+1)$ of $\randTable$ where $i$ is the number of times
we have used randomness for $u$ (namely, number of times $u$ has been part of some sampling process either by an update or a $\Resample$ call); the first entry tells us whether $u$ gets to sample a color in $\Resample$ (with a small constant probability), and the second one tells us which 
color $c_u$ is going to be. We pick $\randTable$ so each entry independently matches the distribution we use in our sampling. 

At the same time, when running the algorithm, we can also maintain a \emph{log} of the behavior of the adversary and the algorithm, \emph{just for the purpose of the analysis}. Specifically, we write down a sequence of the form: 
\[
	\cdots, \texttt{insert}(u_1,v_1), \texttt{insert}(u_1,v_2), \Resample(u_1), \texttt{delete}(u_3,v_9), \Resample(u_4), \Resample(u_5), \cdots
\]
which allows one to keep track of the updates by the adversary at any point in time and the $\Resample$ calls that happened between them. A \textbf{crucial point} here is this log only maintains the information about $\colS$ (and not $\varphi$) and for that we needed our sampling process to entirely ignore $\varphi$ while $\Attempt$ is not allowed to change $\colS$ (only $\varphi$). Specifically, given a log and its corresponding table $\randTable$, we can
reconstruct all colors in $\colS$ sampled by our algorithm throughout. In the following, we exclusively refer to the coloring $\colS$ unless stated otherwise. 


We say a log $L$ is \textbf{consistent} with a table $\randTable$ if running the algorithm using the randomness of the table results in the given log. Our goal is to say that ``long'' logs cannot be consistent with the randomness of the table used by 
the algorithm with high probability and use it to bound the number of $\Resample$ calls. This is done in two steps as follows. 

\paragraph{Step 1: number of logs.} Consider any log consisting of $t$ updates and $r$ resample calls. We can bound the total number of such logs by $n^{O(t+r)}$ given there are $\binom{t+r}{r}$ choices for ordering 
the different types of instructions in the log and that each instruction can have $O(n^2)$ possibilities.

\paragraph{Step 2: probability for a fixed log.} Fix any possible log $L$. For $L$ to be consistent with a random $\randTable$, we check if each $\Resample(v)$ call is consistent with $\randTable$, as in, at that time step, 
$v$ has a flaw in the graph under the coloring $\colS$ (determined by using the table and log until this step). 

Consider a time $t$ when a $\Resample(v)$ call is logged and consider the graph $G_t$ which is fully determined by the log $L$. Let $u_1,\ldots,u_{\Delta}$ form $G_t[N(v)]$. For any $i \in [\Delta]$, 
let $t(u_i)$ denotes the last time step $u_i$ sampled a color and let that color be $c(u_i)$. Order $u_1,\ldots,u_{\Delta}$ in ithe ncreasing order of $t(u_i)$'s and break the ties in lexicographic order of vertex names (as is done in $\Resample$).
Define $B(u_i) := [\Delta+1] \setminus \Avail{u_i}{\colS}$ \emph{at the time $t(u_i)$}. 

We have a couple of \textbf{key observations}: 
\begin{enumerate}
\item	 Consider any pair of $u_i,u_j$ and let $t_{ij}$ be the last time step before $t$ when $(u_i,u_j)$ was changed, namely, inserted or deleted. Then, both $t(u_i),t(u_j) \geq t_{ij}$ since we always sample a new color in $\colS$ 
for endpoints of any updated edge; 
\item The color of $u_i$ at time $t$ is either $c(u_i)$ or $\perp$ and moreover, this decision was already fixed at the time $t(u_i)$ depending on whether or not $u_i \in B(u_i)$ (and if $t(u_i)=t(u_j)$ and $c(u_i)=c(u_j)$, whichever
of $u_i,u_j$ appears lexicographically first): the color of a vertex only changes as part of a sampling process and $t(u_i)$ was the last time we sampled a color for $u_i$. 
\end{enumerate}
The samplings of colors (but not necessarily their eventual assignment) is done independently and thus each of $c(u_i)$'s is sampled independently. 
Equipped with the observations above, and even though the samplings are done at different time steps, we can apply~\Cref{lem:informal-sampling} with the parameters above to say that the probability of having $o(\Delta)$ same-colored anti-edges between
$u_1,\ldots,u_{\Delta}$ \emph{at the time step $t$} is only $2^{-\Omega(\Delta)}$. Put differently, the probability $v$ has a flaw at time $t$ and thus $\Resample(v)$ is consistent with the table is $2^{-\Omega(\Delta)}$. 

In addition, the probability that each $\Resample$ call is consistent with $L$ is \emph{independent} of the rest; this is because for each call, we are looking at the different series of entries in $\randTable$ (as after a $\Resample$ happens, we move
to the next entry for each vertex) and all entries 
of the table are independent. This implies that for a fixed log $L$ with $r$ $\Resample$ calls: 
\[
	\Pr\paren{\text{$L$ is consistent with randomness of the algorithm}} \leq 2^{-\Omega(r \cdot \Delta)}. 
\]
\paragraph{Conclude the proof sketch of~\Cref{lem:informal-heart}.} By a union bound, 
\[
	\Pr\paren{\text{algorithm has more than $r$ $\Resample$ calls in $t$ updates}} \leq n^{O(r+t)} \cdot 2^{-\Omega(r \cdot \Delta)},
\]
which for $r \gtrsim (t \cdot \log{n})/\Delta$ is negligibly small (assuming $\Delta \gtrsim \log{n}$ which holds since we can assume w.l.o.g. $\Delta = \Omega(\sqrt{n})$ when running this algorithm). 
This concludes the proof sketch of~\Cref{lem:informal-heart}. 

At this point, we have stated an $\Ot(\sqrt{n})$ update time dynamic algorithm and its analysis for coloring sparse vertices in the decomposition. This algorithm, almost verbatim, is enough to establish~\Cref{cor:locally-sparse} (as we actually
never needed to use $(\Delta+1)$ colors here) but to prove~\Cref{res:main}, we need more work to handle dense vertices as well. 

\begin{remark}
\emph{Before moving on, we note that this particular idea of analyzing \underline{appearance} of flaws in $\Resample$ calls (as opposed to the more natural way of analyzing probability of a $\Resample$ call \underline{fixing} a flaw) is a core idea in the original Moser-Tardos framework~\cite{moser2009, moserTardos2010}. Those work however require quite a lot more different machineries---e.g., witness tree and their Galton-Watson sampling---that are crucial for LLL applications, but are not 
needed for our dynamic applications. Thus, stripping away those machinery, allows us to analyze dynamic algorithms directly and focus on their unique intricacies, e.g., role of updates and changes to the underlying graph in the middle of $\Resample$ calls. 
This has allowed us to provide a dynamic-specific approach (e.g., the generality of \Cref{lem:informal-sampling} and time-obliviousness of our sampling) 
that bypasses the $\poly(\Delta)$ update time barriers in the original applications of dynamic entropy 
compression~\cite{AssadiY2026,HaeuplerMRST2026}. }

\emph{We suspect dynamic entropy compression in general will find a place in the standard toolkit of dynamic algorithms, given it provides a rather unique way to analyze randomized algorithms against adaptive adversaries. See~\cite{AssadiY2026,HaeuplerMRST2026} for further justifications and discussions.}  
\end{remark}

\subsection{Handling Dense Vertices} 

Up until now, including in the discussion of prior work, we have entirely ignored the dense vertices, even though virtually in all scenarios, those are the technically more challenging parts to handle.
This continues to be the case in our algorithm as well, but these challenges are for the most part \emph{technical} and follow the prior work in~\cite{BehnezhadRW25,FlinH25} as well as~\cite{AssadiCK19,AssadiY25}; at the \emph{conceptual level}, the main 
novelty of our work lies in its handling of the sparse vertices. In the following, we briefly describe our approach on how to handle dense vertices especially as how they relate to our earlier treatment of sparse vertices. 
We start with a quick review of~\cite{FlinH25} which we follow more closely (although several techniques in both~\cite{FlinH25} and ours are also borrowed from~\cite{BehnezhadRW25}). 

\subsubsection*{Prior Approach of~\cite{FlinH25} for Dense Vertices}
Recall that dense vertices are partitioned into almost-cliques. For any almost-clique $K$ and a vertex $v \in K$, define $\out(v)$ as the degree of $v$ outside $K$ and $\antiDegree(v)$ as the anti-edge degree of $v$ inside $K$. 

A simple yet brilliant idea in~\cite{FlinH25} is to handle certain almost-cliques also as sparse vertices: almost-cliques wherein \emph{average} $\out$ or $\antiDegree$ of vertices is at least $\eta$ for some $\eta=o(\Delta)$ chosen later. 
Specifically, they prove that such vertices are $\gamma$-sparse for $\gamma \approx \eta/\Delta$, which means  one-shot coloring generates $\Omega(\eta)$ slack for these
vertices as well. For~\cite{FlinH25}, handling these vertices effectively only requires two changes: (1) reducing phase lengths to $\Theta(\eta)$ instead of $\Theta(\Delta)$ earlier, and (2) spending $O(n/\eta)$ time instead of $O(n/\Delta)$ 
to color these vertices using their $\Omega(\eta)$ slack throughout the phase. This means the update time of their algorithm on this part will be 
$
\Ot(\frac{n^2}{\Delta \cdot \eta} + \frac{n}{\eta})
$
by amortizing the $\Ot(n^2/\Delta)$ time of  one-shot coloring over the $\Theta(\eta)$ updates of a phase. 

For the remaining almost-cliques, they use an idea called \emph{inlier-outlier} partitioning with the essential effect of allowing one to assume $\out$ and $\antiDegree$ of \emph{every} vertex in the almost-clique is almost the same as 
the \emph{average} values of these quantities, and thus is upper bounded by $O(\eta)$ that are remained to be handled. Yet another technique called \emph{anti-edge matching}, introduced in~\cite{AssadiCK19} for design of sublinear time algorithms for $(\Delta+1)$ coloring, 
and subsequently used in~\cite{BehnezhadRW25} for dynamic algorithms reduces the problem to the following: whenever we need to color a vertex $v$ in a remaining almost-clique $K$, we only need to go over $\out(v)$ vertices of $N(v) \setminus K$ to find a color available to $v$. But then this step takes $O(\eta)$ time leading to the final update time of this algorithm for coloring dense vertices to be: 
\[
	\Ot\paren{\frac{n^2}{\Delta \cdot \eta} + \frac{n}{\eta} + \eta} = \Ot(\frac{n^{4/3}}{\Delta} + n^{2/3}), 
\]
for $\eta = n^{2/3}$. Balancing this with the $O(\Delta)$ time algorithm gives their $\Ot(n^{2/3})$ time algorithm. 

\subsubsection*{A Challenge for Improving Update Time on Dense Vertices?} 
While we have come up with a way to speed up coloring sparse vertices by removing phases and global restarts, 
there is another challenge to handle dense vertices when following~\cite{FlinH25} by delegating some almost-cliques to be handled
as sparse vertices. This is because the sparsity parameter of these vertices is $\gamma \approx \eta/\Delta$ which can be (polynomially) sub-constant, 
which our algorithm cannot handle without extra cost directly (see the statement of~\Cref{cor:locally-sparse} also). 

This bottleneck is quite fundamental for our earlier approach: for $\gamma$-sparse vertices,~\Cref{lem:informal-sampling}  (provably) can only obtain a $2^{-\Omega(\gamma \Delta)}$ probability, while we still need $\Ot(n)$ time
to run $\Resample$ which \emph{problematically} uncolors $O(\Delta)$ vertices that need to be handled. Going through the entropy compression arguments using these bounds then gives an 
\[
	\Ot\paren{\frac{n}{\gamma\Delta}}~\text{update time and}~\Ot\paren{\frac{1}{\gamma}}~\text{uncolored vertices per update}.\footnote{\label{fn:general-rule}As a general rule, if the flaw probability under resampling can be bounded by $2^{-\ell}$ (and everything aligns well for the argument), then entropy compression allows us to bound the number of resamples by $\approx t/\ell$ over $t$ updates.}
\]
Given $\gamma \approx \eta/\Delta$, we get an $\Ot(n/\eta)$ update time for $\Resample$ calls (which is perfectly fine) but also have to color $O(\Delta/\eta)$ new vertices per update. Since the time needed for coloring these
vertices by color trials is $\Ot(n/\eta)$ (given they have $\Theta(\eta)$ slack), the time needed to color the remaining vertices becomes $O(n\Delta/\eta^2)$ time. This means the update time of the algorithm is 
\[
	\Ot\paren{\frac{n\Delta}{\eta^2} + \frac{n}{\eta} + \eta},
\]
which unfortunately, balanced out against $O(\Delta)$ time algorithm, still does not get a better than $\Ot(n^{2/3})$ time algorithm!\footnote{Even when $\eta$ is as large as $n^{2/3}$ but $\Delta = \Theta(n)$, the first term will also be $n^{2/3}$ now.}

A \textbf{crucial observation} to our rescue is that the almost-cliques we are delegating here are \emph{not arbitrary} $\gamma$-sparse vertices, and we can further exploit their almost-clique properties. With this in mind, we are ready
to describe our approach for dense vertices. 

\subsubsection*{Our Approach for Dense Vertices} 

In our algorithm, we further partition the delegated almost-cliques based on whether average $\antiDegree$ was higher than $\eta$ or just the average $\out$. It turns out the former ones can be handled 
using the dense vertices approaches of~\cite{FlinH25} with an update time of $\Ot(\frac{n}{\eta} + \eta)$ with no particular difficulty. This basically means \emph{not} delegating 
those almost-cliques to sparse vertices and handling them as dense ones. But, the latter almost-cliques do need to be delegated to sparse vertices and we shall handle them using our dynamic entropy compression arguments. 
Remaining almost-cliques can be handled similar to~\cite{FlinH25} in $\Ot(\frac{n}{\eta} + \eta)$ update time. 

The only remaining case is then almost-cliques with average $\out$ of vertices larger than $\eta$, which we call \text{out-heavy} almost-cliques. Recall that the challenge with these almost-cliques is that the $\Theta(\Delta)$ \emph{recourse} of the $\Resample$ subroutine will be prohibitively large to handle the uncolored vertices individually using existing subroutines. 
Here, we can show that calling $\Resample(v)$ for $v$ in an out-heavy almost-clique $K$ only leads to expected $O(\out(v))$ uncolored vertices \emph{outside} the almost-clique and the remaining $\Theta(\Delta)$ uncolored are all inside $K$. 
We then design a new sublinear time algorithm---inspired by~\cite{AssadiY25}---for coloring all but $O(\eta)$ vertices in out-heavy almost-cliques in $\Ot(n)$ time only (instead of $O(n\Delta/\eta)$ time of coloring $O(\Delta)$ vertices one by one). 
This means that effectively we can think of this $\Resample$ call as 
spending $\Ot(n)$ time and generating $O(\out(v)+\eta)$ uncolored vertices instead of $O(\Delta)$ uncolored ones. 

The final part of the argument is to prove we actually do have the budget to color the remaining vertices using the $\Ot(\frac{n}{\eta} + \eta)$ update time for each one. 
This involves proving that (1) the sparsity of a vertex $v$ in an almost-heavy almost-clique is proportional to both $\out(v)$ and $\eta$ simultaneously; and, (2) number of $\Resample$ calls on 
vertices of such sparsity over $t$ updates is $\Ot(t/(\out(v)+\eta))$ essentially (recall~\Cref{fn:general-rule}). This means the amortized runtime of handling these $O(\out(v)+\eta)$ vertices 
is $\Ot(\frac{n}{\eta} + \eta)$ which is within our budget. 

All in all, putting everything together, the above arguments give an amortized update time of 
\[
	\Ot(\frac{n}{\Delta} + \frac{n}{\eta} + \eta), 
\]
for coloring dense vertices. Picking $\eta = \sqrt{n}$ and balancing out against the $O(\Delta)$ time algorithm gives an $\Ot(\sqrt{n})$ update time algorithm for coloring dense vertices as well. 

\subsubsection*{Oversimplified Details and High Probability Guarantees} 

We now briefly touch upon some oversimplified details in the previous discussions and point out how the high probability guarantee of our algorithm is proven. 

Firstly, when handling delegated almost-cliques as sparse vertices, we need to handle a spectrum of sparsities (given appearance of $\out(v)$-term in the bounds), which requires a more nuanced arguments in the entropy compression part. Specifically, we now associate a \emph{cost} with the logs we maintain (analytically) that is a proxy to the time/recourse needed to handle them and bound the probability of ``high'' cost logs. 

Secondly, we described our algorithms as if the decomposition was fixed and that sparse and dense vertices had no 
interaction with each other (not even being neighbor). In actuality, we need to handle vertices moving between different parts of the decomposition, as well as making sure handling of vertices in some parts do not violate invariants of the other; the most important 
part to maintain for entropy compression is to ensure that colorings of all vertices in $\colS$ are always done through a consistent and graph-oblivious sampling. These all lead to further technical challenges in the algorithm as which vertices are allowed
to participate in these sampling steps and how to prioritize their colorings over the rest.  

Finally, the high probability guarantees of prior work in~\cite{BehnezhadRW25,FlinH25} was obtained \emph{per update} (despite still being amortized) by maintaining invariants such as ensuring size of all color classes are fixed to $O(n/\Delta)$ (instead 
of in expectation for a random color that always holds) or only uncoloring $O(1)$ vertices in the worst case. However, our dynamic entropy compression argument relies on taking \emph{expected} time and recourse (due to the obliviousness of the distribution it 
crucially uses), and thus we cannot prove a high probability guarantee per update. Instead, we use a less standard concentration inequality due to~\cite{kuszmaul2025} (\Cref{prop:KQ-inequality}) 
to bound the runtime of the entire algorithm over an adaptive sequence of $t$ updates, which as a by-product gives us higher exponential probability of success for this part\footnote{Our algorithm still only succeeds with high polynomial probability and not exponential, because it uses subroutines like the decomposition by~\cite{BehnezhadRW25} whose (existing) proof only implies a high probability bound. We have not explored improving the probability bound of our overall algorithm given it is a secondary consideration.}.

%% file: algorithm.tex

\newcommand{\SampleColor}{\textnormal{\ensuremath{\algs{SampleColor}}}\xspace}

\section{A Fully Dynamic Algorithm for $(\Delta+1)$ Coloring}\label{sec:algorithm}

We present our fully dynamic algorithm for $(\Delta+1)$ coloring and its key subroutines and lemmas in this section. We postpone the analysis of the key subroutines to the subsequent sections. 

\begin{theorem}\label{thm:delta+1}
	There exists a randomized fully dynamic algorithm for maintaining a $(\Delta+1)$ coloring against adaptive adversaries with $O(\sqrt{n}\cdot \log^3{n})$ amortized update time with high probability. 
\end{theorem}

\Cref{thm:delta+1} formalized \Cref{res:main}.
When proving \Cref{thm:delta+1}, we assume $\Delta=\Omega(\sqrt{n}\cdot \log^{3} n)$ as otherwise the proof follows trivially from the naive $O(\Delta)$ update time algorithm.

\subsection{The Sparse-Dense Decomposition}\label{sec:sparse-dense-decomp}
 
Our algorithm uses the sparse-dense decomposition that was introduced by \cite{Reed98} and extended
in~\cite{HarrisSS16,AssadiCK19} (we will also use an idea due to~\cite{FlinH25} in their refinement of this definition called ``sparser-denser decomposition'' but not as a black-box). 
We first need the following definitions for any graph $G=(V,E)$ and parameter $\eps \in (0,1)$: 
\begin{itemize}
	\item A vertex $v$ is \textbf{$\bm{\eps}$-sparse} iff $G[N(v)]$ contains at most $\binom{\Delta}{2} - \eps  \Delta^2$ edges; 
	\item A set $K \subseteq V$ of vertices is an \textbf{$\eps$-almost-clique} iff 
	\begin{enumerate}[label=$(\roman*)$]
		\item $(1-\eps) \cdot \Delta \leq \card{K} \leq (1+\eps) \cdot \Delta$; and, 
		\item every $v \in K$ has $\card{N(v) \cap K} \geq (1-\eps) \cdot \Delta$. 
	\end{enumerate}
\end{itemize}

A sparse-dense decomposition is then defined as follows. 
\phantomsection
\begin{Definition}[Sparse-dense Decomposition~\cite{Reed98,HarrisSS16,AssadiCK19}]\label{def:sparse-dense-decomp}
    Let $G$ be a graph with maximum degree $\Delta$ and $\eps\in (0,1)$ be a parameter. A sparse-dense decomposition of $G$ with parameter $\eps$ is a 
    partition of vertices of $G$ into sets $\Vsparse \sqcup K_1 \sqcup \cdots \sqcup K_k$  such that 
    \begin{enumerate}
        \item every vertex $v\in \Vsparse$ is $\Omega(\eps^2)$-sparse (for some hidden absolute constant), and
        \item for every $i\in [k]$, $K_i$ is an $\eps$-almost-clique. 
    \end{enumerate}
\end{Definition}

\noindent
For any $\eps$-almost-clique $K$ and a vertex $v \in K$, define: 
\begin{itemize}
	\item $\OUT(v)$: the set of neighbors of $v$ outside $K$, i.e., $\OUT(v) := N(v) \setminus K$; 
	\item $\ANTIDegree(v)$: the set of vertices of $K$ that are not neighbor to $v$, i.e., $\ANTIDegree(v) := K \setminus N(v)$.
\end{itemize}
\noindent 
We further set 
	\[
		\out(v) := \card{\OUT(v)} \qquad \text{and} \qquad \antiDegree(v) := \card{\ANTIDegree(v)}. 
	\]
Notice that $\out(v) \leq \eps \Delta$ and $\antiDegree(v) \leq 2\eps \Delta$ for every $v \in K$. 
Finally, we define the averaged versions of these quantities in the almost-clique as 
\[
	\avgout(K) := \frac1{\card{K}} \sum_{v \in K} \out(v) \qquad \text{and} \qquad \avgantiDegree(K) := \frac1{\card{K}} \sum_{v \in K} \antiDegree(v). 
\]

A fully dynamic algorithm for maintaining a sparse-dense decomposition was presented in~\cite{BehnezhadRW25} and further refined in~\cite{FlinH25}. 

\begin{proposition}[\!\!{\cite[Theorem~2]{BehnezhadRW25}}]\label{prop:maintaining-decomposition}
   Let $\eps > 0$ be sufficiently small. 
    There is a fully dynamic algorithm with $O(\eps^{-4}\log n)$ amortized update time that with high probability maintains 
    a sparse-dense decomposition of $G$ with parameter $\eps$. Specifically, the algorithm maintains: 
    \begin{itemize}
	\item an array that for every $v \in V$ indicates if $v$ is in $\Vsparse$ and if not, stores $K(v) \in [k]$ which is the index of the almost-clique $v$ belongs to;      
       	\item for each almost-clique $K_i$ for $i \in [k]$ and $v \in K_i$, two separate lists that contain $\ANTIDegree(v)$ and $\OUT(v)$; the algorithm also stores $\avgout(K_i),\avgantiDegree(K_i)$ for $K_i$.
    \end{itemize}
   The total number of changes to $\ANTIDegree(v),\OUT(v)$ across all $v \in V$ is amortized $O(1/\eps^4)$ per update. 
\end{proposition}

\cite[Theorem~2]{BehnezhadRW25} only proves the upper bound of $O(1/\eps^4)$ on changes to $\ANTIDegree(v)$ but their same exact analysis (Lemma 39 and Corollary 5) also
imply the desired bound on $\OUT(v)$ as well\footnote{If the almost-cliques do not change after an update, then clearly, at most two changes across all $\OUT$-sets can happen (for the endpoints of the updated edge). However, 
it is also possible that a vertex joins or leaves an almost-clique (leading to $\Theta(\Delta)$ changes) or an entire almost-clique forms or collapses (leading to $\Theta(\Delta^2)$ changes);~\cite{BehnezhadRW25} bounds the number of 
these changes as well in their analysis to bound the update time and the same analysis also applies verbatim to changes to $\OUT$-sets.}.

\subsubsection*{Classification of Almost-Cliques} 

For our purpose, we need to further classify the almost-cliques in the decomposition (similar ideas are used frequently in the literature, e.g. in~\cite{AssadiKM22,FischerHM23,FlinH25}). 

\begin{Definition}\label{def:classification}
    Let $G$ be a graph with maximum degree $\Delta$ and $\eps\in (0,1)$ and $\eta = o(\eps^2\Delta)$ be parameters. 
    Suppose $\Vsparse \sqcup K_1 \sqcup \cdots \sqcup K_k$ is a sparse-dense decomposition of $G$ with parameter $\eps$. 
    We further define a \textbf{class} for each almost-clique $K$ in this decomposition\footnote{It is possible for an almost-clique to belong to more than one class but it will belong to at least one class.}: 
    \begin{itemize}
	\item $K$ is \textbf{anti-heavy} iff $\avgantiDegree(K) \geq \eta/\eps^2$;
	\item $K$ is \textbf{out-heavy} iff it $\avgantiDegree(K) \leq 2\eta/\eps^2$ and $\avgout(K) \geq 4\eta/\eps^2$; 
	\item $K$ is \textbf{out-light} iff $\avgout(K) \leq 6\eta/\eps^2$.   
\end{itemize}
	We say a \emph{partition} $\KKah \sqcup \KKoh \sqcup \KKol$ of almost-cliques is an \textbf{$\eta$-classification} iff every 
	$K \in \KKah$ is anti-heavy, $K \in \KKoh$ is out-heavy, and $K \in \KKol$ is out-light. 
\end{Definition}

As stated in~\Cref{def:classification}, an almost-clique can belong to more than one class by definition; this (small) degree of overlap in the definitions is required for the dynamic handling of almost-cliques when they transition between these different classes. However, inside a fixed $\eta$-classification, each almost-clique belongs to exactly one of the three sets. 

\cite{FlinH25} proves that an out-heavy almost-clique $K$ is $\gamma$-sparse for some $\gamma$ proportional to $\avgout(K)$. We show a slight generalization to also include dependence on $\out(v)$ separately. The following lemma is a simple corollary of {\cite[Lemma 3.6]{FlinH25}} and its proof is postponed to \Cref{sec:lem:OH-is-out-sparse}. 

\begin{lemma}\label{lem:OH-is-out-sparse}
	Fix any sparse-dense decomposition and its classification. Any vertex $v$ in an out-heavy almost-clique is $\gamma$-sparse for
	\[
		\gamma = \Omega(\eps^2) \cdot \paren{\frac{\avgout(K) + \out(v)}{\Delta}}.
	\] 
\end{lemma}

We have the following simple corollary of~\Cref{prop:maintaining-decomposition} and~\Cref{def:classification} that allows us to maintain an $\eta$-classification of almost-cliques as well.

\begin{lemma}\label{lem:decomp}
    Let $\eps > 0$ be small and $\eta = o(\eps^2 \Delta)$. 
    There is a fully dynamic algorithm that with high probability maintains 
    a sparse-dense decomposition of $G$ for $\eps$ together with its $\eta$-classification. 
    The algorithm provides all guarantees of~\Cref{prop:maintaining-decomposition} and additionally the number of almost-cliques that change from one set of the $\eta$-classification to another set is amortized $O(\frac{1}{\eps^2 \eta \Delta})$ per update. 
\end{lemma}

We emphasize that in~\Cref{lem:decomp}, the amortized number of times per update that almost-cliques change a class is (considerably) sub-constant and that is needed in our analysis (although, we only need this bound to be $O(1/\eta)$ for the 
rest of the analysis). 
We prove this lemma in \Cref{sec:lem:decomp}.

\subsection{Setup and Invariants}

Throughout the algorithm, we use the following parameters:
\begin{equation} 
\begin{aligned}
	\delta &:= \text{a small constant used to suppress various absolute constants appearing in the proofs,} \\ 
	\eps &:= \text{a small constant chosen for the decomposition } (\eps>\delta),\\
	\eta &:= \text{an integer in $o(\eps^2 \Delta)$ to be chosen later (it will be set to $\approx \sqrt{n}$)}. 
\end{aligned}\label{eq:parameters}
\end{equation}
The goal is to maintain a proper $(\Delta+1)$ coloring $\varphi: V \rightarrow [\Delta+1]$ of the graph $G_t$ at each time step $t$. For simplicity, we let
$\varphi$ to use $\perp$, namely, uncolor a vertex, at \emph{intermediate} states but will ensure it does not have blank vertices before the next update. 
Moreover, the algorithm maintains
a sparse-dense decomposition $(\Vsparse,K_1,\ldots,K_k)$ of $G_t$ plus an $\eta$-classification of its almost-cliques using~\Cref{lem:decomp}. We refer to this subroutine as $\Decomposition$ for ease of reference. 

Throughout, we use $V(\KKah)$ to denote the vertices in the almost-cliques of the class $\KKah$ (similarly $V(\KKoh)$ and $V(\KKol)$ for $\KKoh$ and $\KKol$). Our algorithm treats different types of vertices differently: vertices in $\Vsparse \cup V(\KKoh)$ in one category, and $V(\KKah) \cup V(\KKol)$ in another. Following~\cite{FlinH25}, we refer to former vertices as \textbf{sparser} vertices, denoted by $S$, and to later vertices as \textbf{denser} vertices\footnote{Technically speaking,~\cite{FlinH25} considers $V(\KKah)$ as sparser vertices but for us they are denser.}, denoted by $D$. 
We now go over the setup for each of sparser and denser vertices separately.

\subsubsection*{Sparser Vertices $S:= \Vsparse \cup \KKoh$}
The algorithm maintains a proper partial coloring $\colS: V \rightarrow [\Delta+1] \cup \set{\perp}$ that forms the backbone of our approach. We maintain the following two invariants for  this coloring. 

\begin{coloredbox}{ForestGreen}
\begin{invariant}\label{inv:sparser}
	For any $v \in D$, $\colS(v) = \perp$. 
\end{invariant}
\begin{invariant}\label{inv:extend}
	For any $v \in V$ with $\colS(v) \neq \perp$, we have $\varphi(v) = \colS(v)$. 
\end{invariant}
\end{coloredbox}


The more important aspect of $\colS$ is the notion of a \textbf{flaw} for vertices. We define the following flaws for vertices under the coloring $\colS$; a vertex is said to have a specific flaw if the condition of the flaw is satisfied for the vertex. 

\begin{coloredbox}{cyan}
\begin{flaw}\label{flaw:sparse}
	$v \in \Vsparse \quad {\wedge} \quad \deg_S(v) \geq \Delta- \delta \cdot \eps^2 \cdot \Delta \quad {{\wedge}} \quad \slack{v}{\colS} < \eps^2\cdot \delta^6 \cdot \Delta.$
\end{flaw}

\begin{flaw}\label{flaw:out-heavy}
$v \in V(\KKoh) \quad {{\wedge}} \quad \deg_S(v) \geq \Delta - \delta \cdot \eps^2 \cdot \eta \quad {{\wedge}} \quad \slack{v}{\colS} <\eps^2\cdot \delta^6 \cdot \eta. $
\end{flaw}
\end{coloredbox}

We should right away note that our algorithm does \emph{not} aim to keep the coloring $\colS$ flawless: it is entirely possible that throughout the algorithm, many vertices have one or more of the flaws during a long sequence of updates. 
Instead, our algorithm fixes these flaws on a need-based basis as will become clear later. Thus, it is best to see flaws (or rather their absence) separate from the invariants that are maintained at all times. 

Another remark about these flaws is that some of the conditions in them are ``out of our hand'', e.g., whether or not $v \in \Vsparse$ is dictated by the graph and not what our algorithm does; these conditions 
say which vertices are relevant for the flaw to begin with. But, the final condition of each flaw is a function of the coloring $\colS$ which we can control. 

\subsubsection*{Vertices in $D := V(\KKah) \cup V(\KKol)$}

We also recall the notion of \emph{colorful anti-matchings} defined by~\cite{AssadiCK19} and used in many graph coloring algorithms 
that rely on sparse-dense decompositions, e.g., in~\cite{FlinGHKN24,BehnezhadRW25,FlinH25}. 

\begin{Definition}\label{def:colorful-anti-matching}
Given an almost-clique $K$ and a coloring $\varphi$, a \textbf{colorful anti-matching} $M$ is a matching over anti-edges inside $K$ such that for each anti-edge $(u,v) \in M$, both
	$u$ and $v$ are colored the same $\varphi(u)=\varphi(v)$ and this color does not appear on any other vertex of $K$. 
\end{Definition}

\begin{figure}[htbp]
    \centering
    \includegraphics[width=0.4\textwidth]{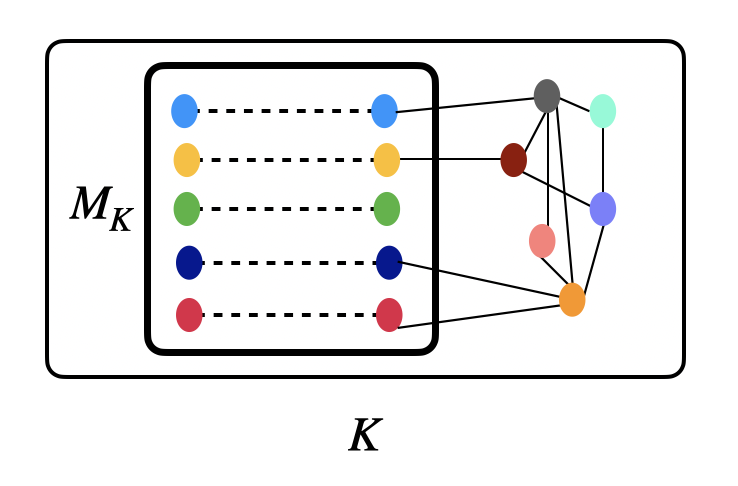}
    \caption{An illustration of a colorful anti-matching $M_K$ inside an almost-clique $K$. Colors in $M_K$ cannot be used by other vertics in $K$.}
    \label{fig:anti-matching}
\end{figure}

Colorful anti-matchings allow the algorithm to save extra colors on almost-cliques---at a rate that is proportional to their number of anti-edges---which will be crucial for coloring the rest of the almost-clique. The next invariant we maintain is regarding the size of colorful anti-matchings in each almost-clique in $\KKah \cup \KKol$.  Similar invariants have been maintained in \cite{BehnezhadRW25,FlinH25}.

\begin{coloredbox}{ForestGreen}
\begin{invariant}\label{inv:dense}
	For any almost-clique $K \in \KKah \cup \KKol$, the algorithm maintains a colorful anti-matching $M_K$ of size $\card{M_K} = \floor{8 \cdot \avgantiDegree(K)}$ using the output coloring $\varphi$. 
\end{invariant}
\end{coloredbox}

We list our key notation and parameters used in our algorithm for ease of reference in~\Cref{tab:notation}.

\input{Table-symbols.tex}

\subsection{Key Subroutines}  

We now list our key subroutines and their objectives. Their implementation and analysis will appear in subsequent sections. 
Recall that the algorithm will be maintaining a sparse-dense decomposition and its classification using $\Decomposition$ (\Cref{lem:decomp}) in the background and the subroutines below may use this decomposition
explicitly or implicitly. 

We partition the subroutines into ``old'' ones that are closely related to or borrowed from~\cite{BehnezhadRW25,FlinH25} and the ``new'' ones that entirely deviate
from those work and are new to our work.  


\subsubsection*{Old Subroutines} 

Our first subroutine is for maintaining~\Cref{inv:dense} after updates to the decomposition.
\begin{ourbox}
	$\ExpandMatching(K)$: Given an almost-clique $K \in \KKah \cup \KKol$ that violates~\Cref{inv:dense}, increases size of the colorful anti-matching $M_K$ until $K$ respects~\Cref{inv:dense}
	again. The subroutine may recolor or uncolor some vertices of $K$ in coloring $\varphi$ and will return a (possibly empty) list of vertices in $K$ which are now uncolored. 
\end{ourbox}


\begin{lemma}\label{lem:anti-matching}
	Let $K$ be an almost-clique in the decomposition violating~\Cref{inv:dense}. The subroutine $\ExpandMatching(K)$ increases
	size of $M_K$ until it again satisfies~\Cref{inv:dense} in $O(b \cdot n/\Delta)$ expected time where $b$ is the number of anti-edges added to $M_K$. 
	This subroutine may recolor vertices of $K$ in $\varphi$, while ensuring properness of the (partial) coloring $\varphi$, and creates $O(b)$ new uncolored vertices in $K$ which it also returns. 
\end{lemma}

\medskip
The next subroutine is for coloring denser vertices. 

\begin{ourbox}
	$\ColorDense(v)$: Given a denser vertex $v \in D$ with $\varphi(v) = \,\perp$ in some almost-clique $K = K(v)$, assuming~\Cref{inv:dense} holds for $K$, it will color $v$ by recoloring some vertices of $K$ in $\varphi$ (without uncoloring
	any vertex at the end).  
\end{ourbox}

\begin{lemma}\label{lem:dense}
	Let $v \in D$ be a denser vertex with $\varphi(v) = \perp$ such that~\Cref{inv:dense} holds for the almost-clique $K(v)$ that $v$ belongs to. The subroutine $\ColorDense(v)$ extends the coloring $\varphi$ to color
	$v$ as well in expected $O(n/\Delta + n/\eta+ \eta)$ time by recoloring vertices inside $K$, while ensuring properness of the (partial) coloring $\varphi$, without creating new uncolored vertices. 
\end{lemma}

We also have a subroutine for coloring sparser vertices, which is allowed to fail in certain cases. 

\begin{ourbox}
	$\ColorSparse(v)$: Given a sparser vertex $v \in S$ with $\varphi(v)=\perp$, it either colors $v$ in $\varphi$ (not $\colS$) by potentially uncoloring some denser vertices, or outputs \Fail. 
\end{ourbox}

\begin{lemma}\label{lem:attempt}
    For a sparser vertex $v \in S$, subroutine $\ColorSparse(v)$ takes $O(n\log{n}/\eta)$ time in expectation. If $v$ does not have any of~\Cref{flaw:sparse} and \Cref{flaw:out-heavy}
	under $\colS$, $\ColorSparse(v)$ with high probability colors $v$ successfully in $\varphi$ without returning \Fail. Moreover, this subroutine uncolors at most $O(1)$ many denser vertices in expectation.
\end{lemma}

As stated earlier, subroutines $\ExpandMatching$, $\ColorDense$, and $\ColorSparse$ and their corresponding lemmas are closely related to the work of~\cite{FlinH25} (as well as~\cite{BehnezhadRW25})  and almost (but not entirely) follow from the results there. Since we will be applying them with different guarantees and parameters, we provide their proofs in \Cref{sec:sparse-subroutines,sec:dense} for completeness. 

\subsubsection*{New Subroutines} 

The first new subroutine we need is a way to rather quickly color most uncolored vertices in an almost-clique $K \in \KKoh$. This is one of the parts that we deviate from
the previous work in~\cite{BehnezhadRW25,FlinH25} as there is no analogue of this algorithm in those work. 

\begin{ourbox}
	$\ColorOH(K)$: Given a partially colored almost-clique $K \in \KKoh$, it will color all but ${4\cdot \avgout(K)}$ uncolored vertices of $K$ in $\varphi$ without recoloring or uncoloring any vertex. 
\end{ourbox}

\phantomsection 
\begin{lemma}\label{lem:colorOH}
	Let $K \in \KKoh$ be any out-heavy almost-clique in the decomposition with at least $4\Delta/5$ uncolored vertices in $\varphi$. Calling $\ColorOH(K)$ extends the coloring $\varphi$ to color all but ${4\cdot \avgout(K)}$ vertices of $K$ in expected $O(n\log{n})$ time without recoloring or uncoloring any vertex. 
\end{lemma}

Up until now, all subroutines we have are more sublinear time algorithms than dynamic ones in that they do not need
to maintain any explicit state and we can analyze their runtime for each invocation directly (although their input is updated dynamically by 
our main algorithm and their outputs will be used dynamically by other subroutines). 


We now get to our final subroutine which is the heart of our new algorithm. This subroutine is called to color sparser vertices
when $\ColorSparse$ fails to do so. This part entirely departs from prior work in~\cite{BehnezhadRW25,FlinH25} that relied on execution of 
sublinear time algorithms such as~\cite{AssadiCK19,AssadiY25} and instead handles this part dynamically.  
\begin{ourbox}
	$\UpdateSparseColor(v)$: Given a sparser vertex $v \in S$ with $\colS(v) = \perp$, it updates the coloring in $N(v)$ by recoloring and uncoloring some vertices in $\colS$ and returns the uncolored vertices. 
\end{ourbox}

We will not provide a standalone lemma for update time and recourse of $\UpdateSparseColor$ as it is closely intertwined with the overall performance of the entire dynamic algorithm. 
However, the following lemma captures the correctness guarantees we need from this subroutine. 
\phantomsection
\begin{lemma}\label{lem:update-sparser}
	Let $v \in S$ with $\colS(v) = \perp$.
    If $\UpdateSparseColor(v)$ terminates, 
    it has recolored or uncolored some vertices in $\colS$ and additionally uncolored some vertices in $\varphi$ that it will return, while ensuring properness of the partial colorings $\colS$ and $\varphi$. This subroutine does not color denser vertices in $\colS$ and ensures $\colS$ and $\varphi$ agree on recolored vertices (maintains~\Cref{inv:extend,inv:sparser}).  
\end{lemma}

We mention one last subroutine used in $\UpdateSparseColor$ and our dynamic algorithm. 
\begin{ourbox}
	$\SampleColor(v)$: Given a sparser vertex $v \in S$, it will sample a color $c_v$ from a to-be-determined distribution and set $\colS(v)=c_v$ if possible and otherwise $\colS(v)=\perp$. 
	This subroutine might also uncolor some vertices in $\set{v} \cup N(v)$ and returns them.
\end{ourbox}
\begin{lemma}\label{lem:SampleColor}
	For any vertex $v\in S$, the subroutine $\SampleColor(v)$ takes expected $O(n/\Delta)$ time to sample a color for $v$ and it uncolors $O(1)$ vertices in expectation.
\end{lemma}
 

The last comment is a \emph{deterministic} upper bound on the \emph{worst-case} runtime of these subroutines.  
We need this worst-case runtime bound to prove the high probability argument of \Cref{thm:delta+1}.

\begin{lemma}\label{lem:worstcase-subroutines}
	The worst-case runtime of each subroutine \ExpandMatching, \ColorDense, \ColorSparse, and \ColorOH is $O(m)$, where $m$ is the number of edges in the graph at the time the subroutine is called.
\end{lemma}



\subsection{The Dynamic Algorithm} 

We can now present our dynamic algorithm. Note that at the beginning the graph is empty and thus, any coloring is proper. We use the following initialization to simplify the exposition.

\begin{ourbox}
	\label{alg:initial}
	\textbf{Initialization phase:} For all $v\in V$, sample $c\in [\Delta+1]$ and let $\varphi(v)=c$. Also, with probability $\delta$, let $\colS(v)=c$; otherwise let $\colS(v)=\perp$.
\end{ourbox}

We can now describe the algorithm (without getting too much into implementation details). 

\begin{Algorithm}\label{alg:delta+1}
Called when $e_t=(u_t,v_t)$ is updated in $G_{t-1}$ to obtain $G_t$. 

\begin{enumerate}[leftmargin=15pt]
\item\label{line:F} Let $B$ be the set of blank vertices, initialized to $B=\emptyset$. In lexicographical order of $v_t$ and $u_t$, call $\SampleColor$ for each one that is in $S$.
If $e_t=(u_t,v_t)$ is inserted, $u_t, v_t\in D$, and $\varphi(u_t)=\varphi(v_t)$ set $\varphi(u_t) = \perp$. Add all uncolored vertices to $B$.

\item\label{line:decomposition} Update the decomposition and classification by calling \Decomposition and do:
\begin{enumerate}
	\item For any $v\in V$ previously in $S$ that now joins $D$, set $\colS(v) = \varphi(v) = \perp$ and add $v$ to $B$.
	\item For any $v\in V$ previously in $D$ that now joins $S$, call $\SampleColor(v)$ and if it gets color $\perp$ add it to $B$.
\end{enumerate}  
\item\label{line:sparse} While there is a sparser vertex $v \in B \cap S$: run $\ColorSparse(v)$; if it returns \Fail, run $\UpdateSparseColor(v)$; add all uncolored vertices to $B$. 
\item\label{line:expand} For any almost-clique $K \in \KKah \cup \KKol$ violating~\Cref{inv:dense}: run $\ExpandMatching(K)$ and add uncolored vertices to $B$. 
\item\label{line:dense} While there is a denser vertex $v \in B \cap D$: run $\ColorDense(v)$. 
\end{enumerate}
\end{Algorithm}

The following lemma ensures the correctness of~\Cref{alg:delta+1}. 

\begin{lemma}\label{lem:delta+1-correct}
Assuming that \Cref{alg:delta+1} terminates, then with high probability, every update, it maintains a proper $(\Delta+1)$-coloring of the underlying graph.
\end{lemma}
\begin{proof}
	After Line~\eqref{line:F}, $\varphi$ is a proper partial coloring of $G_t$ and all blank vertices are in $B$.
	In Line~\eqref{line:decomposition}, \Decomposition correctly updates the sparser-denser decomposition 
	and its classification with high probability by~\Cref{lem:decomp} and also, it ensures \Cref{inv:sparser} is not violated.
	Line~\eqref{line:sparse} maintains properness of the coloring by~\Cref{lem:attempt} and \Cref{lem:update-sparser} and only finishes when there are no sparser blank vertices left while maintaining \Cref{inv:sparser,inv:extend}. 

	Then, running Line~\eqref{line:expand}, by~\Cref{lem:anti-matching}, ensures that~\Cref{inv:dense} holds for all almost-cliques; this is because once an almost-clique has \Cref{inv:dense} satisfied it will remain so in this for-loop as $\ExpandMatching(K)$ does not modify vertices outside $K$ for each almost-clique $K$. Moreover, by~\Cref{lem:anti-matching} we know this step does not uncolor sparser vertices. 
	
	Finally, moving into Line~\eqref{line:dense} knowing \Cref{inv:dense} is maintained,   all denser vertices get properly colored by~\Cref{lem:dense}. Note that \ColorDense does not uncolor any vertex. Therefore, by the time the algorithm finishes, all vertices are colored and the coloring is proper. 
\end{proof}

The main part of the analysis is to bound the update time of~\Cref{alg:delta+1} and in particular \UpdateSparseColor in the following two lemmas. The proofs appear in the next section. 

\begin{lemma}\label{lem:recourse-updatesparse-whp}
    The total number of uncolored vertices caused by subroutine \UpdateSparseColor over $t$ updates is at most $O(t \log^2 n)$ with high probability.
\end{lemma}

\begin{lemma}\label{lem:runtime-UpdateSparseColor}
    Throughout running \Cref{alg:delta+1} over $t$ updates, with high probability the time spent on subroutine \UpdateSparseColor is $O(t\cdot n\log^4 n/\eta)$ time. 
\end{lemma}


\subsection{Proof of \Cref{thm:delta+1}}
We need the following concentration inequality which generalizes Chernoff bound by allowing the distribution in
one time step to depend on the outcomes of previous steps (see~\Cref{app:KQ-inequality}). 

\begin{proposition}[\!\!{\cite[Corollary~11]{kuszmaul2025}}]\label{prop:KQ-inequality}
	Suppose Alice constructs a sequence of random variables $X_1, \ldots, X_n$, with $X_i\in [0, c]$ using the following iterative process. For each $i\in [n]$, once the outcome of $X_1, \ldots, X_{i-1}$ are determined, Alice then selects the probability distribution $\mathcal{D}_i$---as a function of the outcomes of $X_1,\ldots,X_{i-1}$---and $X_i$ is drawn from $\mathcal{D}_i$. The only constraint on Alice is that
	if we define $a_i := \expect{X_i\mid \mathcal{D}_i}$ and $A = \sum_i a_i$, then, for some $\gamma \in (0,1)$ and $\mu > 0$, $\Pr\paren{A \geq \mu} \leq \gamma$; 
	that is, the sum of the means of the probability distributions $\mathcal{D}_1, \ldots, \mathcal{D}_n$ must be at most $\mu$ with probability $1-\gamma$. If $X=\sum_i X_i$, then for any $\delta>0$, we have
	\[
	\Pr\paren{X\geq (1+\delta) \cdot \mu}\leq \gamma + \exp\paren{-\frac{\delta^2 \cdot \mu}{(2+\delta) \cdot c}}.
	\]
\end{proposition}

In the following lemma we bound the runtime of the remaining subroutines over all updates.

\begin{lemma}\label{lem:whp-runtime-other-subroutines}
    Throughout running \Cref{alg:delta+1} over $t$ updates, with high probability, the time spent and number of uncolored vertices generated by subroutine: 
    \begin{enumerate}
		\item\label{item-attempt} \ColorSparse is $O(t\cdot n\cdot \log^3 n /\eta)$ time and $O(t \cdot \log^2{n})$ vertices; 
		\item\label{item-sampleColor} \SampleColor is $O(t \cdot n\cdot \log n/\Delta)$ time and $O(t\log{n})$ vertices;
        \item\label{item-expand-matching}\ExpandMatching is $O(t\cdot \log^2 n \cdot n/\Delta)$ time and $O(t\log^2{n})$ vertices;
        \item\label{item-dense} \ColorDense is $O(t\cdot \log^2 n\cdot (n/\Delta+n/\eta+\eta))$ time and zero vertices. 
    \end{enumerate}
\end{lemma}
\begin{proof}
Proof of all the listed claims follows from similar structure and arguments. We present the proof of (\ref{item-attempt}) completely and 
postpone proof sketches of the remaining ones to~\Cref{sec:lem:whp-runtime-other-subroutines} to avoid repetition.
Also, as a notation, we let $m$ denote the \emph{largest} number of edges at any point during the $t$ updates. Clearly, we have $m \leq t$. 

\paragraph{Proof of (\ref{item-attempt}):} 
	Let $s_a$ be the number of times \ColorSparse has been called in $t$ updates. 
	For each $i\in [s_a]$, let $X_i$ be the random variable showing the time spent on the $i$-th call to \ColorSparse and let $Y_i$ be the random variable showing the number of uncolored vertices caused by the $i$-th call to \ColorSparse.

    By \Cref{lem:attempt} we know that $\expect{X_i}=O(n\log n/\eta)$ and $\expect{Y_i}=O(1)$ for each $i$. 
	By \Cref{lem:worstcase-subroutines} we know $X_i=O(m)$. Also, the number of uncolored vertices is bounded by $\min\set{m,n}$, that is $Y_i\leq \min\set{m,n}$. Let $X=\sum_{i\in [s_a]}X_i$ and $Y=\sum_{i\in [s_a]}Y_i$. By linearity of expectation, we have 
	\[
	\expect{X}=s_a\cdot O(n\log n/\eta) \quad \text{and} \quad \expect{Y}=s_a\cdot O(1).
	\]
	\ColorSparse is invoked for a new uncolored sparser vertex. An uncolored sparser vertex is caused either by \SampleColor after an update (bounded by $O(t)$) or by \UpdateSparseColor (bounded by $O(t\cdot \log^2 n)$ by \Cref{lem:recourse-updatesparse-whp}). Thus, with high probability $s_a=O(t\cdot \log^2 n)$.

	For each $i\in [s_a]$ define $\mathcal{D}_i$ as the distribution from which $X_i$  is drawn from. We know that this distribution depends on the output of the algorithm at that point as the updates are given by an adaptive adversary. Thus, $\mathcal{D}_i$ depends on the value of $X_1,\ldots, X_{i-1}$.  

	These random variables follow the structure of \Cref{prop:KQ-inequality} with parameters $c=m$ and $\gamma=n^{-\Omega(1)}$, $\mu=O(t\cdot n\log^3 n/\eta)$ and, $\delta=1$. Therefore, by \Cref{prop:KQ-inequality} we get
	\begin{align*}
	\Pr\paren{X = \omega\paren{\frac{ t\cdot n\cdot \log^3 n}{\eta}}}
	&\leq n^{-\Omega(1)}+\exp\paren{-\frac{ t\cdot n\cdot \log^3 n}{m\cdot \eta}} \tag{by \Cref{lem:recourse-updatesparse-whp} we have $\gamma=n^{-\Omega(1)}$}	\\
	&\leq  n^{-\Omega(1)}. \tag{as $\eta=o(n)$ and $t=\Omega(m)$}
	\end{align*}
	Thus, the total time spent on this subroutine is $O(t\cdot n\log^3 n/\eta)$ with high probability. Similarly, we can use the same analysis to bound the number of uncolored vertices by \Cref{prop:KQ-inequality} with parameters $c=\min\set{m,n}$, $\gamma=n^{-\Omega(1)}$, $\mu=O(t\cdot \log^2 n)$ and $\delta=1$. We get
		\begin{align*}
	\Pr\paren{Y = \omega\paren{t\cdot \log^2 n}}
	&\leq n^{-\Omega(1)}+\exp\paren{\frac{t\cdot \log^2 n}{\min\set{n,m}}} \tag{by \Cref{lem:recourse-updatesparse-whp} we have $\gamma=n^{-\Omega(1)}$}	\\
	&\leq  n^{-\Omega(1)}. \tag{as $\min\set{n,m} \leq t$}
	\end{align*}
	Thus, with high probability the total number of uncolored vertices caused by \ColorSparse over $t$ updates is $O(t\cdot \log^2 n)$, concluding the proof of this part. 
\end{proof}

\begin{proof}[Proof of \Cref{thm:delta+1}]
    If $\Delta=O(\sqrt{n} \cdot \log^{3} n)$ we  run the greedy $O(\Delta)$ update time algorithm. 
    For $\Delta=\omega(\sqrt{n} \cdot \log^{3} n)$, we set $\eta=\sqrt{n} \cdot {\log n}$ and run \Cref{alg:delta+1}. 

    By \Cref{lem:decomp}, maintaining the decomposition has $O(\log n)$ amortized update time with high probability.
    By \Cref{lem:runtime-UpdateSparseColor,lem:whp-runtime-other-subroutines}, the time spent on  remaining subroutines over $t$ updates is
    \[
    O\paren{t\cdot \log^2 n\cdot \paren{\frac{n\log^2 n}{\eta}+\frac{n}{\Delta}+\frac{n}{\eta}+\eta}} = O(t\cdot \sqrt{n}\cdot \log^{3}{n}), 
    \]
    time with high probability, concluding the proof. 
\end{proof}





%% file: Table-symbols.tex
\begin{table}[h]
\centering
\renewcommand{\arraystretch}{1.2}
\begin{tabular}{|c|p{12cm}|}
\hline
\textbf{Symbol} & \centering\textbf{Description}\arraybackslash \\ \hline

\multicolumn{2}{|c|}{\textbf{Sets}} \\ \hline

$\Vsparse$ &
set of $\Omega(\eps^2)$-sparse vertices
\\ 
$\KKah$ &
set of anti-heavy almost-cliques
\\
$\KKoh$ &
set of out-heavy almost-cliques
\\
$\KKol$ &
set of out-light almost-cliques
\\
$S$&
set of sparser vertices i.e. $\Vsparse\cup V(\KKoh)$
\\
$D$&
set of denser vertices i.e. $V(\KKah)\cup V(\KKol)$
\\
$M_K$&
the colorful anti-matching maintained for almost-clique $K$
\\
$\colS$&
the main \emph{partial} coloring maintained in the algorithm
\\
\hline

\multicolumn{2}{|c|}{\textbf{Parameters}} \\ \hline
$\delta$ &
a small constant used throughout the proof 
\\
$\eps$ &
a small constant $(\delta<\eps)$ used for as the parameter of decomposition
\\
$\eta$ &
an integer in $o(\eps^2\Delta)$ used to define the classification
\\
\hline
\end{tabular}
\caption{A summary of notation and parameters used in our algorithm.}
\label{tab:notation}
\end{table}

%% file: sparse.tex
\newcommand{\test}{\ensuremath{\algs{TestFlaw3}}\xspace}
\newcommand{\cost}{\textnormal{\ensuremath{\varbs{cost}}}}
\newcommand{\LOG}{\ensuremath{L}\xspace}

\newcommand{\INDEX}[2]{\ensuremath{\text{Index}_{#2}(#1)}\xspace}
\newcommand{\eventTwo}{\ensuremath{\mathcal{E}}\xspace}

\newcommand{\checkFlaw}{\ensuremath{\algs{CheckFlaw3}}\xspace}

\section{Update Time Analysis of our Dynamic Algorithm}\label{sec:sparser}

We now get to the main part of the proof where we use dynamic entropy compression for handling the update time over sparser vertices. In the following, 
we describe $\UpdateSparseColor$ in detail and then move to proving its overall recourse and runtime.

\subsection{The \UpdateSparseColor Subroutine}

The subroutine \SampleColor is used to color a single sparser vertex in $\colS$. 
Here we present the details of subroutine \SampleColor and prove \Cref{lem:SampleColor}.

\begin{Algorithm} $\SampleColor(v)$: With probability $\delta$ sample $c_v\in [\Delta+1]$ and go over color class $\mathcal{C}_{c_v}$ in $\colS$, and if there is no neighbor of $v$ in $\mathcal{C}_{c_v}$ set $\colS(v)=c_v$ (and $\varphi(v)=c_v$) and 
go over color class $\mathcal{C}_{c_v}$ in $\varphi$ now and uncolor (in $\varphi$) $\mathcal{C}_{c_v} \cap N(v)$. In all other cases, let $\colS(v)=\perp$. 
\end{Algorithm}
\begin{lemma*}[Restatement of \Cref{lem:SampleColor}]
    	For any vertex $v\in S$, the subroutine $\SampleColor(v)$ takes expected $O(n/\Delta)$ time to sample a color for $v$ and it uncolors $O(1)$ many vertices in expectation.
\end{lemma*}
\begin{proof}
    Runtime of \SampleColor depends only on the size of the color classes it searches.
    As $c_v$ is sampled uniformly at random, then $\Exp \abs{\mathcal{C}_{c_v}}=O(n/\Delta)$ in either of $\colS$ and $\varphi$, and therefore the expected runtime of \SampleColor is $O(n/\Delta)$. 
    Moreover, \SampleColor might only uncolor vertex $v$ and if not, it might uncolor neighbors of $v$ in $\mathcal{C}_{c_v}$ in $\varphi$. As there are at most $\Delta$ neighbors in $N(v)$ and $c_v$ is uniformly sampled from $[\Delta+1]$, then the expected number of vertices that might get uncolored in $\varphi$ is $O(1)$, concluding the proof.
\end{proof}

Since existence of any of \Cref{flaw:sparse}, \ref{flaw:out-heavy} for a vertex $v$ is a function of the coloring of its neighbors, to fix these flaws, we need to recolor the neighbors of $v$. This is the role of the following subroutine which is the key subroutine in \UpdateSparseColor.

\begin{Algorithm} $\Resample(v)$: For $u\in N(v)\cap S$ (in lexicographical order), call $\SampleColor(u)$.
\end{Algorithm} 

The following lemmas capture the main probabilistic guarantees when we sample colors for vertices using \SampleColor. The first lemma is a generalized version of bounding the probability of \Cref{flaw:all} remaining to hold for a vertex after a call to \Resample. The proof is a simple application of Chernoff bound and is postponed to \Cref{sec:lem:prob-flaw-sparse}.

\begin{lemma}\label{lem:prob-flaw-sparse}
  Fix a vertex $v\in S$ along with a set $B(u) \subseteq [\Delta+1]$ and a time $t_u$ associated to each vertex $u\in N_S(v)$. Suppose for every $u\in N(v)$, with probability $\delta$, we sample a color $c(u)$ and assign it to $u$ iff: $(1)$ the color $c(u) \notin B(u)$, and (2) no neighbor $w\in N_S(v)\cap N_S(u)$ with $t_w<t_u$ or with lower lexicographical order when $t_w=t_u$ has also sampled this color. All other
	vertices are marked blank. Let $\chi$ be the partial coloring function obtained this way. 
    Then,  
    \[
    \Pr(\Ncolor{v}{\chi}> 2\delta\Delta)\leq 2^{-{\delta\cdot \Delta}/{3}}.
    \]
\end{lemma}

The second lemma applies to vertices that have enough available colors (we will come back to this condition soon when defining~\Cref{flaw:all}). 
We prove that when vertices in a neighborhood of a vertex $v$ do not have many blocked colors, sampling their colors according to $\SampleColor$ can ensure $v$ does not satisfy \Cref{flaw:sparse} or \Cref{flaw:out-heavy} with high probability.

\begin{lemma}\label{lem:prob-slack-sparse}
  Let $\gamma\in (0, 1)$ be sufficiently small. Fix a $\gamma$-sparse vertex $v$ with $\deg_S(v)>\Delta-\delta\cdot\gamma\cdot \Delta$ along with a set $B(u) \subseteq [\Delta+1]$ and a time $t_u$ associated to each vertex $u\in N_S(v)$ such that $\abs{B(u)}\leq \Delta/4$.
	Suppose for every $u\in N(v)$, with probability $\delta$, we sample a color $c(u)$ and assign it to $u$ iff: $(1)$ the color $c(u) \notin B(u)$, and (2) no neighbor $w\in N_S(v)\cap N_S(u)$ with $t_w<t_u$ or with lower lexicographical order when $t_w=t_u$ has sampled this color. 
	Let $\chi$ be the partial coloring  obtained this way. 
	Then, 
    \[
    	\Pr\paren{\slack{v}{\chi} < \delta^6\cdot\gamma\cdot \Delta} \leq 2^{-\delta^5\cdot \gamma \cdot \Delta}. 
    \]
\end{lemma}

The proof of this lemma is detailed and technical, and is closely related to the proofs of similar results in prior work \cite{MolloyR97, MolloyR02,AssadiCK19}; to avoid distracting from the main argument, we defer the proof of \Cref{lem:prob-slack-sparse} to \Cref{sec:lem:prob-slack-sparse}.

We are now ready to describe the $\UpdateSparseColor$ algorithm itself. Recall that the goal of this algorithm is to ensure a given sparser vertex $v$ does not have \Cref{flaw:sparse} or \Cref{flaw:out-heavy}. 

\begin{Algorithm} 
$\UpdateSparseColor(v):$  Run $\ColorSparse(v)$, if it returns \Fail run $\Resample(v)$ and in case $v \in V(\KKoh)$ call $\ColorOH(v)$; otherwise terminate and return all uncolored vertices.
\end{Algorithm}

\phantomsection
\begin{lemma*}[Restatement of \Cref{lem:update-sparser}]
	Let $v \in S$ with $\colS(v) = \perp$.
    If $\UpdateSparseColor(v)$ terminates, 
    it has recolored or uncolored some vertices in $\colS$ and additionally uncolored some vertices in $\varphi$ that it will also return, while ensuring properness of the partial colorings $\colS$ and $\varphi$.
	Moreover, this subroutine does not color denser vertices in $\colS$ and ensures $\colS$ and $\varphi$ agree on recolored vertices (maintains~\Cref{inv:extend,inv:sparser}).   
\end{lemma*}
\begin{proof}
    The algorithm returns all vertices that receive a new color or get uncolored in $\colS$ during $\UpdateSparseColor$ and ensures $\varphi(u)=\colS(u)$ for each $u\in N(v)\cap S$, maintaining \Cref{inv:extend}.

    By construction, the subroutines \Resample only colors the sparser vertices, and the subroutines \ColorOH assigns colors to vertices only in $\varphi$, consequently maintaining \Cref{inv:sparser}. 
    
    Moreover, whenever a vertex is colored in $\colS$ in a call to \Resample, the color is from $\Avail{v}{\colS}$. By construction, all neighbors of $v$ that have the same color under $\varphi$ are subsequently uncolored, thereby preserving the properness of both $\colS$ and $\varphi$.
    The subroutines \ColorSparse and \ColorOH only recolor vertices in $\varphi$. Since $\colS$ and $\varphi$ agree on all vertices colored under $\colS$ and each subroutine either assigns an available color to a vertex or uncolor one endpoint of new monochromatic edges, these subroutines also preserve the properness of both colorings.
\end{proof}

The subroutines $\Resample$ and $\ColorOH$ are always invoked consecutively. Therefore, we analyze them jointly. We first prove that if $v\in V(\KKoh)$ then $\Resample(v)$ result in most vertices of $K(v)$ to get uncolored.

\begin{claim}\label{clm:resample-uncolors-in-OH}
For $v\in V(\KKoh)$ after $\Resample(v)$ at least $4\Delta/5$ vertices of $K(v)$ are uncolored in $\varphi$.
\end{claim}
\begin{proof}
When $\Resample(v)$ is called, every vertex $u\in N(v)$ receives the color $\perp$ with probability at least $1-\delta$, independently of
the others.

By properties of almost-cliques, $v$ has at least $(1-\eps)\Delta$ neighbors in $K(v)$. 
For $i\in[\deg_{K(v)}(v)]$, let $X_i$ be the indicator random variable that the $i$-th neighbor of $v$ in $K(v)$ becomes blank during $\Resample(v)$, and let
\[
X=\sum_{i\in[\deg_{K(v)}(v)]}X_i.
\]
We have
\[
\expect{X}= \sum_{i\in [\deg_{K(v)}(v)]} \Pr(X_i=1)\geq \deg_{K(v)}(v)\cdot (1-\delta)\geq (1-\delta)(1-\eps)\Delta.
\]
Since $\delta$ and $\eps$ are sufficiently small and $\eps >\delta$ we get $\expect{X}\geq (1-2\eps)\Delta\geq 9\Delta/10$. By Chernoff bound (\Cref{prop:chernoff}), we obtain 
\[
\Pr\paren{X\leq \frac{4\Delta}{5}}\leq \exp\paren{-\frac{\Delta}{200}}.
\]
Therefore, as without loss of generality we assumed $\Delta=\Omega(\sqrt{n})$ (see after \Cref{thm:delta+1}) , with high probability $X > 4\Delta/5$, concluding the proof.
\end{proof}

The following lemma allows us to bound the runtime of $\UpdateSparseColor$ based on the number of times it invokes the subroutine $\Resample$. 

\begin{lemma}\label{lem:recourse-runtime}
    A call to $\Resample(v)$ for  $v\in S$, together with a subsequent call to the subroutine $\ColorOH(K(v))$ when $v\in V(\KKoh)$, has expected running time $O(n\log n)$.
\end{lemma}
\begin{proof}
Each call to \Resample samples a color for at most $\Delta$ neighbors. To determine whether a sampled color is available, the algorithm searches the corresponding color class. Since the expected size of a color class is $O(n/\Delta)$, a call to \Resample takes expected time $O(n)$.


By \Cref{clm:resample-uncolors-in-OH}, $K(v)$ contains at least $4\Delta/5$ uncolored vertices after the call to $\Resample(v)$ with high probability. Therefore, the assumptions of \Cref{lem:colorOH} are satisfied, and the subsequent call to $\ColorOH(K(v))$ has expected running time $O(n\log n)$.

Combining the above two yields a total expected $O(n\log n)$ runtime.
\end{proof}

It is not enough to bound the runtime of \Resample and \ColorOH given that they uncolor vertices which need to be handled later. In order to account for these uncolored vertices, we need the following definition which we then show is a proxy to the number of vertices they uncolor. 


\begin{Definition}\label{def:cost-vertex}
    For each $v\in S$, define function $\cost(v)$ as below
    \[
    \cost(v)=
    \begin{cases}
    \Delta & v\in \Vsparse, \\
    \max\paren{\avgout(K(v)), \out(v)} & v\in V(\KKoh).
    \end{cases}
    \]
\end{Definition}

The following lemma allows us to bound the total number of uncolored vertices caused by $\UpdateSparseColor$ based on the number of times it invokes the subroutine $\Resample$. 

\begin{figure}[htbp]
    \centering
    \includegraphics[width=0.7\textwidth]{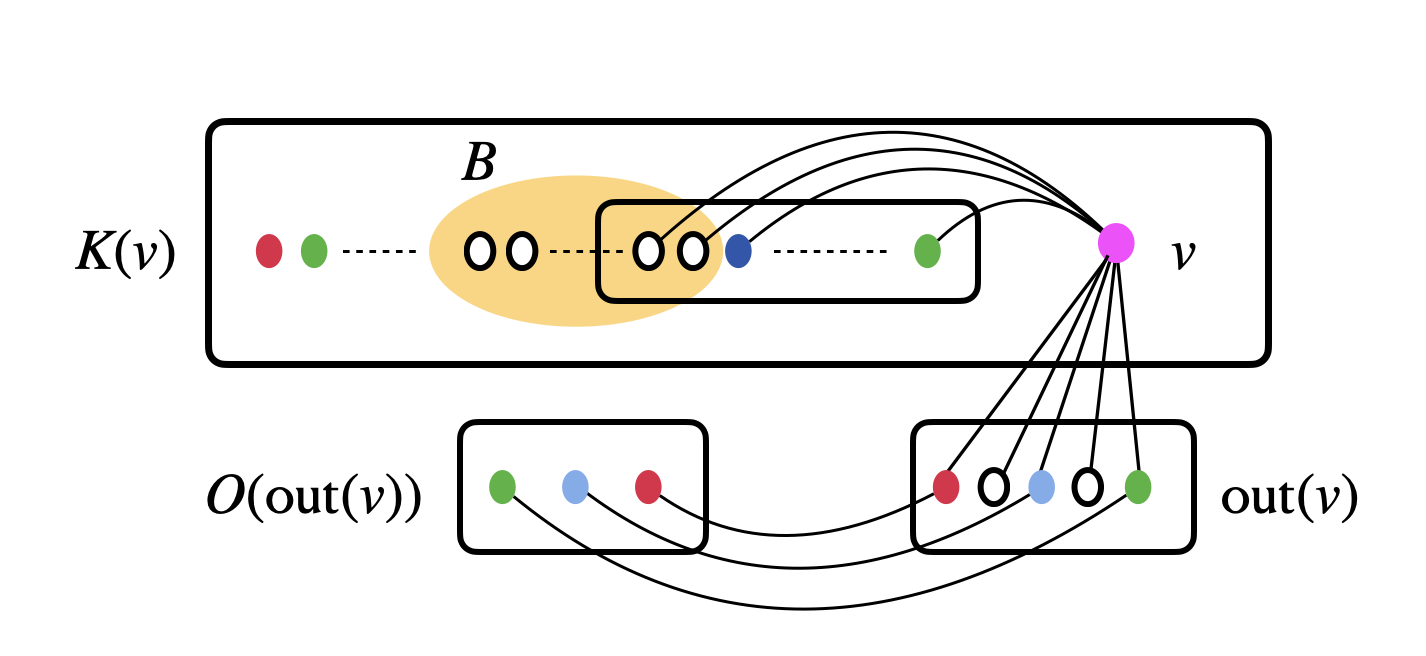}
    \caption{For an out-heavy vertex $v$, after running $\Resample(v)$ and then $\ColorOH(K(v))$ the possible uncolored vertices are either in (I) vertices uncolored by $\Resample(v)$ in $\OUT(v)$ or (II) the vertices in the neighborhood of the ones in $\OUT(v)$ that steal their colors, or (III) set $B$ which shows the vertices left uncolord by $\ColorOH(K(v))$ in $K(v)$.}
    \label{fig:recourse-colorOH}
\end{figure}

\begin{lemma}\label{clm:recourse-of-resample}
    A call to $\Resample(v)$ for  $v\in S$, together with a subsequent call to the subroutine $\ColorOH(K(v))$ when $v\in V(\KKoh)$, uncolors  $O(\cost(v))$ vertices in expectation.
\end{lemma}
\begin{proof}
If $v\in\Vsparse$, then $\Resample(v)$ may modify the colors of all neighbors of $v$. By \Cref{lem:SampleColor}, for each $u\in N_S(v)$, the call to $\SampleColor(u)$ inside $\Resample(v)$ might cause only $O(1)$ additional vertices to become uncolored in expectation. 
In addition, vertices in $N(v)$ may become uncolored directly by \Resample. Since $\deg(v)\leq \Delta=\cost(v)$, the number of vertices that can become uncolored is at most $O(\cost(v))$ in expectation. 

Suppose instead that $v\in V(\KKoh)$. By \Cref{clm:resample-uncolors-in-OH}, we have that after running $\Resample(v)$ the almost-clique $K(v)$ has at least $4\Delta/5$ many uncolored vertices with high probability and therefore satisfies the assumptions of \Cref{lem:colorOH}. Hence, by \Cref{lem:colorOH}, at most ${4\cdot \avgout(K(v))}$ vertices of $K(v)$ remain uncolored after $\ColorOH(K(v))$ terminates.

Furthermore, $\Resample(v)$ modifies the colors of vertices in
$N(v)$. Since $v$ has at most $\out(v)$ neighbors outside of $K(v)$,
the number of vertices outside $K(v)$ that may be uncolored directly by $\Resample(v)$ is at
most $\out(v)$. Additionally, if a vertex $u\in \OUT(v)$ gets recolored by $\Resample(v)$, then in expectation $O(1)$ many vertices in $N(u)$ have the same color and become uncolored in $\varphi$; resulting in $O(\out(v))$ many uncolored vertices in expectation (see \Cref{fig:recourse-colorOH}).

Let $Q\subseteq N(v)\cap K(v)$ be the set of vertices inside $K(v)$ that have received a color by $\Resample(v)$. We showed above that each recolored vertex may cause only $O(1)$ additional vertices to become uncolored in expectation. However, any vertex that becomes uncolored inside $K(v)$ is subsequently handled by $\ColorOH(K(v))$, which leaves at most ${4\cdot \avgout(K(v))}$ vertices uncolored. Therefore, it suffices to account for vertices that become uncolored in $\OUT(u)$ for each $u\in Q$.

For a vertex $u\in Q$, the probability that its sampled color is used by a vertex in $\OUT(u)$ is at most $\out(u)/(\Delta+1)$. By linearity of expectation
\begin{align*}
    \expect{\text{ number of vertices outside of } K(v)\text{ uncolored by }Q}
    &=\sum_{u\in Q}\frac{\out(u)}{\Delta+1}\\
    &\leq \sum_{u \in K(v)}\frac{\out(u)}{\Delta+1}\tag{as $Q\subseteq K(v)$}\\
    &\leq (1+\eps)\cdot \avgout(K(v)),
\end{align*}
where the last inequality follows from $\abs{K} \leq (1+\eps)\Delta$.
Thus, in this case at most $(1+\eps)\cdot \avgout(K(v))$ many vertices are uncolored in expectation.

Consequently, the total number of uncolored vertices
by the joint execution of $\Resample(v)$ and $\ColorOH(K(v))$ is in expectation
\[
(1+\eps)\cdot \avgout(K(v))+ {4\cdot \avgout(K(v))}+O(\out(v))=O(\cost(v)),
\]
concluding the proof.
\end{proof}

\subsection{Avoiding Yet Another Flaw}\label{sec:another-flaw}

Recall the discussion before~\Cref{lem:prob-slack-sparse}: a property we would need in our algorithm is that in the coloring $\colS$, the number of blocked colors for each vertex is not too large. 
This is guaranteed by~\Cref{lem:prob-flaw-sparse} for a vertex $v$ right after a $\Resample(v)$ call. But, we would like to maintain this property in the algorithm for all vertices and at all times. 

Given the above discussion, we define the following additional flaw. 

\begin{coloredbox}{cyan}
	\begin{flaw}\label{flaw:all}
	$v \in S \quad {{\wedge}} \quad \Ncolor{v}{\colS} > 4\delta \cdot \Delta. $
	\end{flaw}
	\noindent We say a vertex $v\in S$ is \textbf{far from} having \Cref{flaw:all} if $ \Ncolor{v}{\colS} < 2\delta \cdot \Delta. $
\end{coloredbox}

We should note that unlike~\Cref{flaw:sparse} and~\Cref{flaw:out-heavy}, our aim is to maintain the graph free of \Cref{flaw:all} (or at least some approximation of it) at all times.  

We can test if a vertex has this flaw quickly (with high probability) using the following lemma.  

\begin{lemma}\label{clm:test-works}
    For any sparser vertex $v\in S$, if $v$ has \Cref{flaw:all}, then $\test(v)$ will return \Fail with high probability. Conversely,  if $v$ is far from having \Cref{flaw:all}, then $\test(v)$ will output \Pass with high probability. Moreover, the running time of $\test(v)$ is $O(\log n)$ with certainty.
\end{lemma}

We prove this lemma by a simple sampling argument and standard concentration inequalities. We postpone this proof to \Cref{sec:omitted-proofs}.

The way to fix a vertex $v$ with a~\Cref{flaw:all} is to simply call $\Resample(v)$ (as proven~\Cref{lem:prob-flaw-sparse}). We can analyze the time spent for these calls subsequently in our proof. But, the catch is 
how to find vertices that violate this flaw? We do this next. 

\paragraph{Checking for~\Cref{flaw:all}.} We have one final subroutine, called $\checkFlaw$,  whose goal is to handle this task.  
The idea behind $\checkFlaw$ is to use a ``lazy'' checking combined with $\test$ in~\Cref{clm:test-works}. A difference of $\checkFlaw$ with our other subroutines is that we do not run this \emph{sequentially} with other subroutines but rather
it is always running in the \emph{background}; think of $\checkFlaw$ as being called in multiple of our other subroutines as we will describe below. Although, to avoid the clutter in our algorithms' description, we do not explicitly write these calls to $\checkFlaw$ 
inside them and only here specify how they should be. 

Specifically, whenever any vertex $v \in S$ changes color in $\colS$, or joins or leaves $S$, or adds and deletes an edge to another vertex (basically whenever $G[S]$ or $\colS$ changes), 
we call $\checkFlaw(v)$: 
\begin{Algorithm}\label{alg:check-flaw}
	$\checkFlaw(v)$: For any vertex $u \in N(v)$, with probability ${\log{n}}/{(\delta^2 \cdot \Delta)}$, run $\test(u)$,  and if it returned $\Fail$, call $\Resample(u)$; go back to $\test(u)$ step.  
\end{Algorithm}
Again, we emphasize that from now on, we consider $\checkFlaw$ as having being incorporated in all our algorithms that deal with $\colS$ without us including this part to each one explicitly. 

The following lemma shows $\checkFlaw$ can actually be implemented quite fast when it does not call $\Resample$. 
\begin{lemma}\label{lem:checkFlaw} 
	A call to $\checkFlaw(v)$ take $O(\log^2{n})$ time with high probability ignoring the time spent per $\Resample$ calls. 
\end{lemma}
\begin{proof}
	Recall that $\delta$ is a constant. An easy application of Chernoff bound proves that $\checkFlaw(v)$, with high probability, only ever needs to check $O(\log{n})$ of neighbors of $v$ for $\test$; thus, even though the algorithm is stated as a for-loop, 
	we can first sample the names of these vertices (from the proper distribution), and then jump directly to testing those vertices instead of looping through all neighbors. Since the runtime of $\test$ is $O(\log{n})$ with certainty (by~\Cref{clm:test-works}), we obtain the final bound.  
\end{proof}

The following lemma captures the correctness of \checkFlaw. 

\begin{lemma}\label{lem:checkFlaw-correct}
	Throughout the first polynomially many calls to $\checkFlaw$, with high probability, $\Resample(v)$ is called 
	on any vertex $v \in S$ before $\Ncolor{v}{\colS} > 5\delta \cdot \Delta$ ever. 
\end{lemma}
\begin{proof}
	Fix any vertex $v$ and suppose $\test(v)$ has been called at some time step and it output \Pass. By~\Cref{clm:test-works}, we know that at this point $\Ncolor{v}{\colS} \leq 4\delta \cdot \Delta$ assuming correctness of \test. 
	Now consider the changes in the neighborhood of $v$ that can affect $\Ncolor{v}{\colS}$; anytime such a change happens, $\checkFlaw$ tosses a coin to re-run $\test(v)$ with probability ${\log{n}}/{(\delta^2 \cdot \Delta)}$. 
	This implies that before another $\delta \cdot \Delta$ changes to $\Ncolor{v}{\colS}$ happens, with high probability, $\test(v)$ again will be called. Conditioned on this, we can conclude that between any two calls to $\test(v)$, 
	there can be at most $\delta \Delta$ changes to $\Ncolor{v}{\colS}$ which implies that as long as $\test$ does not output wrong, vertex $v$ should never have $\Ncolor{v}{\colS} > 5\delta\Delta$ before $\Resample(v)$ is called 
	(and the algorithm goes into a go-back loop of running $\test$ again). As such, since first polynomially many calls to $\test$ succeed with high probability, a union bound over all vertices and calls to $\checkFlaw$ concludes the proof. 
\end{proof}

\Cref{lem:checkFlaw-correct} allows us to maintain a state of being approximately~\Cref{flaw:all}-free at all times.


\subsection{Dynamic Entropy Compression: Bounding Number of Calls to $\Resample$}

We are now ready to analyze the runtime and recourse of $\UpdateSparseColor$, throughout all updates of our algorithm. 
The key is to bound the ``cost'' the calls to $\Resample$ and $\ColorOH$ throughout all updates. We use dynamic entropy compression for this, following the description in~\Cref{sec:high-level} (with various changes for the formal proof). 


\paragraph{Maintaining a log.}
We need to record an accurate journal of the algorithm in order to account for all possible sequence of events that result in many calls to $\Resample$. 
We do this by maintaining a log which records every update to the graph along with the changes applied to the decomposition and every call to $\Resample$ throughout the execution of \Cref{alg:delta+1}. 
Every entry of this log is of one of the following forms:
\begin{itemize}
    \item $U_i = (\texttt{type = update}, u_i, v_i, \pm)$: which records the $i$-th update to the graph for some integer $i$, where the edge $(u_i,v_i)$ is either inserted or deleted according to the sign $\pm$.
    
    \item $R_j = (\texttt{type = resample},\texttt{source}=\ell\in \set{\text{`U',`C'}}, v_j)$: which records the $j$-th call to $\Resample$ for some integer $j$, called on vertex $v_j$, and $\ell=$`U' if the call was from \UpdateSparseColor and $\ell=$`C' if it came
    from $\checkFlaw$. 
    
    \item $D_k=(\texttt{type = decomp}, v_k, X_k)$: which represents the $k$-th change applied to the decomposition as a tuple which shows a vertex $v_k$ that have moved in the decomposition and the new set $X_k$ it belongs to (which is either $\Vsparse$ or $K(v_k)$).
\end{itemize}

Note that each update incident on sparser vertices causes a call to \SampleColor in \Cref{alg:delta+1}. Thus, as we are recording the updates in the log, we can also recover these calls to \SampleColor as well (this also holds 
for the moves between $S$ and $D$ in the decomposition that are logged). 

Using the construction of this log and given the number of updates, total changes applied to the decomposition, and the calls to \Resample, one can bound the number of possible logs as below.

\begin{claim}\label{obs:number-of-logs}
    The number of distinct logs containing $t$ updates, $d$ changes to the decomposition, and $r$ calls to \Resample is at most $n^{5(t+r+d)}$.
\end{claim}
\begin{proof}
    For each $i\in [t]$ there are at most $2n^2$ many different content for the bits of record $U_i$ as there is $n$ choice for each endpoint of the edge and 2 choices for insertion or deletion. And for each $j\in [r]$ there are at most $2n$ different content possible for $R_j$ as there are $n$ possible choices for $v_j$ and 2 choices for the line on which \Resample was called. Each change to the decomposition is written as a tuple of a vertex and its corresponding set, so there are $n^2$ different possible contents for each tuple (assuming each set of the decomposition is represented by one of its vertices).

    As all possible arrangements between $\{U_i\}_{i\in [t]}$, $\{D_k\}_{k\in [d]}$, and  $\{R_j\}_{j\in [r]}$ is at most ${t+r \choose r}\cdot {t+r+d \choose d}$, we get that the number of possible distinct contents for such log is at most 
    \[
    {t+r \choose r}\cdot {t+r+d \choose d}\cdot (2n^2)^t\cdot (2n)^r\cdot (n^2)^d \leq n^{5(t+r+d)}.
    \]
    The inequality is true as ${a+b \choose b}\leq 2^{(a+b)}\leq n^{a+b}$, for integers $a$ and $b$,
    concluding the proof.
\end{proof}

\Cref{obs:number-of-logs} bounds the number of possible logs. However, this is not enough for our purpose as not all logs have equal ``cost’’ in terms of update time and recourse. To measure this cost, we extend the definition of cost of vertices to logs as well. 
Formally, we define the cost of a log as follows. 
\begin{Definition}[Cost of a log]\label{def:cost-log}
    For a fixed log \LOG with $\{R_i\}_{i\in [r]}$ as its set of recorded calls to \Resample, let $v_i$ denote the vertex written in $R_i$. Define $R^{(U)}$ to be the set of indices $j$ such that parameter $\ell$ of $R_j$ is `U' (so sourced from \UpdateSparseColor) and the set $R^{(C)}$ as the set of indices referring $\ell=$`C' (so sourced from \checkFlaw). Define 
    \[
    \cost(\LOG):=\abs{R^{(C)}} \cdot \frac{\delta\Delta}{3}+\delta^6\cdot \eps^2\cdot \sum_{i\in R^{(U)}}\cost(v_i).
    \]
\end{Definition}


\paragraph{Random table of assignments.}
In order to analyze the behavior of the algorithm we need to control the randomness used in the calls to \SampleColor and therefore \Resample.
We simulate all random bits used in these calls via a table of random assignments
\[
\randTable:V \times \IN \rightarrow \{0, 1\}\times [\Delta+1],
\]
where each entry $\randTable(u, i)$ contains a tuple whose first element is 1 with probability $\delta$ and 0 otherwise, and whose second element is drawn uniformly at random from $[\Delta+1]$.

We may assume the entries of $\randTable$ are sampled before the algorithm begins. That is, instead of sampling new randomness each time \SampleColor is called, the algorithm samples a table of assignments \randTable uniformly at random
at the beginning and uses its randomness throughout. Note that all changes to $\colS$, for vertices currently inside $S$, comes through calls to $\SampleColor$ (when a vertex $v$ moves from $S$ to $D$, the algorithm immediately sets $\colS(v) = \perp$). 
The algorithm will then use this random table of assignments as follows:
\begin{itemize}
    \item The initial color assigned to a vertex $v$ is determined by $\randTable(v, 1)$: if the first element of $\randTable(v, 1)$ is 0, then $\colS(v)=\perp$; otherwise, $v$ is assigned the color given by the second element of $\randTable(v, 1)$ 
    (same as in the initiation phase of \Cref{alg:initial}). 
     
    \item After starting the algorithm, each call to \SampleColor consumes fresh randomness in the table, that is, it uses a cell in $\randTable$ that has not been used (or even read) up to that point. Consequently, each call to $\Resample(v)$ for some $v\in S$ consumes fresh randomness for every vertex $u\in N_S(v)$. We simulate this by advancing to the next unused entry in row $u$ of $\randTable$ (see \Cref{fig:indices}). 
    
    \item Let $\INDEX{u}{i}$ denote the number of times that $\SampleColor(u)$ has been called on $u$ by the time the algorithm reaches the call $R_i$, i.e. the $i$-th call to \Resample as recorded in \LOG. 
    We can recover this information from the log directly. 

\end{itemize}

\begin{figure}[htbp]
    \centering
    \begin{subfigure}[b]{0.38\textwidth}
        \centering
        \includegraphics[width=\textwidth]{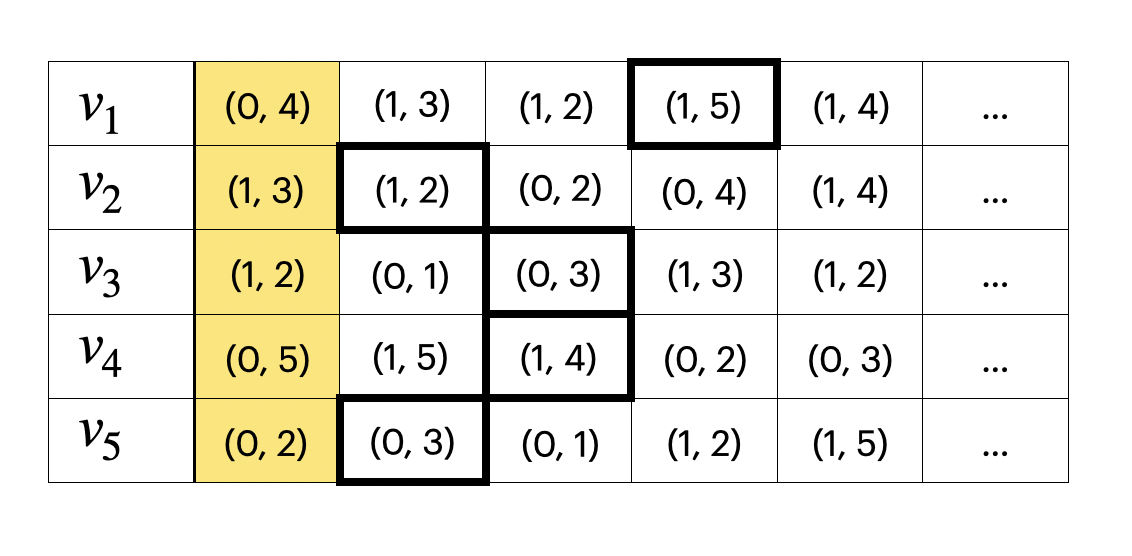}
        \caption{Before calling $\Resample(v_1)$.}
        \label{fig:index-sub1}
    \end{subfigure}
    \hfill
    \begin{subfigure}[b]{0.38\textwidth}
        \centering
        \includegraphics[width=\textwidth]{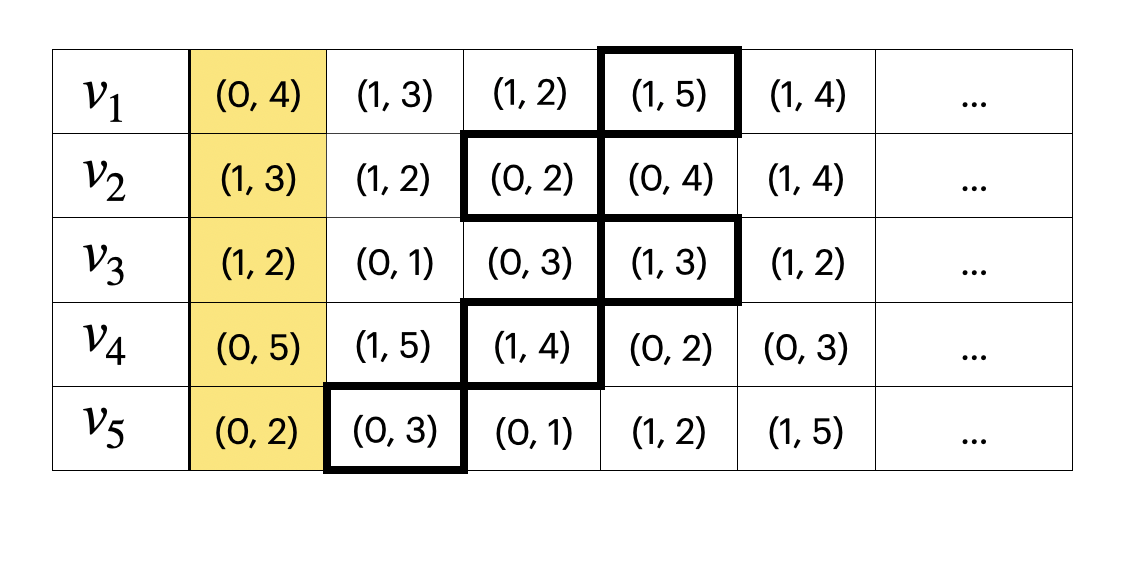}
        \caption{After calling $\Resample(v_1)$.}
        \label{fig:index-sub2}
    \end{subfigure}
    \begin{subfigure}[b]{0.20\textwidth}
        \centering
        \includegraphics[width=\textwidth]{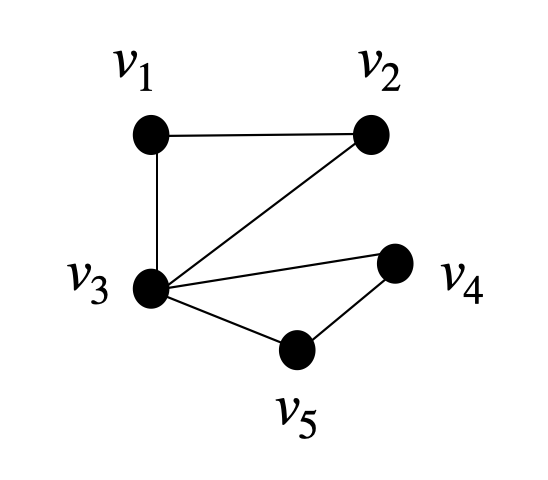}
        \caption{The Graph}
        \label{fig:index-graph}
    \end{subfigure}
    \hfill
    \caption{The table show the random table of assignments for the graph in \Cref{fig:index-graph} where all its vertices are in $S$ at this point of time. The bold boxes show the ``Index'' of each vertex at the current state. The first (yellow) column shows the initial colors assigned to the vertices.}
    \label{fig:indices}
\end{figure}

Using $\randTable$ as the source of randomness for \Resample does not affect the behavior of the algorithm. As all entries are mutually independent, we may therefore assume that \Cref{alg:delta+1} draws the randomness that it needs for calls to \Resample and \SampleColor from $\randTable$, as described above.

By \Cref{lem:attempt}, we know that the subroutines \ColorSparse correctly detect flaws with high probability. This is similarly true for~\Cref{clm:test-works} and $\test$. 
Finally \checkFlaw also does its job correctly with high probability by~\Cref{lem:checkFlaw-correct}. 
By a union bound over the first polynomially many calls to them, these guarantees continue to hold. Define the following event:
\begin{align*}
\eventTwo:=~&\text{all three of \ColorSparse, \test, and~\checkFlaw behave correctly as} \\
&\text{in~\Cref{lem:attempt}, \Cref{clm:test-works}, and~\Cref{lem:checkFlaw-correct}, respectively}. 
\end{align*}
We can condition on $\eventTwo$ and fix randomness of these three subroutines (without $\Resample$ calls inside the last one) and treat them as deterministic in the remainder of the analysis. The randomness of \SampleColor and therefore \randTable remain entirely independent.


\paragraph{Consistency of a log with \Cref{alg:delta+1}.}
We say that a log $\LOG$ is \textbf{consistent} with the random table $\randTable$ if, when \Cref{alg:delta+1} is executed using the randomness stored in $\randTable$, it invokes the subroutine \Resample exactly as recorded in $\LOG$. Formally, let $R_i$ denote the $i$-th \Resample record in $\LOG$, say, $\Resample(v_i)$. Then $\LOG$ is consistent with $\randTable$ if, at the time of $R_i$, the randomness in $\randTable$ certifies the flaw responsible for that invocation: 
\begin{itemize}
\item if $R_i$ logs a call to \Resample with $\ell=$`C', then $v_i$ is \underline{not} far from having \Cref{flaw:all}, and
\item if $R_i$ logs a call to \Resample with $\ell=$`U', then $v_i$ has one of \Cref{flaw:sparse} or \Cref{flaw:out-heavy}.
\end{itemize}

This notion allows us to reduce the analysis to bounding the probability that a fixed candidate log is consistent with the random table, and then union bounding over all logs. 

To verify consistency, it suffices to reconstruct the state of each recorded vertex $v_i$ immediately before the corresponding call to \Resample. The decomposition stored in the log uniquely determines the neighborhood of $v_i$
and following the log and using $\randTable$ we can determine the current coloring in the neighborhood of $v_i$ and check the consistency.

This definition immediately implies the following observation.
\begin{observation}
    A log $\LOG$ of an execution of \Cref{alg:delta+1} using \randTable is consistent with $\randTable$.
\end{observation}

For each record $R_i$ and sparser neighbor $u\in N_S(v_i)$, consider the mapping
\[
(R_i, u) \longmapsto (u, \INDEX{u}{i}).
\]
This mapping is injective by the definition of $\INDEX{u}{i}$: every call to \Resample consumes fresh randomness for each participating vertex and advances the corresponding index. Consequently, no entry of $\randTable$ is associated with two distinct pairs $(R_i,u)$. Thus, for a fixed log $\LOG$, every recorded call accesses a unique collection of random-table entries.


We bound the probability that a fixed log with ``large cost'' is consistent with the random table. Let $\LOG$ be a fixed log describing  resampling records $R_1,\ldots,R_r$. By chain rule,
\begin{align*}
    \Pr_{\randTable}\paren{\LOG \text{ is consistent with $\randTable$}\mid \eventTwo} 
    &=\Pr\paren{R_1, \ldots, R_r\text{ all are consistent with $\randTable$}\mid \eventTwo} \\
    & =\prod_{i\in [r]} \Pr\paren{R_i \text{ consistent with $\randTable$}\mid R_{<i}\text{ consistent}, \eventTwo};
\end{align*}
when we condition on $R_{<i}$, we can even all the randomness used by the algorithm before the call to $R_i$.
Therefore, it suffices to bound each conditional probability
\[
\Pr(R_i\text{ is consistent}\mid R_{<i}\text{ are consistent},\eventTwo),
\]
that is, conditioned on all earlier randomness from \randTable used by the algorithm prescribed by the log, we bound the probability that the next recorded call is also justified by the randomness in $\randTable$.

\begin{lemma}\label{lem:bounding-conditional-prob}
    Given $t$ updates on the graph, a random table of assignments \randTable, and a fixed log \LOG with the same $t$ updates and $R_1, \ldots, R_r$ set of its records for calls to \Resample, let $v_i$ be the vertex in $R_i$ for each $i\in [r]$; then the following holds:
    \begin{enumerate}
        \item If $v_i\in S$ and $R_i$ is called with $\ell=$`C', 
        \[
        \Pr(R_i \mid R_{<i}, \eventTwo)\leq 2^{-{\delta\cdot\Delta}/{3}}.
        \]
        \item If $v_i\in \Vsparse$ and $R_i$ is called with $\ell=$`U', 
        \[
        \Pr(R_i \mid R_{<i}, \eventTwo)\leq 2^{-\delta^6\eps^4\cdot \cost(v)}.
        \]
        \item If $v_i\in V(\KKoh)$ and $R_i$ is called with $\ell=$`U', 
        \[
        \Pr(R_i \mid R_{<i}, \eventTwo)\leq 2^{-\delta^5\eps^4\cdot \cost(v)}.
        \]
    \end{enumerate}
\end{lemma}
\begin{proof}
    Throughout the proof we assume everything is calculated conditioned on \eventTwo.

    \textbf{Case 1: $R_i$ is called with $\ell=$\textnormal{`C'}.} 
    
    As we conditioned on \eventTwo, this case happens only when $v_i$ is not far from having~\Cref{flaw:all}. 
    Therefore, conditioned on the randomness $R_{<i}$, the probability that $R_i$ occurs is at most the probability that $v_i$ is not far from \Cref{flaw:all}, conditioned on the coloring outside of $N_S(v_i)$.

    Let $\mathcal{R}_i$ denote the collection of most recent randomness used by vertices in $N_S(v_i)$. By the definition of $\INDEX{\cdot}{i}$, the mapping
    \[
    (R_i, u) \longmapsto (u, \INDEX{u}{i}),
    \]
    is an injection over $u\in N_S(v_i)$. Consequently, the randomness of each $\mathcal{R}_i$ corresponds to distinct entries of \randTable and are therefore independent.

    Moreover, the condition $R_{<i}$ can only reveal information about randomness outside of $\mathcal{R}_i$, and thus can only fix the coloring outside of $N_S(v_i)$. By construction, the random entries in $\mathcal{R}_i$ are distributed exactly as the 
    randomness used by $\SampleColor$ (and thus also $\Resample(v_i)$). 

    The content of $A(u, \INDEX{u}{i})$ was sampled by a call to $\SampleColor(u)$ at some time $t_u$. At this time $t_u$, there were some colors that were not available for $u$ to keep: (1) colors used by sparser neighbors of $u$ under $\colS$ or (2) 
    colors sampled and retained at the same time $t_u$ by a vertex $w$ with lower lexicographical order. Let $B(u)$ be the set of these colors that are not available to $u$. 
    
    Consequently, we can apply \Cref{lem:prob-flaw-sparse} as vertices in $N_S(v_i)$ follow the same process for sampling colors. Specifically, even though a vertex $u \in N_S(v_i)$ may have sampled its color at time $t_u$ 
    and the decision to retain that color is taken at that time, the color of $u$ at the time we are checking $R_i$ and $N_S(v_i)$ is the same, as all changes to colors of $u$ can only be done through $\SampleColor(u)$ (thus, the 
    coloring $\colS$, restricted to $N_S(v_i)$, at this time is the same as the coloring $\chi$ of the lemma). As such, we obtain that 
    \[
    \Pr\paren{\Ncolor{v_i}{\colS}>2\delta\Delta}\leq 2^{-\delta\cdot\Delta/3}.
    \]
    Equivalently, the probability that $v_i$ is not far from \Cref{flaw:all} is at most $2^{-\delta\cdot\Delta/3}$. Therefore,
    \[
    \Pr(R_i\mid R_{<i})\leq \Pr(\Ncolor{v_i}{\colS}\geq 2\delta\Delta \mid R_{<i})\leq 2^{-\delta\cdot\Delta/3}.
    \]

    \textbf{Case 2:  $R_i$ is called with $\ell=$\textnormal{`U'}.} 
    
    Under the condition $\eventTwo$, the probability that $R_i$ is consistent with \randTable is at most the probability that $v_i$ has \Cref{flaw:sparse} or \Cref{flaw:out-heavy}. 
     This is because $\ColorSparse(v_i)$ is invoked immediately before this call to $\Resample$. Consequently, if $v_i$ has neither \Cref{flaw:sparse} nor \Cref{flaw:out-heavy}, then by the conditioning
      we have, \ColorSparse would color $v_i$ in $\varphi$, and the call to \Resample would not occur. Moreover, we have conditioned in $\eventTwo$ that no vertices $u \in S$ having $\Ncolor{u}{\colS} > 5\delta\Delta$ at any point until now. 

    Similar to the previous case, we know that $R_{<i}$ fixes the randomness sampled up to this point, except for the most recent randomness corresponding to $\mathcal{R}_i$ as described above.
    For each $u\in N_S(v_i)$ let $t_u$ be the last time $\SampleColor(u)$ was called and define $B(u)$ as its blocked colors at time $t_u$.
    By our conditioning earlier, we know that $\card{B(u)} \leq 5\delta\Delta$ at time $t_u$. 

    Therefore, the random colors corresponding to random bits of $\mathcal{R}_i$ follow the same process as \Cref{lem:prob-slack-sparse}. 
    As argued above, the random bits in $\mathcal{R}_i$ are mutually independent and conditioning on $R_{<i}$ only fixes randomness used outside $N_S(v_i)$ and consequently fixes the blocked sets.

    Recall that a vertex $v_i \in \Vsparse$ has \Cref{flaw:sparse} when $\deg_S(v)\geq \Delta- \delta\cdot \eps^2\cdot \Delta$ and $\slack{v}{\colS} < \eps^2\cdot \delta^6\cdot \Delta$. Moreover, we know that $v_i \in \Vsparse$ is $(\delta\cdot \eps^2)$-sparse for small constant $\delta$ and we have defined $\cost(v_i)=\Delta$.
    If $\deg_S(v_i)< \Delta-\delta\cdot \eps^2\cdot \Delta$ then $v_i$ does not have \Cref{flaw:sparse} with certainty. 
    Thus, assuming $\deg_S(v_i)\geq \Delta-\delta\cdot \eps^2\cdot \Delta$, using \Cref{lem:prob-slack-sparse} with $\gamma= \delta\cdot \eps^2$ we have (again, notice that $\chi$ of that lemma corresponds to $\colS$ here by our argument earlier
    that the colors in $N_S(v_i)$ have not changed since their sampling time), 
    \[
    \Pr(R_i\mid R_{<i})\leq \Pr(\slack{v}{\colS} < \eps^2\cdot \delta^6\cdot \Delta)\leq 2^{-\delta^6\cdot \eps^2\cost(v)}.
    \] 

    Similarly a vertex $v_i\in V(\KKoh)$ has \Cref{flaw:out-heavy} if $\deg_S(v_i) \geq \Delta-\delta\cdot \eps^2\cdot \eta$ and $\slack{v_i}{\colS} < \eps^2\cdot \delta^6\cdot \eta$
    By definition of cost for $v_i \in V(\KKoh)$ we have $\cost(v_i) = \max (\out(v_i),\avgout(K(v_i)))$. Therefore, by \Cref{lem:OH-is-out-sparse}, we get that each $v_i\in V(\KKoh)$ is $(\delta\cdot \eps^2\cdot \cost(v_i)/\Delta)$-sparse. Hence, given $v_i \in V(\KKoh)$ with $\deg_S(v_i) \geq \Delta-\delta\cdot \eps^2\cdot \eta$, by \Cref{lem:prob-slack-sparse} for $\gamma= \eps^2\cdot \cost(v_i)/\Delta$ we have
    \[
    \Pr(R_i\mid R_{<i})\leq \Pr(\slack{v}{\colS} < \eps^2\cdot \delta^6\cdot \eta) \leq 2^{-\delta^5\cdot \eps^4\cost(v)},
    \] 
    concluding the proof.
\end{proof}

We can now bound the probability that an entire specific log \LOG is consistent with the table.

\begin{lemma}\label{lem:prob-of-log}
    For a random table of assignments \randTable, the probability that a given fixed log \LOG becomes consistent with \randTable, conditioned on \eventTwo, is at most $2^{-\cost(\LOG)}$.
\end{lemma}
\begin{proof}
    Throughout this proof we assume everything is calculated conditioned on \eventTwo.

    Let $\{R_i\}_{i=1}^r$ be the set of records in \LOG for calls to \Resample. For each $i\in [r]$, let $v_i$ be the vertex on which $R_i$ was called for. We assume conditioning on \eventTwo throughout this proof and avoid writing it repeatedly.
    We already saw that by chain rule, 
    \begin{align*}
    \Pr\paren{\LOG \text{ is consistent with }\randTable} 
    &=\Pr\paren{R_1, \ldots, R_r\text{ consistent with }\randTable} \\
    & = \Pr(R_1)\cdot \Pr(R_2\mid R_1) \cdot \ldots \cdot \Pr(R_r \mid R_{<r}).
    \end{align*}
    Let $Z$ be the set of indices $i\in [r]$ where $R_j$ is called with $\ell=$`C' of \UpdateSparseColor. And let $Y$ be the ones with $\ell=$`U'.
    By \Cref{lem:bounding-conditional-prob},  for each $j\in Z$ we have
    \[
    \Pr(R_j\mid R_{<j})\leq 2^{-\delta\cdot \Delta/3},
    \]
    and for each $j\in Y$ we have 
    \[
    \Pr(R_j\mid R_{<j})\leq 2^{-\delta^6\cdot \eps^2\cdot \cost(v_j)}.
    \] 
    Combining the above inequalities, we have
    \begin{align*}
        \Pr\paren{\LOG \text{ is consistent with }\randTable}
        &\leq \prod_{i\in Z} 2^{-\delta\cdot \Delta/3} \cdot \prod_{j\in Y} 2^{-\delta^6\cdot \eps^2\cdot \cost(v_j)}
        \leq 2^{-\paren{\abs{Z}\cdot(\delta\cdot \Delta/3) + \delta^6\cdot \eps^2\cdot \sum_{j\in Y} \cost(v_j)}}.
    \end{align*}
    Therefore, by definition of $\cost(\LOG)$ we can conclude that
    \[
    \Pr\paren{\LOG \text{ is consistent with }\randTable}
    \leq
    2^{-\cost(\LOG)}. \qed
    \]

\end{proof}

We now prove that over $t$ updates, $\cost(\LOG)$ is limited with high probability.

\begin{lemma}\label{lem:cost-log-limited}
    If \LOG is a log created through the execution of \Cref{alg:delta+1} over $t$ updates with $d$ changes to the decomposition, we have
    \[
    \Pr\paren{\cost(\LOG) \geq 100 (t+d)\cdot\log n\mid \eventTwo}< n^{-t}.
    \]
\end{lemma}
\begin{proof}
    In this proof we condition on the event \eventTwo and avoid writing it repeatedly.
    
    Assume \randTable is the table of random assignments used by \Cref{alg:delta+1}.
    By \Cref{lem:prob-of-log}, we saw that given a fixed log \LOG with $\cost(\LOG)=C$, the probability that this log becomes consistent with \randTable, and therefore is a possible log for the algorithm, is at most $2^{-C}$. 

    Let $\mathcal{L}_{r, C}$ be the set of all possible logs with cost exactly $C$ and $r$ calls to resample over $t$ updates with $d$ many changes to the decomposition.
    Using union bound we have
    \begin{align*}
        \Pr(\cost(\LOG)> 100(t+d)\cdot\log n)
        &\leq \sum_{\substack{r\geq 0,\\
        C>100\cdot (t+d)\log n}}\sum_{L\in \mathcal{L}_{r, C}} \Pr(L \text{ is consistent with }\randTable)\\
        &\leq \sum_{\substack{r\geq 0,\\
        C>100\cdot (t+d)\log n}}\sum_{L\in \mathcal{L}_{r, C}} 2^{-C}. \tag{by \Cref{lem:prob-of-log}}
    \end{align*}

    We next bound the number of possible pairs consisting of a log and its corresponding cost value when there are $t$ updates, $d$ changes in decomposition, and $r$ calls to \Resample.
    
    By \Cref{obs:number-of-logs} we know that there are at most $n^{5(t+r+d)}$ many logs with $t$ updates, $d$ changes in decomposition, and $r$ calls to \Resample.

    Consider a fixed log. By \Cref{def:cost-log}, each call to \Resample contributes either a fixed quantity $\delta\cdot\Delta/3$ (if it is invoked with $\ell=$`C') or a quantity of the form $\delta^6\eps^2\cdot\cost(v)$ (if it is invoked 
    with $\ell=$`U') to $\cost(\LOG)$.
    
    If $v\in V(\KKoh)$, then
    $\cost(v)\in \left[4\eta/\eps^2,\eps\Delta\right]$,
    while if $v\in\Vsparse$, then $\cost(v)=\Delta$. Consequently, each call to \Resample can contribute at most
    \[   
    \left(\eps\Delta-\frac{4\eta}{\eps^2}\right)+2
    \le
    \eps\Delta
    \]
    different values to the cost, where the additional two possibilities account for the fixed contribution $\delta\cdot\Delta/3$ and the value $\cost(v)=\Delta$. Therefore, for a fixed set of $r$ calls to \Resample recorded in a log, the number of possible values of $\cost(\LOG)$ is at most
    \[
    (\eps\Delta)^r \leq 2^{r\log(\eps\Delta)} 
    \leq 2^{r\log n}.
    \] 
    Combining this with the bound on the number of logs, the total number of possible $(\LOG,\cost(\LOG))$ pairs is at most
    \[
    2^{r\log n}\cdot 2^{5(t+r+d)\log n}
    \leq
    2^{6(t+r+d)\log n}.
    \]

    As each $v_i$ has $\cost(v_i)\geq \eta$, by definition of $\cost(L)$ for some log $L$ with $r$ many calls to \Resample, we know that 
    \[
    r \leq \frac{\cost(L)}{\eps^2\delta^6\eta} \leq \frac{\cost(L)}{\delta^8\eta} 
    \]
    as $\delta<\eps$. Therefore, the union bound can be stated as below 
    \begin{align*}
        \Pr(\cost(\LOG)> 100(t+d)\cdot \log n)
        &\leq \sum_{C>100(t+d)\log n}\ 
        \sum_{0\leq r\leq C/(\delta^8\eta)}\
        \sum_{L\in\mathcal{L}_{r,C}} 2^{-C}\\
        &\leq \sum_{
        C>100(t+d)\log n}
        \paren{\sum_{0\leq r\leq C/(\delta^8\eta)} 2^{6(t+r+d)\log n}\cdot  2^{-C}}\tag{as there are at most $2^{6(r+t+d)\log n}$ such logs}\\
        &\leq \sum_{
        C>100(t+d)\log n} \frac{C}{\delta^8\eta}\cdot 2^{6(t+d)\log n}\cdot 2^{-C + 6C\log n/(\delta^8\eta) } \tag{as $r\leq C/(\delta^8\eta)$}\\
        &\leq \sum_{
        C>100(t+d)\log n} 2^{6(t+d)\log n}\cdot 2^{-C/2}\tag{as $\eta\gg \log n$}\\
        &\leq \sum_{
        C>100(t+d)\log n}  2^{-C/4}\tag{as $C\geq 100(t+d) \cdot \log n$}\\
        &\leq 2^{-t\log n},
    \end{align*}
    where the last inequality is true as the first term in the geometric series is $2\cdot 2^{-25\cdot t\log n}$ which dominates the other terms, concluding the proof.
\end{proof}

\subsection{Proof of~\Cref{lem:recourse-updatesparse-whp,lem:runtime-UpdateSparseColor}}

We are now ready to bound the number of uncolored vertices caused by \UpdateSparseColor and subsequently its total runtime. 
We emphasize that since we are running $\checkFlaw$ in the background, we have to be extra careful with analyzing its runtime and recourse. 
However, at least for the recourse, $\checkFlaw$ does not change any of the colorings and instead interacts with $\colS$ through calls to $\Resample$
which we are accounting for here.

\begin{lemma*}[Restatement of ~\Cref{lem:recourse-updatesparse-whp}]
    The total number of uncolored vertices caused by subroutine \UpdateSparseColor over $t$ updates is at most $O(t \log^2 n)$ with high probability.
\end{lemma*}

\begin{proof}
    Let \LOG be a log created after running \Cref{alg:delta+1} over $t$ updates.

    By \Cref{lem:decomp}, maintaining the decomposition has $O(\log n)$ amortized update time and recourse, and the number of almost-cliques that change from one set of the $\eta$-classification into another set is amortized $O(\frac{1}{\eta\Delta})$. As each almost-clique contains at most $(1+\eps)\Delta$ vertices, the amortized number of vertices that change sets is $O(\frac{1}{\eta\Delta})\cdot (1+\eps)\Delta$. Therefore, up to the end of the first two lines of \Cref{alg:delta+1} after a single update the runtime is $O(\log n)$ and the number of uncolored vertices is 
        \[
        O(\log n)+ O(\frac{1}{\eta\Delta})\cdot (1+\eps)\Delta = O(\log n). 
        \]
    By \Cref{clm:recourse-of-resample} we know that for each joint call to $\Resample(v)$ and $\ColorOH(K(v))$ (if $v \notin V(\KKoh)$ we just consider the first subroutine), the number of uncolored vertices caused by them is $O(\cost(v))$ in expectation. On the other hand, by definition of $\cost$ we have
    \begin{equation*}
        \cost(\LOG) \geq \delta^6\eps^2\cdot  \sum_{i\in[r]}\cost(v_i),
    \end{equation*}
    where $v_i$ is the vertex recorded in $R_i$.
    Therefore, the total uncolored vertices caused by the subroutine \UpdateSparseColor over $t$ updates, according to \LOG is at most $O(\cost(\LOG))$ in expectation.

    Throughout running \Cref{alg:delta+1} over $t$ updates, let $v_i$ be the vertex on which the $i$-th call to \Resample was invoked.
    Define $X_i$ to be the random variable showing the number of uncolored vertices caused by the $i$-th joint call to \Resample and \ColorOH (if $v\in V(\KKoh)$). We know by \Cref{clm:recourse-of-resample} that $\expect{X_i}=O(\cost(v_i))$. Note that $\cost(v_i)$ depends on the graph at that point and also $v_i$ which depends on the outcome of the algorithm up to this point. Define $\mathcal{D}_i$ as the distribution on which $X_i$ is drawn from where this distribution depends on the outcome of the algorithm up to this point, and therefore depends on the values of $X_1, \ldots, X_{i-1}$.

    As \ColorOH deterministically leaves at most ${4\cdot \cost(v_i)}\leq \Delta$ many vertices uncolored and \Resample recolors at its neighbors  where each can also result in uncoloring their neighbors with the same color in $\varphi$, therefore we always have $X_i\leq m$ where $m$ is the maximum number of edges that are present in the graph after some update.

	Let $X=\sum_{i\in[r]}X_i$ for whatever value $r$ has on this execution of the algorithm. For some fixed log \LOG and by linearity of expectation we have $\expect{X}=O(\cost(\LOG))$. However, we do not have a fixed log and this is also determined by the output of the algorithm. By \Cref{lem:cost-log-limited}, for some log \LOG created over $t$ updates which contains $d$ changes of decomposition recorded, we have
    \[
    \Pr\paren{\cost(\LOG)\geq 100(t+d)\cdot \log n\mid \eventTwo}< n^{-t},
    \]
    and by \Cref{lem:decomp} we know $d=O(t\log n)$, thus we have
    \[
    \Pr\paren{\cost(\LOG)\geq 200\cdot t \cdot \log^2 n\mid \eventTwo}< n^{-t}. 
    \]
    Therefore, all random variables $\{X_i\}_i$ follow the same structure as stated in \Cref{prop:KQ-inequality} with parameters $c=m$ (as each $X_i\leq m$), $\gamma=n^{-t}$, and $\mu=200\cdot t \cdot \log^2 n$. Consequently, with $\delta=1$,
	\begin{align*}
		\Pr\paren{X\geq 400\cdot t \cdot \log^2 n\mid \eventTwo}
		&\leq n^{-t} + \exp\paren{-\frac{200\cdot t \cdot \log^2 n}{3m}}\tag{as $X_i\leq m$ for each $i$ and $\gamma=n^{-t}$}\\
		&\leq n^{-\Omega(t\cdot \log n/m)}. 
	\end{align*}
    Finally, we know that $\eventTwo$ holds with high probability over $t=\poly(n)$ updates. 
    Therefore, as we always have $t\geq m$, there are $O(t \cdot \log^2{n})$ uncolored vertices caused by \UpdateSparseColor over $t$ updates with high probability.
\end{proof}

Next, we prove the bound on total runtime of \UpdateSparseColor. Here, we have to account for the runtime of $\checkFlaw$ in the background as well. 

\begin{lemma*}[Restatement of~\Cref{lem:runtime-UpdateSparseColor}]
  Throughout running \Cref{alg:delta+1} over $t$ updates, with high probability the time spent on subroutine \UpdateSparseColor is $O(t\cdot n\log^4 n/\eta)$ time. 
\end{lemma*}

\begin{proof}[Proof of~\Cref{lem:runtime-UpdateSparseColor}]
    To bound the time spent on \UpdateSparseColor, we bound the time spent on its subroutines. That is \ColorSparse, \test, \Resample, and \ColorOH plus separately accounted for $\checkFlaw$.

    \begin{itemize}
    \item Runtime of \Resample and \ColorOH: For the $i$-th call to \Resample and \ColorOH (if the input vertex is in $V(\KKoh)$), define $X_i$ be the time spent on this call. By \Cref{lem:recourse-runtime} we know that $\expect{X_i}=O(n\log n)$. Moreover, $X_i\leq m$ because we can check at most all neighbors of all vertices. If the log constructed by the algorithm is \LOG and it has $r$ calls to resample, by \Cref{def:cost-log} we know that $r\leq \cost(\LOG)/(\delta^8\eta)$. If $X=\sum_{i\in[r]}X_i$, then $\expect{X} = O(r\cdot n) = O(\cost(\LOG)\cdot n\log n/\eta)$. 
    By \Cref{lem:cost-log-limited}, if the log \LOG created over $t$ updates contains $d$ changes of decomposition, we have
    \[
    \Pr\paren{\cost(\LOG)\geq 100(t+d)\cdot \log n\mid \eventTwo}< n^{-t},
    \]
    and by \Cref{lem:decomp} we know $d=O(t\log n)$, thus we have
    \[
    \Pr\paren{\cost(\LOG)\geq 200\cdot t \cdot \log^2 n\mid \eventTwo}< n^{-t}. 
    \]
    Therefore, all random variables $\{X_i\}_{i\in[r]
    }$ follow the same structure as stated in \Cref{prop:KQ-inequality} with parameters $c=m$ (as each $X_i\leq m$), $\gamma=n^{-t}$, and $\mu=200\cdot t\cdot n \cdot \log^3 n/\eta$. Consequently, with $\delta=1$ we get
	\begin{align*}
		\Pr\paren{X\geq 400\cdot \frac{t \cdot n \cdot \log^3 n}{\eta} \mid \eventTwo}
		&\leq n^{-t} + \exp\paren{-\frac{200\cdot t \cdot n \cdot \log^3 n}{3m\cdot \eta}}\tag{as $X_i\leq m$ for each $i$ and $\gamma=n^{-t}$}\\
		&\leq n^{-\Omega(t\cdot \log n/(m\eta))}. 
	\end{align*}
    We know that with high probability \eventTwo holds over $t=\poly(n)$ updates and setting $\eta=o(n)$ and $t=\Omega(m)$ we get that with high probability the total time spent on the joint calls to \Resample and \ColorOH over $t$ updates is 
    $O(t\cdot n \log^3 n/\eta)$ time.
    
    \item Runtime of \test: As each $\UpdateSparseColor(v)$ is invoked when $v$ had a flaw, therefore at least one call to \Resample is invoked inside \UpdateSparseColor to resolve the flaw. Thus, we can bound the number of calls to \UpdateSparseColor by $r$. For each \UpdateSparseColor, there are at most $\Delta$ calls to subroutine \test that do not result in a call to \Resample, and all the other calls to \test are accompanied with a \Resample and can be bounded by $r$. As there are at most $r$ many calls to \UpdateSparseColor, then \test is called at most $r\cdot \Delta+r$ times over all updates. By \Cref{clm:test-works}, each call to $\test$ takes $O(\log n)$ time. Therefore, the total time spent on this subroutine is $O(r\cdot \Delta\log n)$. 
    As $r\leq \cost(\LOG)/(\delta^8\eta)$ (argued earlier) and since with high probability $\cost(\LOG)=O(t\log^2 n)$ by~\Cref{lem:cost-log-limited}, we have  $O(r\cdot \Delta\log n) = O(t\cdot \Delta\log^3/\eta)$ with high probability.

    \item Runtime of \ColorSparse invoked inside \UpdateSparseColor: We know that the subroutine \ColorSparse is invoked at most once in \UpdateSparseColor. Moreover, if there are $r$ calls to \Resample recorded in the log, then there are at most $r$ calls to \UpdateSparseColor, and therefore $r$ calls to \ColorSparse inside this subroutine.
    
    Redefine $X_i$ to be the time spent on the $i$-th call to \ColorSparse.
    By \Cref{lem:attempt} we know that $\expect{X_i}=O(n\log n/\eta)$ for each $i$. Also, by \Cref{lem:worstcase-subroutines} we have $X_i\leq m$. And if $X=\sum_{i\in [r]} X_i$, we know that $\expect{X}=O(r \cdot n\log n/\eta)$. By \Cref{def:cost-log} we have $r\leq \cost(\LOG)/(\delta^8\cdot \eta)$ and by \Cref{lem:cost-log-limited} and \Cref{lem:decomp} with high probability $\cost(\LOG)\leq 400\cdot t\cdot \log^2 n$. Thus, using the same argument as above, we get that the total time spent on this subroutine is $O(t\cdot n\log^3 n/\eta)$ with high probability.
    
    \item Runtime of \checkFlaw in the background: By~\Cref{lem:checkFlaw}, the time to run $\checkFlaw$ with high probability is $O(\log^2{n})$ once we ignore $\Resample$ calls (which we already accounted for). 
    The number of calls to \checkFlaw include $O(t)$ for the updates to the graph and $O(d)$ for the changes in the decomposition, plus every call to $\SampleColor$. These themselves can be split into calls to $\SampleColor$ outside
    and inside $\Resample$. The outside calls are similarly upper bounded by $O(t+d)$ whereas in the inside calls are at most $O(\Delta)$ per each of $O(\cost(\LOG)/\eta)$ many $\Resample$ calls (the bound on the $\Resample$ is calculated above a few times already). By~\Cref{lem:cost-log-limited}, we have $\cost(\LOG) = O(t \cdot \log^2{n})$ and by~\Cref{lem:decomp}, $d=O(t\log{n})$, both with high probability. Thus, the total runtime in this step is 
    \[
    	O(t \cdot \log{n}) \cdot O(\log^2{n}) + O(t \cdot \log^2{n}) \cdot \frac{\Delta}{\eta} \cdot O(\log^2{n}) = O(t \cdot \frac{\Delta \cdot \log^4{n}}{\eta}). 
    \]
    \end{itemize}
    Therefore, the total time spent on \UpdateSparseColor is 
    \[
    O(t\cdot n\log^3 n/\eta) + O(t\cdot \Delta\log^4/\eta) =O(t\cdot n\log^4 n/\eta). \qed
    \]
    
\end{proof}

%% file: sparse-subroutines.tex
\section{Handling Remaining Sparser Vertices}\label{sec:sparse-subroutines}

We now describe and analyze the $\ColorSparse$ and $\ColorOH$ subroutines.

\subsection{The \ColorSparse Subroutine}
 
We use the following subroutine to color in $\varphi$ any sparser vertex $v$ that does not have \Cref{flaw:sparse} or \Cref{flaw:out-heavy} and is blank in $\colS$. 

\begin{Algorithm}\label{alg:attempt-color-sparse} 
$\ColorSparse(v):$  
\begin{enumerate}
	\item For $i=1$ to $\Delta \cdot \log n/(\delta^8\cdot \eta)$ trials:
	\begin{enumerate}
        \item Sample a color $c\in [\Delta+1]$ uniformly at random.
        \item Iterate over all vertices in $\mathcal{C}_c$ under $\varphi$ and if any is a sparser neighbor of $v$ continue to the next iteration of for-loop. Otherwise, set $\varphi(v)=c$, uncolor the denser neighbors with color $c$ in $\varphi$, return \emph{Success}, and terminate.
    \end{enumerate}
    \item If did not terminate until now, return \Fail.
\end{enumerate}
\end{Algorithm} 
\begin{lemma*}[Restatement of \Cref{lem:attempt}]
    For a sparser vertex $v \in S$, $\ColorSparse(v)$ takes $O(n\log{n}/\eta)$ time in expectation. If $v$ does not have any of~\Cref{flaw:sparse} and \Cref{flaw:out-heavy}
	under $\colS$, $\ColorSparse(v)$ with high probability colors $v$ successfully in $\varphi$ without returning \Fail. Moreover, this subroutine uncolors at most $O(1)$ many denser vertices in expectation.
\end{lemma*}
\begin{proof}
    For a random color $c\in [\Delta+1]$, size of $\mathcal{C}_c$ is at most
	$n/(\Delta+1)$ in expectation. Therefore, checking availability of each color takes $O(n/\Delta)$ time in expectation. As we have $O(\Delta \cdot \log{n}/\eta)$ trials, the expected runtime of \ColorSparse is $O(n\log n/\eta)$ ($\delta$ is constant).

    Assuming vertex $v$ does not have any of~\Cref{flaw:sparse} and \Cref{flaw:out-heavy}, then if $v\in \Vsparse$ we either have $\deg_S(v) <\Delta-\delta\cdot \eps^2\cdot \Delta$ or $\slack{v}{\colS}>\delta^6\cdot \eps^2\cdot \Delta$. 
    If $\slack{v}{\colS}>\delta^6\cdot \eps^2\cdot \Delta$ then, there are at least $\delta^6\cdot \eps^2\cdot \Delta$ colors not used by any neighbor of $v$. This is because, the value of $\slack{v}{\colS}$ shows the minimum number of available colors left for $v$ even when all vertices of $\Free{v}{\colS}$ receive distinct colors by $\varphi$.
    If $\deg_S(v) <\Delta-\delta\cdot \eps^2\cdot \Delta$, then there are at least $\delta\cdot\eps^2\cdot \Delta$ many colors that are not used by any sparser neighbor of $v$.
    
    Thus, the probability that $v$ does not get colored in $\ColorSparse$ is at most
    \begin{align*}
        \prob{v \text{ does not get colored}}
        &\leq (1-\frac{\delta^6\cdot \eps^2\cdot \Delta}{\Delta+1})^{\Delta\cdot \log n/(\delta^8\cdot \eta)}\tag{as $\slack{v}{\colS}>\delta^6\cdot \eps^2\cdot \Delta$}\\
        &\leq \exp\paren{-\frac{\eps^2\cdot \Delta\log n}{\delta^2\cdot \eta}}\tag{as $1-x\leq e^{-x}$ for $x \in [0,1]$}\\
        &\leq n^{-1/\delta}\tag{as $\eps^2\gg\delta$ and $\Delta>\eta$}
    \end{align*}
    Similarly, when $v\in V(\KKoh)$, assuming $v$ does not have \Cref{flaw:out-heavy} we have $\slack{v}{\colS}\geq \delta^6\cdot \eps^2\cdot \eta$ or $\deg_S(v) <\Delta-\delta\cdot \eps^2\cdot \eta$ and thus, there are at least $\delta^6\cdot \eps^2\cdot \eta$ colors not used by any sparser neighbor of $v$ in $\varphi$, therefore we have 
    \begin{align*}
        \prob{v \text{ does not get colored}}
        &\leq (1-\frac{\delta^6\cdot \eps^2 \cdot \eta}{\Delta+1})^{\Delta\cdot \log n/(\delta^8\cdot \eta)}\\
        &\leq \exp\paren{-\frac{\eps^2\cdot\log n}{\delta^2}}\leq n^{-1/\delta}.
    \end{align*}
    Therefore, with high probability $v$ successfully gets colored in $\ColorSparse$.

    As a random color is used by $O(1/(\delta^6\cdot\eps^2))=O(1)$ denser neighbors of $v$ in expectation, $v$ uncolors at most $O(1)$ denser neighbors in expectation.
\end{proof}

We also prove the worst-case runtime of this subroutine with the bound mentioned in \Cref{lem:worstcase-subroutines}.
\begin{proof}[Proof of \Cref{lem:worstcase-subroutines} for \ColorSparse]
    In $\ColorSparse(v)$, as an alternative to checking color classes we can always check the color of all neighbors of $v$. Whenever the runtime of $\ColorSparse(v)$ reaches the limit of $m$, go over all neighbors of $v$ (which is trivially bounded by the number of edges). Thus, the total runtime is $O(m)$ in the worst-case.   
\end{proof}

\subsection{The \ColorOH Subroutine}

The subroutine $\ColorOH$ is called when an almost-clique $K\in \KKoh$ have many uncolored vertices. The algorithm uses random color trials to color uncolored vertices in $K$. 

\begin{Algorithm}\label{alg:color-OH} 
$\ColorOH(K):$  
\begin{enumerate}
    \item Let $B$ be the set of uncolored vertices in $K$.
	\item While $\abs{B}> {4\cdot \avgout(K)}$, let $v\in B$ be a vertex maximizing $\out(v)$, and do:
	\begin{enumerate}
        \item Run $\RandomColorTrial(v)$, and remove $v$ from $B$.
    \end{enumerate}
\end{enumerate}
\end{Algorithm}

\begin{lemma*}[Restatement of \Cref{lem:colorOH}]
	Let $K \in \KKoh$ be any out-heavy almost-clique in the decomposition with at least $4\Delta/5$ uncolored vertices in $\varphi$. Calling $\ColorOH(K)$ extends the coloring $\varphi$ to color all but ${4\cdot \avgout(K)}$ vertices of $K$ in expected $O(n\log{n})$ time without recoloring or uncoloring any vertex.
\end{lemma*}
\begin{proof}
    For each $v\in K$, by definition, we have
    \[
    \deg(v) = \abs{K} - \antiDegree(v) + \out(v) -1 \leq \Delta.
    \]
    Averaging this inequality over all $v\in K$ yields
    $\abs{K} + \avgout(K) - \avgantiDegree(K) \leq \Delta+1.$
    Since $K\in\KKoh$, we have $\avgantiDegree(K)\leq 2\eta/\eps^2\leq \avgout(K)/2$. Thus, we have 
    \begin{equation}\label{eq:bounding-K}
        \abs{K} \leq \Delta+1 - \frac{\avgout(K)}{2}.  
    \end{equation}

    Initially, $\abs{B}\geq 4\Delta/5$ and $\abs{K}\leq (1+\eps)\Delta$, hence the number of initially colored vertices in $K$ is at most $\eps\Delta+\Delta/5$.

    Let $s:=\abs{B}$ and $v_1, \ldots, v_s$ be all vertices in the initial set $B$ in non-increasing order of $\out(\cdot)$. We have $s\geq 4\Delta/5$. Since $\avgout(K)$ is the average of all out degrees in $K$, at least half of the vertices $K$ have $\out(v)\leq 2\cdot \avgout(K)$. Thus, there are at most $\abs{K}/2$ many indices $i\in [s]$ such that $\out(v_i)>2\cdot \avgout(K)$.
    Let $B_i$ denote the set of vertices remaining in $B$ at the time $v_i$ is processed.
    When \ColorOH reaches $v_i$, we have
    \[
    \avail{v_i}{\varphi} \geq \Delta+1 - \abs{K\setminus B_i} - \out(v_i). 
    \]
    This is because all the colored neighbors of $v_i$ are in the set $(K\setminus B_i) \cup \OUT(v_i)$.

    \textbf{Case 1:} For an index $i\leq\abs{K}/2$, we know 
    \begin{align*}
    \avail{v_i}{\varphi} 
    &\geq \Delta+1 - \abs{K\setminus B_i} - \out(v_i)\\
    &\geq
    \Delta+1 - (i+(\eps+1/5)\cdot \Delta) -\eps\Delta \tag{as initially colored vertices in $K$ are at most $(\eps+1/5)\cdot \Delta$, and $\out(v_i)\leq\eps\Delta$}\\
    &\geq
    \frac{4\Delta}{5} - \frac{\abs{K}}{2}-2\eps\Delta \tag{as $i\leq \abs{K}/2$}\\
    &\geq 
    \frac{\Delta}{5} -3\eps\Delta.\tag{as $\abs{K}\leq (1+\eps)\Delta$} 
    \end{align*}

    Therefore these vertices have at least $(1/5-3\eps)\Delta$ many available colors. Consequently, as the algorithm runs \RandomColorTrial , by \Cref{prop:random-color-trial} the time spent for each of these vertices is $O(n/\Delta)$
    in expectation.

    \textbf{Case 2:} For an index $i>\abs{K}/2$, we know that $\out(v_i)\leq 2\cdot \avgout(K)$. At the time $v_i$ is processed,
    \begin{align*}
    \avail{v_i}{\varphi}  
    &\geq \Delta+1 - \abs{K} + \abs{B_i} - \out(v_i)\\
    &\geq s-i + \avgout(K)/2 -\out(v_i)\tag{as $\abs{B_i}=s-i$ and by \Cref{eq:bounding-K}}\\
    &\geq s-i - 4\cdot \avgout(K).\tag{as $\out(v_i)\leq 2\cdot \avgout(K)$}
    \end{align*}
    In particular, the success probability of a random color trial is at least
    \[
    \frac{s-i - 4\cdot \avgout(K)}{\Delta+1}.
    \]
    Thus, the expected number of samples needed is $O(\Delta/(i - 4\cdot \avgout(K)))$.
    We continue this process as long as $i < s- 4\avgout(K)$, ensuring that sufficient available colors remain. Consequently, at most $4\cdot \avgout(K)$ vertices in $K$ remain uncolored.

    Finally, as the algorithm uses \RandomColorTrial, by \Cref{prop:random-color-trial},  
    the time spent for the $i$-th vertex is $O(n/(s-i - 4\cdot \avgout(K)))$ in expectation.
    Therefore, the expected running time is
    \begin{align*}
    \sum_{i\le |K|/2} O(\frac{n}{\Delta})
    +
    \sum_{i>|K|/2}^{s-4 \cdot \avgout(K)-1}
    O\left(\frac{n}{s-i - 4\cdot \avgout(K)}\right)\\
    = \frac{(1+\eps)\cdot \Delta}{2}\cdot O(\frac{n}{\Delta}) + O\paren{\sum_{j=1}^{s-\abs{K}/2+1} \frac{n}{j}}\tag{as $\abs{K}\leq (1+\eps)\Delta$}\\
    = O(n\log n),
    \end{align*}
    where we get the final bound by the bounds on harmonic series, concluding the proof.
\end{proof}

We also prove the worst-case runtime of \ColorOH with the bound mentioned in \Cref{lem:worstcase-subroutines}.
\begin{proof}[Proof of \Cref{lem:worstcase-subroutines} for \ColorOH]
    In $\ColorOH(K)$ for some almost-clique $K$, as an alternative to checking color classes in \RandomColorTrial we can always check the colors used on endpoints of all edges incident on every vertex $v\in K$. Therefore, we can deterministically find an available color for each of them. Whenever the runtime of $\ColorOH(K)$ reaches the limit of $m$, we can stop the algorithm and go over all edges to do so. Thus, the total runtime is $O(m)$ in the worst-case.   
\end{proof}

%% file: dense.tex

\section{Handling Dense Vertices}\label{sec:dense}

In this section, we fill out the remaining parts of the proof of~\Cref{thm:delta+1} by describing and analyzing how the algorithm handles denser vertices. 
This part of the proof is similar to the approach of~\cite{FlinH25} and we borrow many ideas from that work. We provide these parts mainly for completeness as the way the algorithm of~\cite{FlinH25} and ours are structured globally is quite 
different and thus we need to carefully modify different steps in~\cite{FlinH25} to make sure they work for us as well. 

\newcommand{\matched}{\ensuremath{\texttt{matched}}\xspace}

Before moving on to the details of the subroutines, we note that throughout \Cref{alg:delta+1} we maintain the following related data structures:
\begin{itemize}
	\item An array $\matched$ for all matched vertices in colorful anti-matching(\Cref{def:colorful-anti-matching}; see also \Cref{inv:dense}), where for each vertex $u$, if $u$ is unmatched we have $\matched[u]=\perp$ and if it is matched to some vertex $v$ we have $\matched[u]=v$.
	\item  A list of available colors $\Avail{K}{\varphi}$ for each almost-clique $K$, which shows all the colors in $[\Delta+1]$ that is not used by any vertex in $K$.
\end{itemize}
Note that both of these data structures can be updated in constant time throughout \Cref{alg:delta+1} whenever a vertex is assigned a new color.

\subsection{The \ExpandMatching Subroutine} 

The goal of $\ExpandMatching$ is to ensure~\Cref{inv:dense} holds for a given almost-clique $K$. 

\begin{Algorithm}\label{alg:expand-matching} 
$\ExpandMatching(K):$ for ensuring~\Cref{inv:dense} for an almost-clique $K$. 
\begin{enumerate}
	\item While $\abs{M_K} < \floor{8\cdot \avgantiDegree(K)}$ do the following:
\begin{enumerate}
	\item Sample an anti-edge $(u,v)$ from anti-edges in $K$ uniformly at random\footnote{This can be done by sampling $v \in K$ proportional to ${\antiDegree(v)}$, and picking a non-edge from $\ANTIDegree(v)$ uniformly.}. If $(u,v)$ is incident on $M_K$, repeat this step, otherwise continue. 
	\item Sample a color $c \in [\Delta+1]$ uniformly at random. Go over all vertices $w$ with color $\varphi(w)=c$ to check if (1) $c$ is used on more than one vertex in $K$, or (2) $c$ is used to color a vertex $w \in \OUT(u) \cup \OUT(v)$; if so, repeat this step, otherwise continue. 
	\item Color $\varphi(u) = \varphi(v) = c$, add $(u,v)$ to $M_K$ and if there is a single vertex $w \in K$ with color $\varphi(w) = c$, uncolor $w$ to $\varphi(w) = \perp$. 
\end{enumerate}
\item Return all uncolored vertices.
\end{enumerate}
\end{Algorithm} 

Given~\Cref{alg:expand-matching}, we can now restate and prove~\Cref{lem:anti-matching}. 

\begin{lemma*}[Restatement of~\Cref{lem:anti-matching}]
	Let $K$ be an almost-clique in the decomposition violating~\Cref{inv:dense}. The subroutine $\ExpandMatching(K)$ increases
	size of $M_K$ until it again satisfies~\Cref{inv:dense} in $O(b \cdot n/\Delta)$ expected time where $b$ is the number of anti-edges added to $M_K$. 
	This subroutine may recolor vertices of $K$ in $\varphi$, while ensuring properness of the (partial) coloring $\varphi$, and creates $O(b)$ new uncolored vertices in $K$ which it also returns. 
\end{lemma*}
\begin{proof}
	Each time the while-loop proceeds, we have $\abs{M_K} < \floor{8\cdot \overline{\antiDegree}(K)}$. Also, we know there are total of $\abs{K}\cdot \overline{\antiDegree}(K)/2$ many anti-edges inside $K$. Moreover, as each vertex $v\in K$ has at least $(1-\eps)\Delta$ neighbors in $K$ and $\abs{K}\leq(1+\eps)\Delta$, then each matched anti-edge inside $M_K$ is incident on at most $4\eps\Delta$ other anti-edges, which means the total number of anti-edges incident on the anti-matching is at most 
	\[
	4\eps\Delta\cdot \abs{M_K}<  4\eps\Delta\cdot \floor{8\cdot \overline{\antiDegree}(K)}\leq 32 \eps\cdot (1-\eps)^{-1}\cdot \abs{K} \cdot \overline{\antiDegree}(K), 
	\] 
	where the last inequality holds as $\abs{K}\geq (1-\eps)\Delta$.
	Therefore, for small enough $\eps$, at least half of the anti-edges inside $K$ can be added to the matching as they are not incident on $M_K$. Thus, \ExpandMatching finds such anti-edge in expected constant time. Also, after sampling an anti-edge $(u, v)$, we can check if any of its endpoints are matched using the data structure $\matched$ in constant time. 

	After finding a proper anti-edge $(u, v)$ that is not incident on $M_K$, the algorithm samples a color $c$ and checks if it is used by any vertex in $\OUT(v)\cup \OUT(u)$ or by more than one other vertex in $K$. This can be done by going over the color class $c$ in expected $O(n/\Delta)$ time. Moreover, there are at most $\abs{K}/2\leq (1+\eps)\Delta/2$ many colors used by more than one vertex in $K$ and by properties of almost-cliques we have $\abs{\OUT(v)\cup \OUT(u)}\leq 2\eps\Delta$. Therefore, there are at least $\Delta/2 - 3\eps\Delta\geq \Delta/3$ available colors for small enough $\eps$. Consequently, by \Cref{prop:random-color-trial} we find such proper color $c$ in $O(n/\Delta)$ expected time.
	
	After successfully finding a color $c$, there is at most one vertex $w\in K$ which we have to recolor in order to `steal' its color and add $(u, v)$ to the matching. Therefore, \ExpandMatching spends $O(n/\Delta)$ expected time to increase $M_K$ by one and uncolors at most one new vertex.
	Therefore, if \ExpandMatching increases $M_K$ by $b$, the expected runtime is $O(b\cdot n/\Delta)$ and the expected recourse is at most $O(b)$,
	concluding the proof.
\end{proof}

We also prove the worst-case runtime of this subroutine with the bound mentioned in \Cref{lem:worstcase-subroutines}.
\begin{proof}[Proof of \Cref{lem:worstcase-subroutines} for \ExpandMatching]
    We can increase the size of the anti-matching by going over all $m$ edges and finding the ones that are incident to vertices in the almost-clique $K$. As a consequence we will find anti-edges in $K$ and the colors incident on them. Therefore, we can extend the matching size by picking an available color for an anti-edge that is not incident on $M_K$.
	
	Thus, whenever the runtime of $\ColorSparse(v)$ reaches the limit of $m$, we go over all edges and do as described above and the total runtime is $O(m)$ in the worst-case.   
\end{proof}

\subsection{The \ColorDense Subroutine} 

Finally, we have the $\ColorDense$ subroutine that, assuming~\Cref{inv:dense}, can color all remaining dense vertices. For each almost-clique $K$, vertices 
\[
I_K := \set{v \in K \mid \out(v) \leq 4\cdot \avgout(K) \wedge \antiDegree(v) \leq 4 \cdot \avgantiDegree(K)}
\] 
are called \textbf{inliers} and the remaining vertices $O_K := K \setminus I_K$ are called \textbf{outliers}. This subroutine treats outlier and inlier vertices separately. 

\clearpage

\begin{Algorithm}\label{alg:color-dense}
	$\ColorDense(v):$ for coloring dense vertex $v$. 
    
    \begin{enumerate}
	\item Let $K:=K(v)$. If $v\in I_{K}$ and $K\in \KKah$, sample a color $c\in [\Delta+1]$ uniformly at random. Go over all vertices $w$ with color $\varphi(w)=c$ to see if there is any in $N(v)$; if so repeat this step, otherwise set $\varphi(v)=c$ and terminate. 
	\item If $v\in I_{K}$ and $K\in \KKol$, go over all neighbors of $v$ in $\OUT(v)$ and find a color $c\in [\Delta+1] \setminus \varphi(\OUT(v))$. Set $\varphi(v)=c$ and terminate.
	
	\item \label{item:sampling-line-dense} Otherwise, if $v$ is an outlier, sample a color  $c\in [\Delta+1]$ uniformly at random, repeat this
    steps until a color $c$ that is accepted by all of \textup{(R1)--(R3)} is found:
	\begin{enumerate}
	\item[\textup{(R1)}] If $c \notin M_K$, reject $c$ and return to step~\ref{item:sampling-line-dense}.
	\item[\textup{(R2)}] If there exists $w \in N(v) \setminus I_K$ with $\varphi(w) = c$, reject $c$ and return to step~\ref{item:sampling-line-dense}.
	\item[\textup{(R3)}] If $|\{w \in I_K : \varphi(w) = c\}| \geq 2$, reject $c$ and return to step~\ref{item:sampling-line-dense}.\footnote{All \textup{(R1)--(R3)} can be checked by going over the vertices in color class $c$}
    \end{enumerate}
	\item Once $c$ is accepted, set $\varphi(v) = c$. Then,
    go over the color class of $c$, if there exists a vertex $w \in I_K$ with $\varphi(w) = c$, let $\varphi(w)=\perp$ and call $\ColorDense(w)$.
	\end{enumerate}
\end{Algorithm}

Given \Cref{alg:color-dense}, we can now restate and prove \Cref{lem:dense}.

\begin{lemma*}[Restatement of~\Cref{lem:dense}]
	Let $v \in D$ be a denser vertex with $\varphi(v) = \perp$ such that~\Cref{inv:dense} holds for the almost-clique $K(v)$ that $v$ belongs to. The subroutine $\ColorDense(v)$ extends the coloring $\varphi$ to color
	$v$ as well in expected $O(n/\Delta + n/\eta+\eta)$ time by recoloring vertices inside $K$, while ensuring properness of the (partial) coloring $\varphi$, without creating new uncolored vertices. 
\end{lemma*}
\begin{proof}
	We first consider the case where $v$ is an inlier and $K\in \KKah$.
	We bound the number of colors available for $v$ inside $\Avail{K}{\varphi}$ that is not used by any vertex in $\OUT(v)$.
	As $\Avail{K}{\varphi}$ denotes the set of colors not used by vertices in $K$ we have,
	\[
	\Avail{K}{\varphi}\setminus \varphi(\OUT(v)) = [\Delta+1]\setminus \varphi(K \cup \OUT(v)).
	\]
	Since $v$ itself is currently uncolored, the number of distinct colors used on $K$ is at most $|K|-1-|M_K|$, as $M_K$ denotes the set of repeated colors in $K$. Therefore,
	\begin{align*}
	\abs{[\Delta+1]\setminus\set{\varphi(K\cup \OUT(v))}}
	&= \Delta+1 - \abs{\varphi(K\cup \OUT(v))}\\
	&\geq \Delta+1 - ((\abs{K}-1-\abs{M_K}) + \out(v)).
	\end{align*}
	Using the fact that  $|K| = \deg(v) + 1 + \antiDegree(v) - \out(v)$, we obtain
	\begin{align*}
	\abs{[\Delta+1]\setminus\set{\varphi(K\cup \OUT(v))}}
	&\geq \Delta+1 - (\deg(v)+\antiDegree(v) - \abs{M_K})\\
	&\geq \abs{M_K} - \antiDegree(v)+1, 
	\end{align*}
	where the last inequality uses $\deg(v)\le \Delta$.

	By \Cref{inv:dense}, we have $\abs{M_K} = \floor{8\cdot \avgantiDegree(K)}$. Since $v$ is an inlier, $\antiDegree(v) \leq 4\cdot \avgantiDegree(K)$, and thus the above quantity is at least $4\cdot \avgantiDegree(K)>4\cdot \eta/\eps^2$.
	Hence, there are at least $\Omega(\eta)$ many available colors for $v$ and in expectation after $O(\Delta/\eta)$ many trials we sample an available color for $v$. Each color class has expected size $O(n/\Delta)$, therefore by \Cref{prop:random-color-trial} this step takes $O(n/\eta)$ in expectation.

	When $v\in I_K$ and $K\in \KKol$, we know that $\out(v)\leq \avgout(K)=O(\eta)$ therefore we can go over all vertices in $\OUT(v)$ and find the first available color in $\Avail{K}{\varphi}$ in $O(\eta)$ time. 
	Similar to above, we have
	\begin{align*}
	\abs{[\Delta+1]\setminus\set{\varphi(K\cup \OUT(v))}}
	&\geq  1 + \abs{M_K} - \antiDegree(v).
	\end{align*}
	We know by \Cref{inv:dense} that $\abs{M_K}=\floor{8\cdot \avgantiDegree(K)}$ and as $v\in I_K$ we have $\antiDegree(v)\leq 4\cdot \avgantiDegree(K)\leq \abs{M_K}/2$. 
	Hence, there exists a color $c \in \Avail{K}{\varphi}$ that is not used by any vertex in $\OUT(v)$, and we can assign $c$ to $v$.

	\medskip
	
	Next, consider the case where $v$ is an outlier. The algorithm samples colors uniformly at random from $[\Delta+1]$ and checks whether a sampled color satisfies conditions (R1)--(R3). Each check can be performed by scanning the color class of $c$. Since a uniformly random color has expected class size $n/\Delta$, verifying a sampled color takes $O(n/\Delta)$ time in expectation.

	We now argue that at least a constant fraction of all colors are valid. A color is \emph{proper} if it satisfies (R1)--(R3), i.e., it is:
	(i) not in $M_K$,
	(ii) not used by any vertex in $\OUT(v)$,
	(iii) not used by any outlier neighbor of $v$ in $K$, and
	(iv) used by at most one inlier in $K$.

	We upper bound the number of colors violating each condition. By \Cref{inv:dense}, $\abs{M_K}=\floor{8\cdot \overline{\antiDegree(K)}}\leq 8\eps\Delta$. By the definition of almost-cliques, at most $\eps\Delta$ colors are used by $\OUT(v)$. 
	By Markov's inequality, the number of outliers in $K$ is at most $|K|/4 \leq (1+\eps)\Delta/4$, and the number of colors used by more than one inlier is at most $\abs{I_K}/2 \leq (1+\eps)\Delta/2$. Therefore, the number of forbidden colors is at most
	\[
	8\eps\Delta + \eps\Delta + \frac{(1+\eps)\Delta}{4} + \frac{(1+\eps)\Delta}{2}.
	\]
	For sufficiently small constant $\eps$, the number of remaining colors is at least $\Delta/8$. Therefore, by \Cref{prop:random-color-trial} the expected time needed to find an available color is $O(n/\Delta)$.

	\medskip

	Finally, note that the procedure may recurse on at most one inlier vertex $w$. As shown above, $\ColorDense(w)$ successfully colors the inlier vertex $w$ without creating new uncolored vertices. Therefore, the overall procedure colors $v$ while maintaining properness and does not introduce any new uncolored vertex.

	The total expected running time is $O(n/\Delta + n/\eta+\eta)$, as claimed.
	\end{proof} 

Now, we complete the proof of \Cref{lem:worstcase-subroutines}	
by proving the worst-case runtime of \ColorDense.
\begin{proof}[Proof of \Cref{lem:worstcase-subroutines} for \ColorDense]
	For every $v\in D$ we can wither find an available color for $v$ or find an inlier neighbor for it which it can steal a color from by going over all its neighbors. As $\abs{N(v)}\leq m$, we can run $\ColorDense(v)$ in $O(m)$ worst-case time by running the described process whenever $\ColorDense$ reaches the limit of $m$ time.
\end{proof}

%% file: missingProofs.tex
\section{Omitted Proofs}\label{sec:omitted-proofs}

We present the proofs of the results missing from the main body of the paper in this appendix.

\subsection{Proof of \Cref{cor:locally-sparse}}\label{sec:locally-sparse}

\begin{corollary*}[Restatement of \Cref{cor:locally-sparse}]
	There is a randomized algorithm that given any $\gamma \in (0,1)$ on any $n$-vertex dynamic $\gamma$-sparse graph with maximum degree $\Delta$ (sufficiently large w.r.t. $\gamma$) maintains a $(1-\Theta(\gamma))\Delta$-coloring  
	against an adaptive adversary in $\Ot(\sqrt{n}/\gamma)$ amortized update time with high probability.
\end{corollary*}
\begin{proof}[Prof sketch]
We run~\Cref{alg:delta+1} without the decomposition and let $S$ be the set of all vertices of the graph, which, by the assumption in the corollary, are all $\gamma$-sparse. 
We also slightly modify $\SampleColor$ so it picks its colors randomly from $(1-\Theta(\gamma))\Delta$ instead of $(\Delta+1)$. 

The proof of \Cref{lem:prob-slack-sparse} continues to apply after this modification to provide a slack of $\Omega(\gamma \cdot \Delta)$ for vertices. For concreteness, let $c > 0$ be the hidden constant in this $\Omega$ notation. 
Since we have generated $c \cdot \gamma \cdot \Delta$ slack and maximum degree is $\Delta$, if let the number of available colors be $(1-(c/2) \cdot \gamma \cdot \Delta)$, then, it means even if we use different colors on all remaining 
neighbors of a vertex, it will still have $(c/2) \cdot \gamma \cdot \Delta$ available colors. Thus, running $\Attempt$ in $\Ot(\frac{n}{\gamma \Delta})$ time with high probability can color any vertex with such a slack. Alternatively, we can also go over 
all neighbors of a vertex and given the slack find an available color to color the vertex in $O(\Delta)$ time. 

We also need to handle the dynamic entropy compression argument. For that,~\Cref{lem:prob-slack-sparse} provides a probability of $2^{-\Omega(\gamma \Delta)}$ for a \Resample call to happen, 
which allows us to bound the total number of $\Resample$ calls over $t$ updates by $O(t/(\gamma \cdot \Delta))$ using exactly the previous arguments. Finally, each call to $\Resample$ can be implemented either in $O(n/\gamma)$ expected time as 
before, or in $O(\Delta^2)$ time by checking neighbors of each vertex directly to see if the color can be assigned (instead of checking the color class). 

All in all, these imply an algorithm that with high probability takes 
\[
	\Ot\paren{\frac{t}{\gamma \cdot \Delta} \cdot \min\set{n,\Delta^2}}
\]
time for handling $t$ updates (the high probability part is also achieved by~\Cref{prop:KQ-inequality} exactly as before). Balancing out for $\Delta=\sqrt{n}$, the amortized update time is $\Ot(\sqrt{n}/\gamma)$ as desired. 
\end{proof}

\subsection{Proof of~\Cref{lem:OH-is-out-sparse}}\label{sec:lem:OH-is-out-sparse}
We have the following result due to~\cite{FlinH25} for out-heavy almost-cliques (their result also applies to anti-heavy ones but for our purpose, this is not needed). 

\begin{lemma}[\!\!{\cite[Lemma 3.6]{FlinH25}}]\label{prop:almost-clique-to-sparse}
	Fix any sparse-dense decomposition and its classification. Any vertex in an out-heavy almost-clique $K$ is $\gamma$-sparse for 
	\[
	\gamma = \Omega(\eps^2) \cdot \frac{\avgout(K)}{\Delta}.  
	\]
\end{lemma}
Using this lemma, we can move to the proof of \Cref{lem:OH-is-out-sparse}.
\begin{lemma*}[Restatement of~\Cref{lem:OH-is-out-sparse}]
		Fix any sparse-dense decomposition and its classification. Any vertex $v$ in an out-heavy almost-clique is $\gamma$-sparse for
	\[
		\gamma = \Omega(\eps^2) \cdot \paren{\frac{\avgout(K) + \out(v)}{\Delta}}.
	\] 
\end{lemma*}
\begin{proof}
	If $\out(v)\leq 4\cdot \avgout(K)$, then by \Cref{prop:almost-clique-to-sparse} we know $v$ is $\paren{\Omega(\eps^2) \cdot \frac{\avgout(K)}{\Delta}}$-sparse, and therefore it is $\gamma$-sparse as well.

	Assume $\out(v)>4\cdot \avgout(K)$.
	Let $X:=N(v)\cap K$ be the neighbors of $v$ in $K$. If there exists $\eps^2\cdot \out(v)\cdot\Delta$ anti-edges with one endpoint in $X$ and another in $\OUT(v)$, then $v$ has at most ${\Delta \choose 2}- \eps^2\cdot \out(v)\cdot\Delta$ edges in its neighborhood, and therefore $v$ is $\paren{\Omega(\eps^2)\cdot \frac{\out(v)}{\Delta}}$-sparse. So we assume there are less than $\eps^2\cdot \out(v)\cdot \Delta$ many anti-edges between $X$ and $\OUT(v)$. Then, at least $\abs{X}\cdot\out(v)-\eps^2\out(v)\cdot \Delta$ many edges exists from $X$ to $\OUT(v)$. Therefore, we get
	\begin{align*}
		\avgout(K) 
		& = \frac{\sum_{v\in V}\out(v)}{\abs{K}}\\
		&\geq \frac{\abs{X}\cdot\out(v)-\eps^2\out(v)\cdot \Delta}{\abs{K}}\tag{by assumption}\\
		&\geq \frac{(1-\eps)\Delta\cdot\out(v)-\eps^2\out(v)\cdot \Delta}{(1+\eps)\Delta}\tag{as $\abs{K}\leq (1+\eps)\Delta$ and $\abs{X}\geq (1-\eps)\Delta$ by definition of almost-cliques}\\
		&\geq \frac{(1-\eps-\eps^2)\cdot \out(v)}{(1+\eps)}\geq \frac{(1-\eps-\eps^2)}{(1+\eps)}\cdot 4\cdot \avgout(K),\tag{as $\out(v)>4\cdot \avgout(K)$}
	\end{align*}
	which is a contradiction as $\eps$ is sufficiently small. Therefore, $v$ is $\paren{\Omega(\eps^2)\cdot \frac{\out(v)}{\Delta}}$-sparse and therefore $\gamma$-sparse as desired.
\end{proof}

\subsection{Proof of~\Cref{lem:decomp}}\label{sec:lem:decomp}

\begin{lemma*}[Restatement of~\Cref{lem:decomp}]
    Let $\eps > 0$ be small and $\eta = o(\eps^2 \Delta)$. 
    There is a fully dynamic algorithm that with high probability maintains 
    a sparse-dense decomposition of $G$ for $\eps$ together with its $\eta$-classification. 
    The algorithm provides all guarantees of~\Cref{prop:maintaining-decomposition} and additionally the number of almost-cliques that change from one set of the $\eta$-classification to another set is amortized $O(\frac{1}{\eps^2 \eta \Delta})$ per update. 
\end{lemma*}

\begin{proof}
	The algorithm in~\Cref{lem:decomp} runs~\Cref{prop:maintaining-decomposition} to maintain the decomposition and thus satisfies the guarantees of~\Cref{prop:maintaining-decomposition}. 
	In addition, it maintains an $\eta$-classification $\KKah \sqcup \KKoh \sqcup \KKol$ using the following procedures. 
	
	When an almost-clique $K$ is \emph{formed}, the algorithm adds it to one of the classes using the following criteria: 
	\begin{itemize}
		\item to $\KKah$ if $\avgantiDegree(K) \geq 2\eta/\eps^2$; 
		\item to $\KKoh$ if $\avgantiDegree(K) < 2\eta/\eps^2$ and $\avgout(K) \geq 6\eta/\eps^2$; 
		\item to $\KKol$ otherwise. 
	\end{itemize}
	When an almost-clique $K$ \emph{collapses}, it is deleted from its class. Finally, after each update, the algorithm \emph{re-evaluates} the class of each almost-clique $K$ and removes it from its current list if: 
	\begin{itemize}
		\item from $\KKah$ if $\avgantiDegree(K) < \eta/\eps^2$; 
		\item from $\KKoh$ if $\avgantiDegree(K) \geq 2\eta/\eps^2$ or $\avgout(K) < 4\eta/\eps^2$; 
		\item from $\KKol$ if $\avgantiDegree(K) \geq 2\eta/\eps^2$ or $\avgout(K) > 6\eta/\eps^2$.
	\end{itemize}
	A removed almost-clique $K$ is then added to another class using the same criteria mentioned earlier as if it was just formed. 
	
	Checking that $\KKah \sqcup \KKoh \sqcup \KKol$ is an $\eta$-classification is now straightforward (through a somewhat tedious verification of each case) and we omit the proof. We now bound the amortized number of change in classifications of almost-cliques. 
	
	Consider an almost-clique $K$ that is just added to $\KKah$. We know $\avgantiDegree(K) \geq 2\eta/\eps^2$ which means it has at least $(1-\eps) \cdot {\eta\Delta}/\eps^2$ anti-edges inside it (by the lower bound on size 
	of $K$ in the decomposition). For $K$ to be removed from $\KKah$ later, its $\avgantiDegree(K)$ needs to go below $\eta/\eps^2$, which requires its number of anti-edges to become at most $(1+\eps)\cdot \eta\Delta/\eps^2$ 
	(by the upper bound on the size of $K$). This means at least $\Omega(\eta\Delta/\eps^2)$ changes to $\ANTIDegree$-sets is needed before $K$ needs another change of class. Since amortized number of changes to $\ANTIDegree$-sets 
	by~\Cref{prop:maintaining-decomposition} is $O(1/\eps^4)$ per update and each change contributes to exactly one almost-clique, the almost-clique $K$ needs $\Omega(\eps^2 \cdot \eta\Delta)$ dedicated updates before it changes class. In other words, 
	the amortized number of times $K$ moves in and out of $\KKah$ is $O(1/(\eps^2\eta\Delta))$ per update. 
	
	A similar analysis can be used to bound the changes between any two other classes as well (by additionally noting that an update to an edge, instead of an anti-edge, can also only affect two almost-cliques at most). 
	Since the overhead of maintenance of the $\eta$-decomposition is just a constant factor over the decomposition, we can conclude the proof. 
\end{proof}

\subsection{Proof of \Cref{prop:KQ-inequality}}\label{app:KQ-inequality}

Different forms of the following proposition appear in the literature, see, e.g.,~\cite{freedman1975tail,ChekuriQ18,BenderFGKY24,kuszmaul2025}; we use the 
following formulation which is an immediate extension of a result of~\cite{kuszmaul2025} that proves this for $\gamma=0$. 

\begin{proposition*}[{\cite{kuszmaul2025}}]
	Suppose Alice constructs a sequence of random variables $X_1, \ldots, X_n$, with $X_i\in [0, c]$, where $c>0$, using the following iterative process. For each $i\in [n]$, once the outcome of $X_1, \ldots, X_{i-1}$ are determined, Alice then selects the probability distribution $\mathcal{D}_i$ from which $X_i$ will be drawn; $X_i$ is then drawn from $\mathcal{D}_i$. Alice is an adaptive adversary in that she can adapt $\mathcal{D}_i$ to the outcome of $X_1, \ldots, X_{i-1}$. The only constraint on Alice is that
	if we define $a_i := \expect{X_i\mid \mathcal{D}_i}$ and $A = \sum_i a_i$, then, for some $\gamma \in (0,1)$ and $\mu > 0$, $\Pr\paren{A > \mu} \leq \gamma$; 
	that is, the sum of the means of the probability distributions $\mathcal{D}_1, \ldots, \mathcal{D}_n$ must be at most $\mu$ with probability $1-\gamma$. If $X=\sum_i X_i$, then for any $\delta>0$, we have
	\[
	\Pr\paren{X\geq (1+\delta)\mu}\leq \gamma + \exp\paren{-\frac{\delta^2\mu}{(2+\delta)c}}.
	\]
\end{proposition*}
\begin{proof}
The proof is by a reduction to the $\gamma=0$ case proven in~\cite[Corollary 11]{kuszmaul2025}. 

Define the following sequence of variables $Y_1,\ldots,Y_n$. Follow Alice's selection of $X_1,\ldots,X_n$ with $Y_1=X_1$ and for each $i$, check if $A$ defined so far by Alice has $A > \mu$ or not; if no, continue setting $Y_i=X_i$ but if yes, 
set $Y_i=Y_{i+1}= \ldots = Y_n = 0$. This is yet another adaptive adversary and the variables $Y_1,\ldots,Y_n$ satisfy the guarantee that 
\[
	\sum_{i=1}^{n} \expect{Y_i\mid \mathcal{D}_i} \leq \mu, 
\]
always. Thus, by the $\gamma=0$ case applied to $Y_1,\ldots,Y_n$, for $Y= \sum_i Y_i$, we have, 
\[
	\Pr\paren{Y \geq (1+\delta)\mu}\leq  \exp\paren{-\frac{\delta^2\mu}{(2+\delta)c}}.
\]
At the same time, we have, 
\[
	\Pr\paren{X \neq Y} = \Pr\paren{A > \mu} \leq \gamma,
\]
by the assumption in the statement of the lemma. By union bound, we have, 
\[
	\Pr\paren{X\geq (1+\delta)\mu} \leq \Pr\paren{X \neq Y} + \Pr\paren{Y \geq (1+\delta)\mu} \leq \gamma + \exp\paren{-\frac{\delta^2\mu}{(2+\delta)c}},
\]
 concluding the proof. 
\end{proof}

\subsection{Proof of \Cref{lem:whp-runtime-other-subroutines}}\label{sec:lem:whp-runtime-other-subroutines}

\begin{lemma*}[Restatement of~\Cref{lem:whp-runtime-other-subroutines}]
    Throughout running \Cref{alg:delta+1} over $t$ updates, with high probability, the time spent and number of uncolored vertices generated by subroutine: 
    \begin{enumerate}
		\item\label{item-attempt} \ColorSparse is $O(t\cdot n\cdot \log^3 n /\eta)$ time and $O(t \cdot \log^2{n})$ vertices; 
		\item\label{item-sampleColor} \SampleColor is $O(t \cdot n\cdot \log n/\Delta)$ time and $O(t\log{n})$ vertices;
        \item\label{item-expand-matching}\ExpandMatching is $O(t\cdot \log^2 n \cdot n/\Delta)$ time and $O(t\log^2{n})$ vertices;
        \item\label{item-dense} \ColorDense is $O(t\cdot \log^2 n\cdot (n/\Delta+\eta))$ time and zero vertices. 
    \end{enumerate}
\end{lemma*}
\begin{proof}
We have already proved part (\ref{item-attempt}) in the main body of the paper. The remaining proofs 
are quite similar and we sketch them here for completeness. Recall that $m$ denotes the \emph{largest} number of edges at any point during the $t$ updates. Clearly, we have $m \leq t$. 

\paragraph{Proof of (\ref{item-sampleColor}):} Let $s$ be the number of times \SampleColor is called directly in \Cref{alg:delta+1} (here we do not consider the calls to \SampleColor invoked inside other subroutines as these calls are considered in the runtime of the corresponding subroutine). For each $i\in [s]$ redefine $X_i$ to be the random variable showing the runtime spent on the $i$-th call to \SampleColor and redefine $Y_i$ to be the number of uncolored vertices caused by this call to \SampleColor.
	
	By \Cref{lem:SampleColor} for each $i\in [s]$ we have $\expect{X_i}=O(n/\Delta)$ and $\expect{Y_i}=O(1)$. Moreover, the number of vertices that move from $D$ to $S$ in the decomposition per update is with high probability $O(\log n)$ by \Cref{lem:decomp}. The other calls to \SampleColor are caused directly by the updates which are at most $O(t)$ many. So with high probability $\expect{X}=O(t\cdot n\cdot \log n/\Delta)$ and $\expect{Y}=O(t\cdot \log n)$. Therefore, by \Cref{prop:KQ-inequality} and using the same argument as Part (\ref{item-attempt}), with high probability the time spent on \SampleColor is $O(t\cdot n\cdot \log n/\Delta)$ and the number of uncolored vertices caused by it is at most $O(t\log n)$.

\paragraph{Proof of (\ref{item-expand-matching}):}
	Assume the total number of times \ExpandMatching is used to increase the size of some anti-matching is $s_m$. For each $i\in[s_m]$ redefine $X_i$ be the random variable showing the time spent in \ExpandMatching for increasing the size of some anti-matching $M_K$ by one. Also, redefine $Y_i$ to show the number of uncolored vertices caused by the same process. 
	
    
    In \Cref{lem:anti-matching} we saw that for each $i\in[s_m]$ we have $\expect{X_i}=O(n/\Delta)$ and $\expect{Y_i}=O(1)$. By \Cref{lem:worstcase-subroutines} we know that $X_i\leq m$, where $m$ is the largest number of edges the graph at a single point in time has over all updates. Also, as there are $n$ vertices in the graph $Y_i\leq n$ trivially holds.
	
	Let $X=\sum_{i\in[s_m]}X_i$ and $Y=\sum_{i\in[s_m]}Y_i$. By linearity of expectation we have $\expect{X} = O(s_m\cdot n/\Delta)$ and $\expect{Y}=O(s_m)$.

	The uncolored vertices that cause a violation in \Cref{inv:dense} and therefore a call to the subroutine \ExpandMatching are caused by \UpdateSparseColor, \ColorSparse, or \SampleColor. By the previous two parts, with high probability \SampleColor causes $O(t\log n)$ many uncolored vertices and \ColorSparse causes $O(t\cdot \log^2 n)$ many uncolored vertices over all updates. By \Cref{lem:recourse-updatesparse-whp} we know that \UpdateSparseColor causes $O(t\cdot \log^2 n)$ many uncolored vertices over all updates with high probability. Therefore, the total number of uncolored vertices that might cause a call to \ExpandMatching and consequently $s_m$ is bounded by $O(t\cdot \log^2 n)$ vertices.

	Using \Cref{prop:KQ-inequality} with the same argument above we get that with high probabilities the total runtime spent on \ExpandMatching is $O(t\cdot n \log n/\Delta)$ and the total number of uncolored vertices it causes over all updates is at most $O(t\cdot \log^2 n)$.
    
\paragraph{Proof of (\ref{item-dense}):} Finally, let $s_d$ be the number of times $\ColorDense$ is called throughout the algorithm. Redefine $X_i$ as the time spent on the $i$-th call to \ColorDense. By \Cref{lem:dense} we know that $\expect{X_i}=O(n/\Delta+\eta)$ for each $i\in [s_d]$. Moreover, by \Cref{lem:worstcase-subroutines} we have $X_i\leq m$. Again, we know that \ColorDense is called when there is an uncolored (denser) vertex, which is caused by one of the subroutines \UpdateSparseColor, \SampleColor, \ColorSparse, or \ExpandMatching.
    
	By \Cref{lem:recourse-updatesparse-whp} and the previous parts of the proof, the number of uncolored vertices left by these subroutines over all updates is $O(t\cdot \log^2 n)$ with high probability. Thus, $\expect{X}\leq s_d \cdot O(n/\Delta+\eta)$, where 
	with high probability $s_d= O(t\cdot \log^2 n)$. So using \Cref{prop:KQ-inequality} and the same argument as above, the total time spent on \ColorDense is $O(t\cdot \log^2 n\cdot \paren{n/\Delta+\eta})$ with high probability. Finally, recall that $\ColorDense$ 
	does not uncolor any vertex. This concludes the proof. 
\end{proof}

\subsection{Proof of \Cref{lem:prob-flaw-sparse}}\label{sec:lem:prob-flaw-sparse}

\begin{lemma*}
    Fix a vertex $v\in S$ along with a set $B(u) \subseteq [\Delta+1]$ and a time $t_u$ associated to each vertex $u\in N_S(v)$. Suppose for every $u\in N(v)$, with probability $\delta$, we sample a color $c(u)$ and assign it to $u$ iff: $(1)$ the color $c(u) \notin B(u)$, and (2) no neighbor $w\in N_S(v)\cap N_S(u)$ with $t_w<t_u$ or with lower lexicographical order while $t_w=t_u$ has also sampled this color. All other
	vertices are marked blank. Let $\chi$ be the partial coloring function obtained this way. 
    Then,  
    \[
    \Pr(\Ncolor{v}{\chi}> 2\delta\Delta)\leq 2^{-\frac{\delta\cdot \Delta}{3}}.
    \]
\end{lemma*}
\begin{proof}
    For each $i\in [\Delta]$, let $X_i$ denote the indicator random variable that is 1 iff the $i$-th neighbor of $v$ exists and is in $S$ and has sampled a color in the described process. Let $X=\sum_{i\in[\Delta]} X_i$. As each vertex $u$ samples a color with probability $\delta$, we have
    \[
    \expect{X}=\sum_{i\in[\Delta]}\Pr(X_i=1)\leq \delta \cdot \Delta.
    \]
    Then by Chernoff bound (\Cref{prop:chernoff}), 
    \begin{align*}
        \Pr\paren{X > 2\delta \cdot \Delta}
        \leq \exp\paren{-\frac{\delta\Delta}{3}}
        \tag{as $\expect{X}\leq \delta \cdot \Delta$}\leq 2^{-\delta\cdot \Delta/3}.
    \end{align*}
    Moreover, any vertex in $N_S(v)$ that keeps a color has first sampled one, hence, $X$ is at least number of vertices in $N_S(v)$ that are colored, namely, $\Ncolor{v}{\chi}$, concluding the proof.
\end{proof}

\subsection{Proof of \Cref{lem:prob-slack-sparse}}\label{sec:lem:prob-slack-sparse}

Here we present the proof of \Cref{lem:prob-slack-sparse}. The proof closely follows the approach in \cite{MolloyR97,MolloyR02}.
Before moving on to the proof we recall the Talagrand's concentration inequality.

\paragraph{Talagrand's inequality:}
A function $f(x_1, \ldots, x_n)$ is called \textbf{$\bm{c}$-Lipschitz} for some $c \geq 0$ iff changing any $x_i$
can affect the value of $f$ by at most $c$. Additionally, $f$ is called \textbf{$\bm r$-certifiable} iﬀ whenever $f(x_1,\ldots , x_n)\geq s$ there are at most $r\cdot s$ variables $x_{i_1},\ldots , x_{i_{r\cdot s}}$ so that knowing the values of these variables certifies $f \geq s$.

\begin{proposition}[Talagrand's inequality; cf.~\cite{MolloyR02}]\label{prop:talagrand}
    Let $X_1, \ldots, X_n$ be $n$ independent random variables and $f(X_1, \ldots, X_n)$ be a $c$-Lipschitz and $r$-certifiable function. For any $t\geq 1$,
    \[
    \Pr\paren{\abs{f-\expect{f}}>t+ 30c\sqrt{r\cdot \expect{f}}}\leq \exp\paren{-\frac{t^2}{8c^2r\expect{f}}}.
    \]
\end{proposition}

We can now prove~\Cref{lem:prob-slack-sparse}. 

\begin{lemma*}[Restatement of~\Cref{lem:prob-slack-sparse}]
    Let $\gamma\in (0, 1)$ be sufficiently small. Fix a $\gamma$-sparse vertex $v$ with $\deg_S(v)>\Delta-\delta\cdot\gamma\cdot \Delta$ along with a set $B(u) \subseteq [\Delta+1]$ and a time $t_u$ associated to each vertex $u\in N_S(v)$ such that $\abs{B(u)}\leq \Delta/4$.
	Suppose for every $u\in N(v)$, with probability $\delta$, we sample a color $c(u)$ and assign it to $u$ iff: $(1)$ the color $c(u) \notin B(u)$, and (2) no neighbor $w\in N_S(v)\cap N_S(u)$ with $t_w<t_u$ or with lower lexicographical order while $t_w=t_u$ has sampled this color. 
	Let $\chi$ be the partial coloring  obtained this way. 
	Then, 
    \[
    	\Pr\paren{\slack{v}{\chi} < \delta^6\cdot\gamma\cdot \Delta} \leq 2^{-\delta^5\cdot \gamma \cdot \Delta}. 
    \]
\end{lemma*}

\begin{proof}[Proof of~\Cref{lem:prob-slack-sparse}]
We prove this lemma by bounding the number of pairs of sparser vertices in neighborhood of $v$ that have received the same color under $\chi$. Let $p(v)$ be the number of pairs in $N_S(v)$ which end up with the same color. By definition we get $\slack{v}{\chi}\geq p(v)$. So it suffices to prove that with high probability $p(v)$ is bounded.

As $v$ is $\gamma$-sparse and $\deg_S(v)\geq \Delta - \delta\cdot \gamma \cdot \Delta$, there are at least $s=\delta^2\cdot \gamma\cdot \Delta^2$ many anti-edges between its sparser neighbors. Let $\bar{e}_1, \ldots, \bar{e}_s$ be $s$ arbitrary anti-edges in $N(v)$. Let $X$ be the number of colors that are sampled by both endpoints of exactly one of these $s$ anti-edges and is also kept by both its endpoints. We can see that $p(v) \geq X$, so it suffices to find a lower bound for $X$. Note that $X$ might be smaller than the number of same color pairs in $N(v)$ as we are considering a more restricted condition.
    
For each anti-edge $\bar{e}_i=(u_i, w_i)$ let $X_i$ be the indicator random variable where $X_i=1$ iff (I) $u_i$ and $w_i$ have sampled same color $c_i$ which is not in $B(u_i)\cup B(w_i)$ and (II) no other vertex in $N_S(v)$ has sampled color $c_i$.
As for each $u\in N(v)$ we have $\abs{B(u)}\leq \Delta/4$, then for each $i\in [s]$ we have
 \begin{equation}\label{eq:avail-cap}
    [\Delta+1]\setminus (B(u_i)\cup B(w_i))
    \geq \Delta+1 - /2\geq \Delta/2.
\end{equation}
Consequently, as $u\in N_S(v)$ is colored by the described process, we have
\begin{align*}
    \Pr\paren{X_i=1}
    & = \Pr\paren{\text{(I) holds }}\cdot \Pr\paren{\text{(II) holds}} \\
    &\geq \paren{\delta^2\cdot \frac{\abs{[\Delta+1]\setminus (B(u_i)\cup B(w_i))}}{\Delta+1}\cdot \frac{1}{\Delta+1}} \cdot (1-\delta\cdot \frac{1}{\Delta+1})^{\Delta-2}\tag{w.p $\delta$ colors are sampled independently for vertices in $N_S(v)$} \\
    & \geq \frac{\delta^2\Delta}{2(\Delta+1)^2} \cdot \exp\paren{-\frac{\delta(\Delta-2)}{(\Delta+1)}} \tag{as $1-x \geq e^{-4x/3}$ for small $x$, and by \Cref{eq:avail-cap}}\\
    &\geq \frac{\delta^2}{8\Delta}.
\end{align*}
Therefore, we have  $\expect{X}\geq s\cdot \frac{\delta^2}{8\Delta}\geq \delta^5\cdot \gamma \cdot \Delta$.

Let $Y$ be the random variable showing the number of colors that are sampled by the endpoints of at least one anti-edge in $\{\bar{e}_1, \ldots, \bar{e}_s\}$, and let $Z$ be the random variable showing the number of colors sampled by the endpoints of at least one anti-edge in $\{\bar{e}_1, \ldots, \bar{e}_s\}$ where the color is also sampled by some other vertex in $N(v)$ or is blocked by its corresponding set $B$. By definition we have $X=Y-Z$.
We can see that both $Y$ and $Z$ are functions of independent random variables that define the choices of random colors for this neighborhood. Moreover, $Y$ is 2-certifiable as it can be verified by knowing the sampled color for two vertices and also $Z$ is 3-certifiable as it can be verified by knowing the sampled color for three vertices or two vertices and element of one of their blocked sets. 
Also, as changing color of a vertex can affect at most two colors involved in $Y$ and $Z$, the new color and the old color, then both are $2$-Lipschitz. We also have
\[
\expect{Y}=\sum_{i\in[s]}\Pr(\bar{e_i} \text{ samples the same color})\leq \gamma\Delta^2\cdot \frac{\delta^2}{\Delta+1}\leq \delta^2\cdot \gamma \cdot \Delta.
\]
Therefore, we can use Talagrand inequality from \Cref{prop:talagrand} as follows
\begin{align*}
    \Pr\paren{\abs{Y-\expect{Y}}\geq \frac{\expect{X}}{10}} 
    &\leq \exp\paren{-\frac{(\expect{X}-\Theta(1)\sqrt{\delta^2\gamma\cdot \Delta})^2}{\delta^2\gamma\cdot \Delta}} \tag{as $\expect{Y}\leq \delta^2\gamma\cdot \Delta$}\\
    & \leq \exp\paren{-{\delta^5\cdot\gamma\cdot \Delta}}< 2^{-\delta^5\cdot\gamma\cdot \Delta}. \tag{as $\expect{X}\geq   \delta^5\cdot\gamma\cdot \Delta$}
\end{align*}
Similarly, as $\expect{Z}\leq \expect{Y}$, we get 
\[
    \Pr\paren{\abs{Z-\expect{Z}}\geq \frac{\expect{X}}{10}} < 2^{-\delta^5\cdot\gamma\cdot \Delta}.
\]
Using the two inequalities above, we have that with probability at least $1-2^{-\delta^5\cdot\gamma\cdot \Delta}$ both $Y$ and $Z$ are concentrated and therefore
\[
    X=Y-Z \geq \expect{Y} - \frac{\expect{X}}{10} - (\expect{Z}+\frac{\expect{X}}{10}) = \expect{X} - \frac{\expect{X}}{5} = \frac{4}{5}\cdot \expect{X}. 
\]
Thus, with probability at least $1-2^{-\delta^5 \cdot\gamma\cdot \Delta}$ we have
\[
     \slack{v}{\chi} \geq  p(v)\geq X \geq  \frac{1}{5}\cdot  \delta^5\cdot\gamma\cdot \Delta \geq  \delta^6\cdot\gamma\cdot \Delta.\qed
\]

\end{proof}

\subsection{Proof of \Cref{clm:test-works}}

The following subroutine determines if a sparser vertex has \Cref{flaw:all} without checking all its neighbors.

\begin{Algorithm}\label{alg:sparse-test} 
$\test(v)$:
\begin{enumerate} 
    \item For each $i\in[\log n/\delta^2]$, sample an integer $k_i$ from $[\Delta]$ uniformly at random. 
    \item Let $R$ be the number of $k_i$ where $k_i\leq \deg(v)$ and $k_i$-th neighbor of $v$ is colored by $\colS$.
    \item If $R>3\log n/\delta$ return \Fail; otherwise return \Pass.
\end{enumerate} 
\end{Algorithm}

\begin{lemma*}
    For any sparser vertex $v\in S$, if $v$ has \Cref{flaw:all}, then $\test(v)$ will return \Fail with high probability. Conversely,  if $v$ is far from having \Cref{flaw:all}, then $\test(v)$ will output \Pass with high probability. Moreover, the running time of $\test(v)$ is $O(\log n)$ with certainty.
\end{lemma*}
\begin{proof}
    Assume $v$ has \Cref{flaw:all}, then $\Ncolor{v}{\colS}> 4\delta\cdot\Delta$, i.e. $v$ has at least $4\delta\cdot \Delta$ many colored neighbors under $\colS$.

    For each $i\in[\log n/\delta^2]$, let $X_i$ be the indicator random variable that is 1 iff the $k_i$-th vertex of $N(v)$ exists and is colored under $\colS$. Let $X=\sum_{i\in[\log n/\delta^2]}X_i$. Observe that $X$ is precisely the value of $R$ in $\test(v)$. Since each sample is chosen independently and uniformly at random, then
    \[
    \expect{X} = \sum_{i\in [\log n/\delta^2]}\Pr(X_i=1) =\sum_{i\in [\log n/\delta^2]}\frac{\Ncolor{v}{\colS}}{\Delta}\geq \frac{4\log n}{\delta},
    \]
    as $\Ncolor{v}{\colS}>4\delta\Delta$.
    Using Chernoff bound (\Cref{prop:chernoff}) we get 
    \begin{align*}
        \prob{X<\frac{3\log n}{\delta}}
        \leq \exp\paren{-\frac{\log n}{8\cdot \delta}}\leq n^{-1/(8\delta)}.
    \end{align*}
    Hence, given $\delta$ is a small constant in $(0, 1)$, with high probability, we have $X>3\log n/\delta$, and therefore $\test(v)$ returns \Fail.

    Now suppose that $v$ is far from having \Cref{flaw:all}. Then, $\Ncolor{v}{\colS}<2\delta\cdot \Delta$, which implies
    $\expect{X} \leq \frac{2\log n}{\delta}$. Then we have 
      \begin{align*}
        \prob{X>\frac{3\log n}{\delta}}
        &\leq \prob{X>\frac{3}{2}\cdot \expect{X}}\tag{as $\expect{X}\leq 2\log n/\delta$}\\
        &\leq \exp\paren{-\frac{\log n}{2\delta}}\leq n^{-1/(2\delta)}.\tag{by \Cref{prop:chernoff}}
    \end{align*}
    Therefore, with high probability, $X\le 3\log n/\delta$, and so $\test(v)$ returns \Pass in this case.

    Finally, the procedure samples and checks the colors of $O(\log n)$ vertices. Hence, its running time is $O(\log n)$ with certainty.
\end{proof}